\documentclass[11pt,letter]{article}
\usepackage{graphicx,amscd,amsmath,amssymb,amsfonts,verbatim}
\usepackage[a4paper, total={6.25 in, 10 in}]{geometry}
\usepackage{mathrsfs}
\usepackage{bbm}
\usepackage[english]{babel}
\usepackage{amsthm}
\usepackage{bbm}
\usepackage{amssymb}
\usepackage[all]{xy}
\usepackage{mathtools}
\usepackage{threeparttable}
\usepackage{algpseudocode}
\usepackage{algorithm}
\usepackage{authblk}
\usepackage{caption}
\usepackage{subfigure}
\usepackage{chngcntr}
\usepackage{setspace}
\usepackage{booktabs}
\usepackage{pgfplots}
\usepackage{adjustbox}

\usepackage{natbib}
\usepackage{authblk}
\usepackage{multirow}
\usepackage{enumitem}
\usepackage{fancyhdr}
\DeclareMathOperator*{\plim}{plim}
\usepackage{hyperref}
\hypersetup{
	colorlinks=true,
	citecolor = blue
}
\usepackage{lastpage}
\usepackage{pdflscape}
\graphicspath{}
\newcommand{\indep}{\perp \!\!\! \perp}

\newtheorem{theorem}{Theorem}
\newtheorem{lemma}{Lemma}

\newtheorem{proposition}{Proposition}
\newtheorem{corollary}{Corollary}
\newtheorem{definition}{Definition}

\newtheorem{assumption}{Assumption}

\DeclareMathAlphabet{\mathpzc}{OT1}{pzc}{m}{it}

\makeatletter

\begin{document}

	\title{Policy Learning With Observational Data~: The Case of Hepatitis C Treatment for HIV/HCV Co-Infected Patients\thanks{I am greatly indebted to my supervisor, Erin C. Strumpf, for her support in this project. I also thank David Benatia, Denzel Fiebig, Saraswata Chaudhuri, Victoria Zinde-Walsh, Umut Oguzoglu, Jason Abaluck, Olivia Yu, and all participants of the 2026 IAAE Meeting in Lisbon, 2026 NASM in Atlanta, 2026 ASHEcon Meeting in Minneapolis, 2026 SCSE Conference in Québec City, 2025 Annual CHESG Meeting in Montréal, and the 2025 Joint Asian and European Workshop on Econometrics and Health Economics in Singapore for comments and advice. Thanks also to Shouao Wang, Alexandra Looky, Mariam El Sheikh, and Marina Klein for the access to and general guidance with the Canadian Co-Infection Cohort dataset.}}
	
	\author{Raphaël Langevin\footnote{Yale School of Management, Yale University, New Haven, CT, USA. E-mail: raphael.langevin@yale.edu.}
	}

	\maketitle
	\thispagestyle{empty}

	\begin{abstract}
		Decision-makers frequently must choose a single action from a finite set of alternatives — for example, physicians selecting a treatment, investors choosing a portfolio risk level, or judges determining sentences. To improve outcomes, policymakers often issue policy rules or guidelines to inform such choices. In this paper, I show how to generally derive policy rules from observational data in a multi-action framework under relatively weak assumptions about the underlying structure of the heterogeneous sampled population. Conditional average treatment effects (CATEs) are consistently estimated via a weighted K-means algorithm, assuming the outcome model is correctly specified within each homogeneous subgroup. I then use standard decision trees to implement feasible policy rules, allowing for both perfect and imperfect adherence to treatment. The methodology is applied to treatment options for Hepatitis C (HCV) among patients co-infected with human immunodeficiency virus (HIV), a setting in which no uniform guideline exists for modern pharmaceutical therapies. The results identify a subgroup of patients with approximately an 80\% probability of spontaneous HCV clearance without treatment. Estimation results also show that reallocating treatments among treated patients could have reduced total treatment costs by CAN\$2.5–4.5 million while increasing aggregate health benefits relative to the status quo. These findings demonstrate that the proposed approach can generate improved, data-driven treatment guidelines for the management of HIV/HCV co-infected patients.\\~\\
		\textbf{JEL Codes:} C38; C44; D61, H51; I18.\\
		\textbf{Keywords:} Policy Rule, Treatment Allocation, Unobserved Heterogeneity, Clinical Practice Guidelines, K-means, Hepatitis C, Direct-Acting Antivirals, HIV/HCV Co-Infection.        
	\end{abstract}
	
	\newpage
	\section{Introduction}\label{sec:1}
	\setstretch{1.3}
	\pagenumbering{arabic}
	When choosing from a menu containing a finite number of alternatives, economic agents do not necessarily have all the information they would like. Entities that have more, or at least different, information may establish rules that aim to guide decisions, often in the hope of increasing welfare or reducing costs. For instance, in clinical settings, physicians make decisions about the treatment they will provide to a given patient in the context of uncertainty regarding the expected benefits or side effects.\footnote{For simplicity, I refer to physicians as the ``decision-makers'' throughout the paper, although treatment choice is often the product of a joint decision between the patient and the physician.} Third-party health organizations (referred to as the ``policymakers'' hereafter) produce clinical practice guidelines (CPGs) to guide and inform physicians’ treatment decisions.
	\par
	For example, the U.S. Centers for Disease Control and Prevention (CDC) produces periodic recommendations concerning the prevention and control of seasonal influenza with vaccines \citep{grohskopf_prevention_2024}, whereas the American College of Gastroenterology (ACG) produced in 2020 a clinical guideline on the management of irritable bowel syndrome \citep{lacy_acg_2021}. In Canada, the Canadian Association for the Study of the Liver (CASL) updated its guidelines on the management of chronic Hepatitis C (HepC) in 2018 by advocating ``against the use of any interferon-based treatment regimens" due to the availability of direct-acting antivirals for nearly all patients \citep{shah_management_2018}.
	\par
	In the clinical context, several challenges make creating high-quality CPGs difficult and can limit their widespread acceptance and use \citep{alpert_why_2010,abaluck_fixing_2020}. First, while the policymaker’s objective function may align with the decision-maker’s, principal-agent settings may also arise where the policymaker (principal) uses guidelines to align the decision-maker’s (agent’s) behavior with their own objectives. The extent to which the objective sought by the policymaker (e.g., increase expected welfare, minimize the risk of serious adverse events, cost control, etc.) is acceptable to the decision-maker can therefore play a role in whether and how guidelines are used. For instance, many patients with atrial fibrillation and high risk of stroke do not receive antithrombotic therapy with warfarin due to the overemphasis on potential bleeding side effects, which contradicts the general recommendation on the subject \citep{gross_factors_2003,abaluck_are_2025}.
	\par
	Throughout the paper, I assume that the agent and the principal both share the same objective of maximizing expected welfare up to certain constraints defined by the policymaker, where guidelines are one of the many tools policymakers can use to influence the decision-makers' choices. Although strict compliance with guideline recommendations could be enforced in practice, uncertainty in health outcomes and information asymmetry are likely to make such a regulation inefficient in several cases \citep{boone_discretion_2025,einav_limitations_2024}.
	\par
	Second is a set of concerns related to the quality of the underlying clinical evidence used to formulate CPGs \citep{graham_clinical_2011}. Evidence coming from meta-analyses and randomized controlled trials (RCTs) is often considered the most internally valid among all types of evidence \citep{chloros_has_2023}. However, it is well known that RCTs often lack external validity: results are often not generalizable to important subsets of the target population due to restrictive eligibility criteria when conducting the trials \citep{lawlor_commentary_2004, atkins_creating_2007}. While studies based on observational data may have the benefit of increased representativeness, concerns about biased treatment effect estimates are prominent. Although quasi-experimental designs that can address confounding are generally not well understood or accepted in clinical contexts, interest in the use of real-world data by industry, health care payers, and regulatory bodies has been increasing in recent years \citep{cadth_guidance_2023, inesss_state_2022, us_food_and_drug_administration_use_2023}.
	\par
	The policy learning (PL) literature has also considered policymakers using observational (\textit{offline}) data to create guidelines (\textit{policy rules}) to inform decision-makers’ choices from among a set of \textit{treatment options}. The use of observational data with rigorous research designs has also been advocated in the clinical literature in cases where typical RCTs cannot answer the research question, costs are prohibitive, or when it is impossible to randomize interventions due to legal or ethical concerns \citep{atkins_creating_2007,greenfield_making_2017,reeves_including_2019}. Although results from observational studies are often biased, several methods have been and are still being developed to reduce biases and better inform policymakers on the effectiveness of medical interventions in the real world \citep{wager_estimation_2018,basu_2sls_2018,chernozhukov_doubledebiased_2018,chernozhukov_fisher-schultz_2025}.
	\par
	In this paper, I derive feasible policy rules from non-randomized, observational data in a framework where multiple treatment options are available to the decision-maker. Specifically, I consider the choice of pharmaceutical treatment options for patients co-infected with HepC and HIV, a setting where at least 3 treatment options exist, and there are no clear CPGs in the Canadian context. I use weaker and more practical assumptions than the ones typically invoked in the corresponding PL literature, namely that treatment assignments are unconfounded conditional on the value of a latent state, which is used to form groups of homogeneous patients. I then estimate average treatment effects (ATEs) for each group and use these estimates to derive the corresponding policy rule. The resulting estimates are known as \textit{conditional average treatment effects} (CATEs), where the conditionality is with respect to the group membership. A feasible version of this rule uses decision trees to predict group membership with a small subset of covariates, which may then be leveraged to improve existing CPGs or build new CPGs.
	\par
	The first contribution of this paper is the development of a more practical and parsimonious approach than those typically employed for PL with observational data. This approach contrasts with the traditional PL literature, which usually assumes that treatment assignment is independent of potential outcomes for every treatment choice conditional on observed characteristics, an assumption known as \textit{unconfoundedness} (selection on observables; \cite{kitagawa_who_2018,kitagawa_equality-minded_2021, chernozhukov_semi-parametric_2019,mbakop_model_2021,athey_policy_2021, zhou_offline_2023}). However, unconfoundedness is often unrealistic in real-world problems where important confounders may remain unobserved due to practical or ethical considerations \citep{kallus_confounding-robust_2019}.
	\par
	More importantly, unconfoundedness does not tell applied researchers which variables should be conditioned on (and at what values) to overcome selection biases. The choice of conditioning variable can be a serious dilemma even in low-dimensional settings where the number of observed covariates is less than the number of observations, but still relatively large \citep{chernozhukov_fisher-schultz_2025}. If the use of machine learning (ML) techniques somewhat alleviates this concern by selecting subsets of $X$ that best predict the observed outcome, it remains unclear which technique offers the best performance in a given context.\footnote{Specifically, \cite{athey_policy_2021} argue that their approach, which relies on unconfoundedness, ``allows for the use of black-box machine learning tools provided we can verify that they are accurate in mean-squared error". When moving from a binary to a multi-action PL problem, \cite{zhou_offline_2023} state that ``[t]his leaves open a question of both theoretical and practical significance, namely which algorithm should be used [...] to achieve maximum statistical efficiency."}
	\par
	To address the latter problem, I replace unconfoundedness with a different, more practical assumption that I name ``group-wise" unconfoundedness, in which assignment to treatment is unconfounded conditional on the value of a latent state that is unobserved by both the econometrician and the decision-maker. In the context of a physician prescribing from among a choice of drugs or medical procedures, it is often reasonable to assume that treatment response heterogeneity will depend on a discrete and finite number of health states rather than on each possible value taken by the entire set of covariates. The health state can be a function of different observed (e.g., age, chronic conditions, income) and unobserved (e.g., frailty, social support) characteristics.
	\par
	Compared to standard unconfoundedness, group-wise unconfoundedness is more parsimonious, as it implies that the conditioning set goes from the space of all possible values taken by $X$ to a small discrete state space without introducing selection bias. This leads to CATE estimates that are less prone to overfitting and that are equally or more precise than the ones produced by standard ML techniques such as random forests, provided that the true number of groups is small relative to the sample size. The estimated CATEs will also be robust to the presence of unobserved confounders, as long as those confounders are constant within each group.
	\par
	Under group-wise unconfoundedness, consistent estimation of all CATEs is possible if the true group memberships are known for all observations in the sample. Because group memberships are unobserved and have to be estimated as well, consistency of the estimated CATEs usually relies on the consistency of the estimated group memberships. Clustering algorithms such as the K-means algorithm are a practical and computationally efficient way to estimate such group memberships. The second contribution of this paper relates directly to the literature on clustering algorithms; I show how a \textit{weighted} K-means algorithm can be used to consistently estimate all within-group parameters (including CATEs) under a set of relatively weak assumptions formulated at the group level.
	I then show how to use those estimated CATEs to build an optimal policy rule that lies within the class of all policy rules that are deemed feasible by the policymaker.\footnote{Since all policy rules are by definition optimal for a given class of feasible policy rules, the adjective ``optimal" becomes superfluous. This is why I omit it for the rest of the paper unless relevant.}
	\par
	One advantage of the proposed weighted K-means algorithm is that it uses all available covariates to recover a small number of groups where untreated observations act as counterfactuals for treated observations in the same group. Another advantage is that all estimated CATEs remain constant within each group if all outcome models are separable in $X$. In this case, predicting every CATE for a given patient reduces to predicting his individual group membership. This contrasts with other ML techniques where selecting a subset of $X$ may discard useful information -- hence leading to regularization bias \citep{chernozhukov_doubledebiased_2018} --, or where the large number of estimated CATEs makes it challenging to interpret the results.
	\par
	Estimating CATEs is a difficult task, and slight violations of critical assumptions, such as (group-wise) unconfoundedness, can lead to important biases in finite samples and asymptotically. To reduce the magnitude of such biases, I use the best linear prediction (BLP) of all CATEs estimated by the weighted K-means algorithm using the approach described by \cite{chernozhukov_fisher-schultz_2025}. In the empirical application, the BLP estimates are highly similar to the original CATEs produced by the weighted K-means algorithm, suggesting that the algorithm yields results that are not overly biased and can be confidently used for decision-making. Because they improve upon the original estimated CATEs in terms of mean squared error, I then use the median BLP values to build different policy rules and perform cost-effectiveness analyses.
	\par
	Consistently estimating group memberships via the weighted K-means algorithm typically requires a large number of covariates. This may lead to practical issues in clinical settings where decision-makers may not be able to gather all the information required by the algorithm to predict a patient's membership. To estimate policy rules that are actually practical for decision-makers, I consider the policy class of decision trees, which are known for their flexibility and interpretability. Decision trees are also widely used in healthcare to guide decision-making processes \citep{podgorelec_decision_2002}. More precisely, I use classification trees to predict the estimated group memberships based on a clinically relevant subset of $X$ that is observable by both the econometrician and the decision-maker. The resulting decision tree might then be used to complement or replace existing CPGs once its predictive performance has been empirically validated \citep{wei_zhao_construction_2020}. However, using classification trees for clinical decision-making entails a trade-off between feasibility and prediction accuracy: the simpler the tree, the greater the classification error and vice versa.
	\par
	If the distribution of the CATEs is similar between the sampled and the target populations, the policymaker can build a population-level policy rule that would assign treatment choices based on any patient's predicted group membership. In the context of the weighted K-means algorithm, a similar distribution of the CATEs between the sampled and target populations means that both populations are represented by the same set of latent groups with similar within-group ATEs for each treatment option. This also implies that the size of the sampled population is sufficiently large such that each within-group ATE is point-identified. Differences between the sampled and the target population do not reduce generalizability as long as both populations contain the same set of groups with sufficiently large group sizes in the sampled population. This represents another benefit of the suggested approach compared to ML techniques, where deviations from perfect representativeness may lead to poor external validity \citep{bay_machine_2024}.
	\par
	If there is no latent group heterogeneity in the target population, then the corresponding policy rule would be to assign the welfare-maximizing treatment option to everyone in this population (provided an unbiased estimate of the ATE for every option). If such a situation is verified, then it is unnecessary to rely on the weighted K-means algorithm to guide clinical decision-making. Nonetheless, treatment effect homogeneity across all treatment options is often unrealistic, and the BLP approach proposed by \cite{chernozhukov_fisher-schultz_2025} allows us to test for the presence of underlying heterogeneity in the sampled population. Applying this test to the empirical application confirms the usefulness of the weighted K-means algorithm to build policy rules that optimally allocate treatment options in both the sampled and the target populations. In the medical literature, such policy rules are often characterized as ``personalized medicine" and are increasingly used to improve both quality of care and patients' quality of life \citep{chakraborty_statistical_2013,einav_limitations_2024}.
	\par
	The third contribution of this paper pertains to the nature of the treatment variable. In most of the papers on PL, the treatment variable is a binary variable (i.e., whether the patient was treated or not) that is directly observed by the econometrician. In this paper, I extend this simple framework to the case of multiple treatment options (also known as a \textit{multi-action framework}) with a semi-continuous treatment variable within each option. The semi-continuous treatment variable represents the fact that patients do not necessarily fully adhere to the prescribed treatment.\footnote{Adherence to treatment corresponds to the proportion of days covered (PDC, see, for instance, \cite{prieto-merino_estimating_2021} for more details) and is defined as the ratio between the time during which the patient has taken the treatment and the prescribed duration of the treatment. Therefore, low adherence can result from either a lack of patient compliance with the prescribed regimen or a switch to another treatment regimen. The idea is that adherence should be directly correlated with the degree of exposure to the prescribed treatment.} Even if most patients in the sampled population fully adhere to the prescribed treatment, treating adherence (or \textit{exposure}) to treatment as a binary variable is similar to introducing a non-random measurement error in the independent variable of interest, thus producing an attrition bias in the estimated CATEs. Furthermore, treating adherence as a binary variable when it is not may lead to suboptimal treatment allocations if adherence varies across treatment options. This is why it is important to account for (potentially) imperfect adherence not only while estimating CATEs, but also when building policy rules. I describe how imperfect adherence is incorporated into the estimated policy rules in Section \ref{sec:33} of this paper.
	\par
	So far, most of the literature on PL has focused on the asymptotic and statistical properties of different estimators/algorithms whose objectives are to learn a policy rule $\pi^{*}$ within a predetermined class of feasible policies $\Pi$ under unconfoundedness or bounded violations of unconfoundedness \citep{kitagawa_who_2018, chernozhukov_semi-parametric_2019, athey_policy_2021, kallus_minimax-optimal_2021, mbakop_model_2021, zhou_offline_2023,ek_off-policy_2023, kitagawa_constrained_2023, sahoo_learning_2024, viviano_policy_2024,zhan_policy_2024,hess_efficient_2025}. To the best of my knowledge, most of these estimators/algorithms have not yet been used in any applied work except for those presented in the papers cited above. Therefore, this paper also serves to bridge the gap between theory and practice via a concrete example of how a policy rule can be derived using standard classification/regression procedures with observational data from a large-scale prospective cohort.
	\par
	The paper's fourth contribution relates to the application of the proposed methodology. According to \cite{saeed_frequent_2024}, 20\% to 30\% of Canadians infected with the human immunodeficiency virus (HIV) are also infected with the hepatitis C virus (HCV). This rate goes above 50\% among people who inject drugs (PWID). According to the CDC, about 21\% of people with HIV in the United States are also infected with HCV, while this rate lies between 62 and 80\% for PWID with HIV \citep{cdc_viral_2025}. Even though the evolution of HIV/HCV co-infection is similar in the two countries, the U.S. guidelines recommend prolonged treatments for co-infected patients (compared to HCV mono-infected ones), whereas Canadian guidelines make no distinction between treatments for HCV mono-infected and HIV/HCV co-infected patients \citep{martel-laferriere_prise_2022-1}. Moreover, potential drug-drug interactions between antiretroviral therapy (ART) and the recently approved direct-acting antiviral (DAA) agents used to treat HCV create an additional unknown risk that needs to be assessed before formulating recommendations or developing new guidelines \citep{american_association_for_the_study_of_liver_diseases_patients_2022}.
	\par
	Even though DAAs are very effective in the HCV mono-infected population, it is less clear in HIV/HCV co-infected patients since most co-infected patients were not eligible to participate in DAA trials \citep{saeed_how_2016}.\footnote{By effectiveness, I refer to the increase in the marginal probability of reaching SVR that is induced by completely adhering to a given DAA regimen, conditional on a relevant set of observed or unobserved individual characteristics.} Observational studies looking at the effectiveness of DAAs administered to co-infected persons obtain significantly lower \textit{sustained virological response} (SVR) rates compared to HCV mono-infected persons \citep{alam_real-world_2017, wilton_real-world_2020}. However, evidence tends to show that weaker responses to treatment are due to social and behavioral rather than biological factors, especially among PWID \citep{wilton_real-world_2020}. Reaching SVR is equivalent to being cured of the disease and implies an undetectable viral load in the blood 12 weeks after the end of treatment. Note also that, unlike Hepatitis B, there is no vaccine against HCV, which makes reinfection always possible for any previously cured patient.
	\par
	Taking on the policymaker's perspective, I use data from the Canadian Co-Infection Cohort (CCC) to learn a policy rule that can be implemented for routine decision-making processes. The objective of this policy rule is to maximize the probability of reaching SVR for all observations in the sample under a set of prespecified constraints. Four treatment options are considered, including the absence of treatment (the ``null" treatment regimen). Decision trees are then used to derive the feasible version of the policy rule. Such a feasible rule could then be included in CPGs to guide future decision-making processes by prescribing physicians, provided that the sampled and target populations share a similar treatment response distribution. This is believed to be the case here given the CCC cohort's very broad eligibility criteria, and given that it has enrolled around 23\% of all HIV/HCV co-infected patients engaged in care in Canada \citep{saeed_frequent_2024}.
	\par
	Assuming that the objective of the policymaker is to maximize social welfare, an ``optimal" policy rule from the perspective of the policymaker balances the cost of each treatment option with its expected health benefits. As suggested by \cite{kitagawa_who_2018}, the outcome $Y_i$ can be formulated such that it corresponds to the health benefits net of the costs for each treatment option. However, such an approach means that the outcome becomes a function of the policymaker's willingness-to-pay (WTP). Another approach, more practical and interpretable, is to adjust the unconstrained policy rule -- that is, a policy rule that does not consider any cost constraint -- so that the recommended treatment option maximizes welfare while satisfying the policymaker's WTP \citep{sverdrup_qini_2025}. Such constraints can be incorporated into the infeasible policy rule based on all available information, or into its feasible version based on the tree's predicted group memberships.\footnote{Another possibility is to impose a total budget constraint on the policy rule, but this type of constraint is not often employed in economic evaluation in healthcare and also leads to important computational challenges. See, for instance, \cite{bhattacharya_inferring_2012} and \cite{sverdrup_qini_2025} for more details.}
	\par
	To illustrate the general estimation strategy, Figure \ref{fig:A} draws a schematic representation of the different objects used in the empirical application to solve the problem of interest. First, I estimate all group memberships via the weighted K-means algorithm. I then use the estimated group memberships to identify within-group CATEs using standard linear probability models with observed SVR as the outcome. The selection of hyperparameter values and variables to include in the outcome models is assessed through 2-fold cross-validation with several different data splittings. Estimation and inference on the BLP of the CATEs is performed during the cross-validation procedure following the general logic of \cite{chernozhukov_fisher-schultz_2025}. I then use the median BLPs to derive both infeasible and feasible policy rules. Finally, I model adherence to treatment via a two-part model where predicted adherence is then incorporated into both policy rules to account for variations in adherence across patients and treatment options. Section \ref{sec:55} graphically shows all four combinations of policy rules as a function of the policymaker's WTP. The resulting policy rules can also be interpreted as four different cost-effectiveness analyses where the in-sample optimal treatment allocation is derived for various levels of the policymaker's WTP, conditional on the cost of each treatment option.
	\begin{figure}[!t]
		\centering
		\includegraphics[width=15cm]{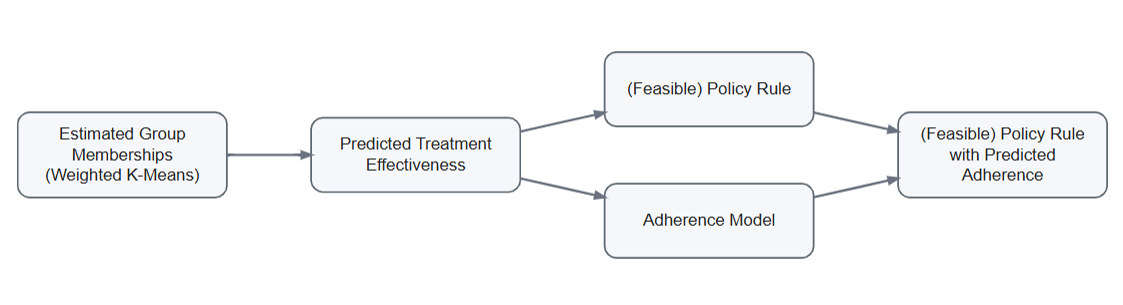}
		\caption{Schematic representation of the process used to learn multi-action policy rules with semicontinuous treatment.}
		\label{fig:A}
	\end{figure}
	\par
	Estimation results show that out-of-sample (OOS) prediction errors are minimized when the total number of groups in the empirical application is equal to four. This result is achieved using simple linear probability models within each group, where the outcome corresponds to whether a patient has reached SVR ten (10) weeks after the end of the prescribed treatment. All median BLPs are strictly positive, but important treatment effect heterogeneity is identified across groups. Specifically, one group identifies BLPs of the estimated CATEs that are not significantly different from zero for every DAA regimen, and where the probability of being ``spontaneously" cured of the HepC infection is very high ($\approxeq 80\%$).\footnote{Spontaneous clearance is associated with an acute HepC infection, and is characterized by an absence of detectable viral load in the patient's blood up to 6 months after the onset of the infection. On the other hand, patients with a chronic HepC infection do not usually experience spontaneous clearance within 6 months. More details on spontaneous clearance are presented in Section \ref{sec:42}} According to the simplest decision tree that distinguishes between all four groups, belonging to this group can be predicted with pre-treatment characteristics such as the degree of liver cirrhosis/fibrosis and the blood concentration of certain enzymes or hormones. Descriptive statistics show that the average patient within this first group is slightly younger and has a higher income compared to the average patients in the three other groups. Every other group feature median BLPs that are all significantly different from zero, with slight differences in the ordering of median CATEs across the three groups.
	\par
	In addition to the literature on PL that was cited earlier, this paper also relates to the literature on statistical decision theory in econometrics \citep{hirano_asymptotics_2009, hirano_asymptotic_2020, stoye_minimax_2009,stoye_minimax_2012, tetenov_statistical_2012, manski_statistical_2004, manski_econometrics_2021}, to the literature on clinical practices with HIV/HCV co-infected patients \citep{laniece_delaunay_gaps_2022,peters_retrospective-prospective_2021,roman_lopez_effectiveness_2024,saeed_eliminating_2020,saeed_how_2016,wilton_real-world_2020}, and also to the literature on individualized treatment rules in biostatistics and epidemiology \citep{chakraborty_statistical_2013, wallace_doubly-robust_2015, zhao_estimating_2012, zhao_new_2015, zhang_robust_2013, li_optimal_2023}. The first strand of literature focuses on the statistical properties of different decision criteria when the treatment effect depends on an unknown ``state of nature". In our context, the possible states of nature are analogous to the group memberships recovered by the K-means algorithm. More details on this topic are given in Section \ref{sec:22}.
	\par
	This paper contributes to the second strand of literature by comparing the effectiveness of three DAA regimens used with HIV/HCV co-infected patients: Mavyret, Epclusa, and Harvoni. My results show that the Epclusa treatment regimen is strictly dominated (in terms of median effect) by the two other DAA regimens when the policymaker's WTP is higher than CAN\$905 per additional percentage point of the probability of reaching SVR. However, the same results show that Epclusa is the minimax-regret treatment option, meaning that it minimizes welfare loss among DAA regimens when predicted group memberships do not correspond to the true groups. In brief, my results show that the optimal treatment choice among the considered options is not trivial and depends on a set of pre-treatment covariates, the policymaker's maximum WTP, and the chosen decision criterion. Finally, the last strand of literature mainly concerns \textit{dynamic treatment regimes} where decisions need to be made at different points in time. Such a setup goes beyond the scope of this paper, but I derive several insights from this literature.
	\par
	The remainder of the paper is structured as follows. Section \ref{sec:2} introduces the notation and the general framework of the problem of interest. Section \ref{sec:3} describes the assumptions and the proposed methodology used to solve this problem. Section \ref{sec:4} presents the characteristics of the cohort and the associated health conditions (HCV and HIV). Section \ref{sec:4} also details how various policy rules can be derived from this dataset when cost constraints are applied to feasible policy rules. Section \ref{sec:5} presents the results, while Section \ref{sec:6} briefly discusses the results and concludes.
	\section{Notation and General Framework}\label{sec:2}
	\subsection{Notation}\label{sec:21}
	Suppose we have access to a random sample of $N$ independent patients. Each patient is characterized by the quadruple $(Y_i,A_i,X_i,S_i)$, where $Y_i \in \mathcal{Y} \subset \mathbb{R}$ is the $i^{th}$ patient's observed outcome, $A_i = (A_{i,0},...,A_{i,D_i},...,A_{i,\bar{D}}) \in [0,1]^{\bar{D}+1}$ is the set of all adherence to treatment with $D_i \in \mathcal{D} = \{0,1,...,\bar{D}\}$ corresponding to the observationally assigned treatment for patient $i$, $X_i = (X_{i,1},...,X_{i,p})^{\top} \in \mathcal{X} \subseteq \mathbb{R}^{p} $ is a vector of pre-treatment covariates of size $p$, and where $S_i \in \mathcal{S} = \{1,2,...,\bar{S}\}$ with $\bar{S}<\infty$ is the value of the unobserved, discrete latent health state (or group).\footnote{In the literature on statistical decision theory, the state value represents a ``state of nature" and is assumed to be constant across units. I depart from this literature by treating the state variable as a latent state variable that varies across units. In this case, the value taken by the state variable $S_i$ is similar to a sampling indicator that determines the joint distribution from which $(Y_i,A_i,X_i)$ is drawn for any unit $i \in [N] = \{1,2,...,N\}$, as described in Section 2.2 of \cite{manski_econometrics_2021}.}
	\par
	Throughout the paper, I assume that the econometrician observes only the triple $(Y_i,A_i,X_i)$ for all $i \in [N]$, which includes the assigned treatment choice $D_i$. The absence of treatment is represented by the number ``0", which automatically implies $A_{i,0}=0$ if no treatment was ever given to patient $i$. For now, I consider the setup where $A_{i,d} = 0$ for all $d \in \mathcal{D}\backslash D_i$, which implies that every patient in the sample can only be treated once and cannot switch between treatments. This restriction can be relaxed if $X_i$ includes information about past treatment history and if the subscript $i$ refers to patient-period rather than patient alone. Cases where the adherence is always either null or perfect can be easily obtained by imposing the constraint $\sum_{d=0}^{\bar{D}} A_{i,d} = \{0,1\}$ on this general setup.
	\par
	Similar to the classical potential outcome notation, I define the set of all potential outcomes as $\{Y_i(d,a)\}_{d=0,a\in[0,1]}^{\bar{D}}$, where $Y_i(d,a)$ corresponds to the value of $Y_i$ we would have observed had the adherence to treatment option $d$ been equal to $a$. Consequently, we have that $Y_i = Y_i(D_i, A_{i,D_i})$ holds only for the observed outcome. I also define the potential adherence $A_i(d)$, which represents the adherence value we would have observed had patient $i$ been assigned to treatment option $d$, and where $A_i(0) = 0$ for all $i \in [N]$. Based on these two definitions, I posit that the set containing all covariates, latent state value and potential outcomes, $(X_i,S_i,Y_i(0,0),...,Y_i(\bar{D},A_i(\bar{D}))) \in \mathcal{X} \times \mathcal{S} \times \mathcal{Y}^{\bar{D}+1}$, is drawn i.i.d. from a fixed data-generating distribution $\mathcal{P}$, where the set containing all potential adherence values and the observed treatment choice is drawn i.i.d. from another underlying distribution $\mathcal{A}$ that can depend on the realization of $\mathcal{P}$. To facilitate the exposition of the general framework, I set the adherence to an arbitrary level $a \in [0,1]$ for all $i \in [N]$ for the remainder of the current section. I relax this restriction in Section \ref{sec:33}, where I estimate $A_i(d)$ for all pairs $(i,d) \in [N] \times \{1,...,\bar{D}\}$ by treating the observed vector $A_i$ as a function of treatment choice and a set of individual and environmental features.
	\par
	\subsection{Policy Rule and Treatment Effect Heterogeneity}\label{sec:22}
	A policy rule $\pi: \mathcal{X} \to \mathcal{D}$ is a mapping from the space of covariates $\mathcal{X}$ to a decision $d \in \mathcal{D}$, where $\pi(X_i) = d$ implies that the policy rule $\pi$, when applied to a patient with covariates $X_i$, selects the $d^{th}$ treatment option as the ``optimal" option. Different notions of optimality are detailed below. For simplicity, this setup excludes probabilistic assignments where the resulting policy rule would randomize treatment assignments for a given value of $X_i$.
	\par
	If we assume a utilitarian social welfare function that is monotonically increasing in $Y_i$, the corresponding (population-level) policy rule $\pi^*(a)$, for a given adherence level $a$, can be defined as follows
	\begin{align}\label{eq:01}
		\pi^{*}(a):= \arg \max_{\pi \in \Pi} \mathbb{E}[Y_i(\pi, a)],
	\end{align}
	where the expectation is taken with respect to the joint distribution $\mathcal{P}$ for a given adherence level $a$, and where $\pi \equiv \pi(X_i)$ to alleviate notation. Eq. (\ref{eq:01}) implicitly assumes that all relevant constraints are embedded into the $\Pi$ class, and that the objective of all agents are aligned with respect to the (constrained) maximization of $Y_i$ for any patient $i \in [N]$. Note also that, since $Y_i$ is a scalar, there is no need to account for the patient's preferences in the above maximization program. Extension of the proposed methodology to the case of multivariate outcome goes beyond the scope of this paper.
	\par
	A central feature of PL is treatment effect heterogeneity (TEH). Suppose each treatment option features the same treatment effect across all patients in the population of interest. According to Eq. (\ref{eq:01}), the corresponding policy rule $\pi^{*}(a)$ would then be to assign the treatment option $d^* = \arg \max_{d \in \mathcal{D}} \mathbb{E}[Y_i(d,a)]$ to everyone in the population regardless of the values taken by $X_i$ and $S_i$, provided such a simple rule lies within $\Pi$. Introducing TEH in the population of interest is not only more realistic, but also allows the policymaker to design more or less complex policy rules that will ultimately lead to higher welfare in the population under a given set of constraints.
	\par
	Given the notation outlined in the previous section, CATEs are used to represent TEH and are defined as follows
	\begin{align}\label{eq:6}
		\tau_{d}(a,x,s) := \mathbb{E}[Y_i(d,a)|X_i = x, S_i = s] - \mathbb{E}[Y_i(0,0)|X_i = x, S_i = s],
	\end{align}
	where $\tau_{d}(a,x,s)$ corresponds to the average treatment effect of treatment option $d$ under adherence level $a$ conditional on the value taken by $(X_i, S_i)$. An important feature of Eq. (\ref{eq:6}) is that it allows for both observed and unobserved heterogeneity in treatment effects through the $x$ and $s$ arguments, respectively.
	\par
	One of the challenges in estimating CATEs is that it requires very large sample sizes to obtain precise estimates for each relevant partition of $X_i$ and $S_i$, especially when some of the covariates in $\mathcal{X}$ are continuous \citep{stoye_minimax_2009}. Obtaining precise values for each estimated CATE generally implies that a large number of observations in the sample share identical or similar values of covariates for each treatment option in $\mathcal{D}$, which is often impossible due to the presence of several continuous variables.
	\par
	One way to deal with this issue is to use non-parametric, kernel-based methods or machine learning techniques such that observations with sufficiently close values of $X_i$ will be used to estimate each CATE (also known as \textit{adaptive neighbour estimator}). Although forest-based approaches have been shown to be well-suited for such a problem \citep{dandl_what_2024}, there is no guarantee that it will perform optimally uniformly across every joint distribution $\mathcal{P}$ in finite samples \citep{cattaneo_accuracy_2026}.
	\par
	Another way to deal with this issue is to assume that the state variable $S_i$ captures all TEH, and that the joint distribution of the observed covariates can be used to recover $S_i$ for every patient in the target population. This assumption is formalized in Assumption \ref{ass:1}.
	\begin{assumption}(Group-wise heterogeneity)\label{ass:1}\\
		Let $\tau_{d}(a,x,s)$ be defined as in Eq. (\ref{eq:6}). Then,  
		\begin{align*}
			\tau_{d}(a,x,s) &:=\mathbb{E}[Y_i(d,a)|X_i = x,S_i=s] - \mathbb{E}[Y_i(0,0)|X_i = x,S_i=s],\\
			&\equiv \mathbb{E}[Y_i(d,a)|S_i=s] - \mathbb{E}[Y_i(0,0)|S_i=s] =: \tau_{d}(a,s),
		\end{align*}
		for any pair $(d,a) \in \mathcal{D}\times [0,1]$, with $\Pr[S_i = s|X_i = x] \ne \Pr[S_i = s]$ for any $x \in \mathcal{X}$ and any $s \in\mathcal{S}$.
	\end{assumption}
	Even though Assumption \ref{ass:1} may not be perfectly satisfied in practice, it can provide approximations of the true underlying TEH that may still lead to optimal decision-making processes. For instance, if the covariate space $\mathcal{X}$ is large and $\mathbb{E}[Y_i(d,a)|X_i = x,S_i=s] \ne \mathbb{E}[Y_i(d,a)|X_i = w,S_i=s]$ for any $x \ne w$, then it is very difficult to obtain unbiased and/or precise estimates of $\tau_{d}(a,x,s)$ for every possible value of $x$ given $d$, $s$, and $a$. However, if the distribution of $\tau_{d}(a,x,s)$ features $\bar{S}$ well-separated modes across $\mathcal{X}$ (with $S << N$), each mode can be treated as a group with a distinct label in $\mathcal{S}$. Such an approximation entails a bias-variance trade-off in the estimated CATEs, where the size of the bias will ultimately depend on the general shape of the $\tau_{d}(a,x,s)$ distribution and on the availability of a consistent classification algorithm under appropriate assumptions (see Section \ref{sec:31}). Given the relatively small sample size of the empirical dataset, using this approximation is not inappropriate and yields better interpretability compared to results obtained under the standard definition of TEH (i.e., conditional on $X_i$ only). A nonlinear version of Assumption \ref{ass:1} could also be formulated at the level of the linear predictor when using a nonlinear parametric specification.
	\par
	Finally, note that if $\Pr[S_i = s|X_i = x] = \Pr[S_i = g|X_i = x]$, then distinguishing between the two groups is practically infeasible unless additional information is employed. Even though it is possible to used observed outcomes to identify unobserved heterogeneity when covariates cannot, observed outcomes cannot be used to predict group membership (or any form of TEH) in new observations where $Y_i$ is yet unrealized. If observed outcomes are used to identify in-sample TEH, these have to be taken out of the information set when building the policy rule $\pi^*$. Such policy rules may therefore lead to overfitting and poor out-of-sample performance compared to those built without relying on in-sample observed outcomes.\footnote{This is namely the case when the outcomes observed in the training sample are used to form \textit{doubly robust scores} of in-sample CATEs \citep{robins_estimation_1994}, but which cannot be directly used to predict TEH with new observations where the outcome is not yet realized. See Section \ref{sec:23} and the end of Section \ref{sec:52} for more details.}
	\subsection{Decision Criteria}\label{sec:23}
	If Assumption \ref{ass:1} is verified, it is then possible to learn policy rule that will account for the fact that TEH is driven by a latent state variable rather than driven by every possible value taken by $X_i$. If the value of the state variable is known with certainty, the corresponding policy rule for true state $S_i$ can be written as follows
	\begin{align}\label{eq:1}
		\pi^{*}_C(a) &:= \arg \max_{\pi \in \Pi} \sum_{s \in \mathcal{S}} \mathbbm{1}[S_i=s]\mathbb{E}[Y_i(\pi, a)|S_i=s],
	\end{align}
	where $\mathbbm{1}[\cdot]$ denotes the indicator function, the $C$ subscript stands for ``certainty", and where the resulting policy rule is a set of $\bar{S}$ (simple) policy rules with each one depending on $S_i$ rather than $X_i$.
	\par
	Because $S_i$ is unobserved, the certainty criterion can lead to suboptimal treatment allocation since it implicitly assumes that some ``oracle" estimate $\hat{S}_i$ is available for all $i \in [N]$. To account for such classification uncertainty, it is possible to adapt the above policy rule using the three main decision criteria used in statistical decision theory. The first decision criterion is the \textit{Bayesian} criterion and leads to the following policy rule
	\begin{align}\label{eq:2}
		\pi^{*}_{B}(a) &:= \arg \max_{\pi \in \Pi} \sum_{s \in \mathcal{S}} \Pr[S_i=s]\mathbb{E}[Y_i(\pi,a)|S_i = s],
	\end{align}
	where $\Pr[S_i=s]$ is an unconditional subjective prior distribution over $S_i$. Note that the unconditional prior $\Pr[S_i=s]$ is used here instead of the conditional one $\Pr[S_i=s|X_i]$ since the effect of $X_i$ on $\Pr[S_i=s]$ is averaged out at the population level.
	\par
	It might also be convenient for the policymaker to build a policy rule that yields good performance over all possible states for every patient. This leads to two additional decision criteria~: the \textit{maximin} criterion and the \textit{minimax regret} criterion. The former criterion maximizes the minimum welfare attainable across all states, whereas the latter minimizes the maximum welfare loss incurred by the chosen policy across all states. In the context of imperfect adherence to treatment, each of these two policy rules can be defined as follows
	\begin{align}
		\pi^{*}_{M}(a) &:= \arg \max_{\pi \in \Pi}  \min_{s\in \mathcal{S}} \mathbb{E}[Y_i(\pi,a)|S_i = s], \label{eq:4} \\
		\pi^{*}_{MR}(a) &:= \arg \min_{\pi \in \Pi} \max_{s\in \mathcal{S}} R(\pi,s,a),
		\intertext{where}\label{eq:5}
		R(\pi,s,a) &:= \mathbb{E}[Y_i(\pi^*(s,a),a)|S_i=s] - \mathbb{E}[Y_i(\pi,a)|S_i=s],
	\end{align}
	and where $\pi^*(s,a) = \arg \max_{\pi \in \Pi} \mathbb{E}[Y_i(\pi,a)|S_i=s]$ is the policy rule that maximizes the expected potential outcome value in state $s$ given adherence level $a$. Note that the definition of the regret used in this paper differs from the one employed in the statistical decision literature due to imperfect adherence to treatment, while making it clear that the maximum regret $\max_{s\in \mathcal{S}} R(\pi,s,a)$ corresponds to the maximum difference between two expected potential outcome values.
	\par
	The policy rules $\pi^{*}_{M}(a)$ and $\pi^{*}_{MR}(a)$ will generally differ from each other, except in special cases where the maximum welfare is invariant across all states \citep{manski_econometrics_2021}. Theoretically speaking, choosing a specific criterion over another is still an open question. The statistical properties of the minimax-regret criterion have been widely studied in econometrics since it is argued to yield less pessimistic results than the maximin criterion, especially when applied to results from randomized experiments \citep{manski_adaptive_2007, manski_econometrics_2021,stoye_minimax_2009,stoye_minimax_2012,tetenov_statistical_2012, hu_minimax-regret_2024}. On the other hand, the Bayesian criterion is easy to interpret and implement, provided that one has access to a prior distribution over the different states.
	\par
	Other decision criteria that account for estimation risk have also been recently proposed in the literature \citep{chernozhukov_policy_2026}. However, such policy rules depend on the degree of risk aversion of the decision-maker, and are useful especially when larger estimated marginal treatment effects are associated with larger estimation variances and vice-versa. This situation does not generally correspond to what is observed in the empirical application.\footnote{This can be caused by the violation of the group-wise unconfoundedness assumption formally introduced in the next section.}
	\par
	In the literature on PL, several authors have tried to identify at what rate the regret of a policy rule $\pi \in \Pi$ converges to zero for a given class $\Pi$ and no latent state heterogeneity \citep{swaminathan_counterfactual_2015, kitagawa_who_2018, athey_policy_2021, zhou_offline_2023}. To do so, most have relied on doubly-robust (DR) scores under the assumption that either the outcome model or the treatment propensity score model is correctly specified and converges to the truth at a sufficiently fast rate. However, when $\mathcal{X}$ is high-dimensional and propensity scores are unknown, both models may be inconsistent, and the regret of the estimated policy rule may not converge to zero in practice unless relatively strong assumptions are satisfied for either model.\footnote{Such assumptions namely include sparsity or super-smoothness of either the outcome or the propensity score model when $\mathcal{X}$ is high-dimensional. See Section 2.2 of \cite{chernozhukov_fisher-schultz_2025} for more details.}
	\par
	Furthermore, it is unclear whether the use of doubly robust scores always improves welfare in finite samples compared to policy rules based solely on estimated CATEs. In the empirical analysis, using doubly robust scores with CATEs obtained from either the weighted K-means algorithm or generalized random forests does not appear to improve welfare under mild assumptions about counterfactual outcomes. For instance, all estimated doubly robust scores indicate that any patient who spontaneously clears from the HepC infection would have maximized their welfare had they been treated with the treatment option that is, on average, the least effective, or simply not treated at all. Although this is plausible for a fraction of patients in the sampled population, the probability that this is true for all patients with spontaneous clearance is near zero. The suggested approach to learning policy rules, which directly uses estimated CATEs to build policy rules, does not generate such unrealistic counterfactual scenarios.
	\section{Methods}\label{sec:3}
	\subsection{CATE Estimation}\label{sec:31}
	Even if treatment choices are randomized, imperfect adherence to the assigned treatment can also lead to biased estimates \citep{bowden_connecting_2021, bowden_instrumental_2024}. When using real-world data with potentially imperfect adherence to treatment, it is therefore convenient to assume that both assigned treatment and observed adherence are independent of all potential outcomes conditional on observables. This assumption, which is commonly known as \textit{unconfoundedness} (or selection-on-observables), is formalized below in Assumption \ref{ass:2}.
	\begin{assumption}\label{ass:2}(Unconfoundedness)
		\begin{align*}
			\{Y_i(d,a)\}_{d=0}^{\bar{D}} &\indep A_i|X_i=x, a \in [0,1].
		\end{align*}
	\end{assumption}
	Assumption \ref{ass:2} differs from the typical unconfoundedness assumption with perfect adherence to treatment since it requires both the observed treatment assignment \textit{and} adherence to be independent of the set of all potential outcomes. In this context, unconfoundedness is violated if patients are able to adjust their adherence level to reach a desired outcome value. Note also that standard unconfoundedness does not condition on any specific state value; it has to hold only for the true state value $S_i$ under the implicit assumption that this value is exogenously determined.\footnote{This assumption might however be violated if the state variable becomes a function of the assigned treatment over its duration. Because the treatment options I consider in the empirical application are relatively short (up to three months), I rule out this possibility for all observations in the sample.} Instead of relying on standard unconfoundedness, I combine Assumption \ref{ass:1} with a group-wise uncounfoundedness assumption that is formulated below.
	\begin{assumption}\label{ass:3}(Group-wise unconfoundedness)
		\begin{align*}
			\{Y_i(d,a)\}_{d=0}^{\bar{D}} &\indep A_i|S_i=s, a \in [0,1].
		\end{align*}
	\end{assumption}
	Assumption \ref{ass:3} states that the observed treatment choice and adherence are exogenous within each group formed by the true value of the state variable. This assumption complements Assumption \ref{ass:1} and is necessary (but not sufficient) to obtain consistent estimates of all within-group CATEs. Assumption \ref{ass:3} corresponds to a special case of Assumption \ref{ass:2}~: the two assumptions become identical if $S = |\mathcal{X}|$, where $|A|$ represents the size of the set $A$. If at least one element in $\mathcal{X}$ is continuous, then $|\mathcal{X}|$ is infinite and local approximations of Assumption \ref{ass:2} (under appropriate smoothness conditions) are required to consistently estimate all CATEs in finite samples. Therefore, Assumption \ref{ass:3} is more practical than Assumption \ref{ass:2} without necessarily being stronger than the latter since the total number of groups $\bar{S}$ can be adjusted such that conditional exogeneity is satisfied in practice.
	\par
	Note that both Assumptions \ref{ass:2} and \ref{ass:3} will be violated under the same general circumstances, which are that treatment selection and/or adherence are driven by the patient's idiosyncratic response to treatment after conditioning on $X_i$ and/or $S_i$. This will occur, for instance, when patient-specific response to treatment is correctly anticipated/predicted by the patient or the treating physician, hence leading to decisions that are adjusted accordingly by either agent. This implies that Assumption \ref{ass:2} (Assumption \ref{ass:3}) is equivalent to assuming that~: $a)$ patient-specific responses to treatment can be unbiasedly predicted by observed characteristics (group memberships) -- therefore making treatment assignment conditionally exogenous -- and $b)$ the patients/physicians are \textit{unable} to form unbiased anticipations (or to act upon those) about the effect of completing the assigned treatment regimen on the patient's welfare for a given group. Nonetheless, Assumption \ref{ass:3} does not totally preclude patients/physicians from forming (unbiased) anticipations and adjusting adherence accordingly, as long as the within-group observed variation in adherence remains uncorrelated with such anticipations. 
	\par
	In practice, Assumption \ref{ass:3} is unlikely to hold perfectly. For instance, patients who do not positively respond to a given treatment may be more likely to exhibit side effects over the course of the treatment. If these side effects are independent of group membership and generally induce the treating physician to modify the current treatment regimen (or the patient to stop compliance), then Assumption \ref{ass:3} is violated. On the other hand, patients who positively respond to a given treatment may be less likely to suffer from side effects (regardless of their true group), and therefore comply more with the assigned treatment. This results in a positively biased estimated adherence coefficient even if the initially assigned treatment $D_i$ is conditionally exogenous \citep{lamiraud_therapeutic_2007}. The degree of violation of Assumption \ref{ass:3} is context-dependent and will drive the size of the bias of the estimated CATEs.
	\par
	Before introducing the remaining assumptions, it is important to clearly define the estimands of interest.
	\begin{definition}\label{def:0}
		\begin{enumerate}[label=(\roman*)] Suppose Assumption \ref{ass:1} is verified.
			\item The conditional expectation of the potential outcome $Y_i(d,a)$ for given values of $X_i$ and $S_i$ is then defined as follows
			\begin{align*}
				m_{d}(a,x,s) := \mathbb{E}[Y_i(d,a)|X_i=x,S_i=s].
			\end{align*}
			\item The estimated value of $m_{d}(a,x,s)$ for given values of $X_i$ and $\hat{S}_i$ is then defined as follows
			\begin{align*}
				\hat{m}_{d}(a,x,s,\hat{S}_i) := \mathbb{E}[\hat{Y}_i(d,a)|X_i=x,\hat{S}_i=s],
			\end{align*}
			where $\hat{S}_i$ is an estimate of $S_i$, and where $\hat{Y}_i(d,a)$ is an estimate of the potential outcome $Y_i(d,a)$ for treatment option $d$ and adherence level $a$.
			\item The estimated CATE for a given value of $\hat{S}_i$ is then defined as follows
			\begin{align*}
				\hat{\tau}_{d}(a,s,\hat{S}_i):= \hat{m}_{d}(a,x,s,\hat{S}_i) - \hat{m}_{0}(0,x,s,\hat{S}_i),
			\end{align*}
			where $\hat{\tau}_{d}(a,s,\hat{S}_i)$ is independent of $X_i$ by Assumption \ref{ass:1}.
		\end{enumerate}
		\end{definition}
	The notation used in Definitions \ref{def:0} $(ii)-(iii)$ makes it explicit that the estimated CATEs are conditional on the value of the estimated classifier $\hat{S}_i$, which depends upon the chosen estimator. Since $S_i$ is unobserved, it is not possible to directly estimate $m_{d}(a,x,s)$ for a given value of $S_i \in \mathcal{S}$ unless $\hat{S}_i$ is estimated in the first stage for some or all $i \in [N]$. Although several ways exist to obtain each $\hat{S}_i$, this paper focuses on the one provided by the weighted K-means algorithm (see Algorithm \ref{alg:Kmeans} in Appendix \ref{app:B}). The rest of this section shows that, under Assumptions \ref{ass:1} and \ref{ass:3} (and the ones shown below), it is possible to consistently estimate all CATEs using the weighted K-means algorithm.
	\begin{assumption}(Bounded outcomes)\label{ass:4}\\
		The vector of potential outcomes $\{Y_i(d,a)\}_{d=0}^{\bar{D}}$ is supported on a bounded set in $\mathbb{R}^{\bar{D}+1}$ for any value of $a \in [0,1]$.
	\end{assumption}
	\begin{assumption}(Propensity score overlap)\label{ass:5}\\
		There exists some $\epsilon > 0$ such that $\Pr[A_{i,d} > 0 |S_i = s] \ge \epsilon, \ \forall (d,s) \in \mathcal{D}\times\mathcal{S}$.
	\end{assumption}
	\begin{assumption}(Error bounds on the outcome models)\label{ass:6}
		\begin{align*}
			\hat{S}_i = (S_i) \ \forall  i \in [N] \Rightarrow \hat{m}_{d}(a,x,s,\hat{S}_i) - m_{d}(a,x,s) = O_p(1/\sqrt{N^0_s}) \ \text{as $N \to \infty$},
		\end{align*}
		for any $(s,d,a,x) \in \mathcal{S}\times \mathcal{D} \times [0,1] \times \mathcal{X}$, where $\hat{S}_i$, $m_{d}(a,x,s)$, and $\hat{m}_{d}(a,x,s,\hat{S}_i)$ are defined as in Definition \ref{def:0}, $(\cdot): \mathcal{S} \to \mathcal{S}$ is a suitable permutation in the group labels, and where $N^0_s = \sum_{i\in [N]} \mathbbm{1}[(S_i) = s]$.
	\end{assumption}
	\begin{assumption}(Finite conditional moments)\label{ass:7}\\
		$\mathbb{E}[X_{i}] = \boldsymbol{\mu}^0_s \equiv (\mu^0_{s,1},...,\mu^0_{s,p})^{\top} \in \mathbb{R}^p$ and $Var[X_{i}] = \Sigma^0_s \in \Omega \Leftrightarrow s = S_i$, where $||\boldsymbol{\mu}^0_s||^2 < \infty$ for any $s \in\mathcal{S}$, where $\Sigma^0_s$ is a $p \times p$ positive-definite matrix with diagonal elements $0 < \sigma^2_{s,ll} < \infty$, $\Omega$ is the space of all symmetric, positive-definite matrices, and where $Cov[X_{i},X_{j}] = 0$ for any $j \ne i$ and $\mathbb{E}[(X_{il})^4] < \infty$ for any $l \in \{1,2,...,p\}$.
	\end{assumption}
	\begin{assumption}(Nonvanishing differences between groups)\label{ass:8}\\
		$U_s \ne U_g$ for any $g \ne s$, and for a nonvanishing proportion of the elements in each matrix as $p \to \infty$, where $U_s$ is a lower triangular matrix such that $\{\Sigma^0_s\}^{-1} = U_sU_s^\top$ for any $s\in \mathcal{S}$.
	\end{assumption}
	\begin{assumption}(Unbounded asymptotic group sizes)\label{ass:9}
		\begin{align*}
			\lim_{N\to \infty} \sum_{i\in [N]} \mathbbm{1}[S_i = s] = \infty, \ \forall s \in \mathcal{S},
		\end{align*}
		where the total number of groups $\bar{S}$ is fixed and known unless stated otherwise.
	\end{assumption}
	Assumption \ref{ass:4} is standard in the PL literature and is necessary to identify the CATEs via the observed outcomes \citep{mbakop_model_2021, zhou_offline_2023}. Assumption \ref{ass:5} states that all treatment options must feature a non-zero probability of being assigned to anyone in the sampled population given the true value of the state variable. In other words, no formal contraindication must prevent the assignment of any treatment option to all observations in a given group. If this assumption is violated, then it would not be possible to identify all CATEs in a given group. Note that it is weaker than the typical overlap assumption since it is conditional on the state value rather than on the value of the observed covariates, which thus allows for the presence of formal contraindications at the patient level that do not translate at the group level.
	\par
	Assumption \ref{ass:6} states that each estimated within-group outcome model is $\sqrt{N^0_s}$-consistent if every observation is correctly classified as the sample size tends to infinity, where $N^0_s$ corresponds to the true number of observations in group $s$. This assumption is weaker than the typical $\sqrt{N}$-consistency assumption made in the semiparametric estimation literature due to its conditional nature. It also remains agnostic about the impact of misclassification on both the bias and convergence rate of every within-group outcome model. To make things more concrete, assume that the observed outcome $Y_i$, treatment assignment $D_i$, and adherence to the assigned treatment $A_{i,D_i}$ are correctly described by the following linear model
	\begin{align}\label{eq:linear}
		Y_{i} = \beta_{S_i} + A_{i,D_i} \times \tau_{D_i}(1,S_i) + \epsilon_i&, \ \ \ \mathbb{E}[\epsilon_i|A_{i,D_i},D_i,S_i] = 0,\\ \label{eq:linear1}
		D_i = \arg\max_{d \in \mathcal{D}} \ \gamma_{S_i,d} + \nu_{i,d} , &\ \ \ \mathbb{E}[\nu_i|S_i] = 0,\\ \label{eq:linear2}
		A_{i,D_i} = \xi_{S_i,D_i} + \eta_i, \ \ \ &\mathbb{E}[\eta_i|S_i] = 0,
	\end{align}
	and $A_{i,d} =0$ if $d \ne D_i$, where $(\beta_{S_i}, \gamma_{S_i,0},...,\gamma_{S_i,\bar{D}},\xi_{S_i,1},...,\xi_{S_i,\bar{D}})$ corresponds to a set of finite group- and/or treatment-specific intercepts, the true TE $\tau_{d}(a,s) = a \tau_{d}(1,s)$ is linear in adherence, and where $\epsilon_i$, $\eta_i$, and $\nu_{i,d}$ are all independent error terms with finite variance. In this context, Assumptions \ref{ass:1}, \ref{ass:3}, \ref{ass:4}, and \ref{ass:6} are all satisfied if $\hat{S}_i = (S_i)$ for all $i \in [N]$ given that both treatment assignment and adherence are exogenous conditional on $S_i$.
	\par
	If it is believed that the untreated potential outcomes, treatment choices, or adherence levels depend also on $X_i$, the linear model shown above can be modified such that each equation incorporates an additive ``baseline" function $B_e(X_i,S_i)$, with $e = {1,2,3}$ representing either the outcome, treatment choice, or adherence equation.\footnote{In the medical literature, the variables entering the baseline function of the untreated potential outcome equation are also known as a ``prognostic" factors \citep{seibold_individual_2018,dandl_what_2024}.} However, if $X_i$ is high-dimensional and the true baseline functions are unknown, using machine learning methods to obtain these functions may induce regularization bias. To reduce the magnitude of such bias in finite samples, one can use either Neyman orthogonalization \citep{chernozhukov_doubledebiased_2018} or the BLP approach (see Section \ref{sec:32} for more details on the latter).
	\par
	Assumptions \ref{ass:7} and \ref{ass:8} are similar to Assumptions 2$(iv)$ and 2$(v)$ of \cite{langevin_bias-reduced_2026}, respectively, and correspond to sufficient requirements on the groups' structure to guarantee the consistency of the weighted K-means procedure when $\bar{S}$ is known. Consistency of a classifier $\hat{S}_i$ implies that $\lim_{N \to \infty }\text{Pr}[\hat{S}_i \ne (S_i)] = 0$ for all $i \in [N]$, where $(S_i)$ is a suitable permutation in the true groups' labels for all $i \in [N]$ \citep{su_identifying_2016}. Assumption \ref{ass:7} may also be relaxed in favor of finite \textit{nested} conditional moments, where both $\mathbb{E}[X_{i}] = \boldsymbol{\mu}^0_s$ and $\text{Var}[X_{i}] = \Sigma^0_s$ are verified if and only if $s = S_i^* \in \mathcal{S}^*$, where $\mathcal{S}^* = \{1,2,...,\bar{S}^*\}$ with $\bar{S}^* > \bar{S}$ and $S_i = (S_i^*)^*$, with $(\cdot)^*:\mathcal{S}^* \rightarrow \mathcal{S}$ being a surjective and non-injective permutation of the group labels. Specifically, the unobserved heterogeneity in the covariates can be more important than the one described by Assumption \ref{ass:1}, as long as the former can be mapped into the latter by merging some groups of covariates. The inverse is however not possible without violating Assumption \ref{ass:7}, meaning that the heterogeneity in the true TEs cannot be larger than the heterogeneity in the covariates.
	\par
	Assumption \ref{ass:8} states that every true covariance matrix $\Sigma_s^0$ has to be different across groups and that this difference must not vanish when the number of covariates goes to infinity. This assumption is not very restrictive since it must hold only for a small fraction of all elements in each covariance matrix; it does not require the mean values of the covariates to be different across groups. However, the overall misclassification rate is likely to be smaller in finite samples if both the mean values and the covariance matrices differ between groups compared to solely differences across covariance matrices.
	\par 
	Finally, Assumption \ref{ass:9} states that the size of each group will go to infinity as the total sample size goes to infinity, while leaving the group's growth rate unspecified. Here, the assumption that $\bar{S}$ is a known and fixed population parameter can be problematic in small samples, especially if some groups represent a very small share of the target population. In practice, within-group treatment effect parameters may not be precisely estimated or identified for all treatment options due to sample size. In such a case, it may be convenient to fix a minimum size for each estimated group, or to drop treatment options (and all corresponding observations) that do not represent a sufficiently large share of the entire sample. Assumption \ref{ass:9} can also be modified such that it holds even if the true number of groups $\bar{S}$ is a function of the sample size $N$.
	\par
	The main theoretical results of the paper are based on the following definitions and are presented below. All proofs are provided in Appendix \ref{app:A}. Note that all theoretical results implicitly assume that the econometrician knows the true number of groups, $\bar{S}$, which can be estimated using $K$-fold cross-validation, or another technique from the relevant literature \citep{tibshirani_estimating_2001,sugar_finding_2003}.
	\begin{definition}\label{def:1}
		\begin{enumerate}[label=(\roman*)] Using the same notation as in Assumption \ref{ass:7}, I define~:
			\item The squared Mahalanobis distance for the $i^{th}$ observation as follows
			\begin{align*}
				d_i^{2M}(\boldsymbol{\mu}_{s}, \Sigma_s) := (X_i - \boldsymbol{\mu}_{s})^{\top} \{\Sigma_s\}^{-1} (X_i - \boldsymbol{\mu}_{s}) + \log \det(\Sigma_s),
			\end{align*}
			for any $\boldsymbol{\mu}_{s} \in \Theta \subseteq \mathbb{R}^{p}$ and any matrix $\Sigma_s \in \Omega$, where $\det(\Sigma_s)$ is the determinant of $\Sigma_s$.
			\item The Mahalanobis distance classifier for the $i^{th}$ observation as follows
			\begin{align*}
				z^M_i(\boldsymbol{\mu},\Sigma):= \arg \min_{s \in \mathcal{S}} d_i^{2M}(\boldsymbol{\mu}_s, \Sigma_s),
			\end{align*}
			where $\boldsymbol{\mu} = (\boldsymbol{\mu}_1,..,\boldsymbol{\mu}_{\bar{S}})$ and $\Sigma = (\Sigma_1,...,\Sigma_{\bar{S}})$.
			\item The total squared Mahalanobis distance (TSMD) as follows
			\begin{align*}
				d^{2TM}(\mathbf{z},\boldsymbol{\mu},\Sigma)&:=  \sum_{s\in \mathcal{S}} \sum_{i\in [N]} \mathbbm{1}[z_i = s] d_i^{2M}(\boldsymbol{\mu}_{s}, \Sigma_s),
			\end{align*}
			for any $\mathbf{z} = (z_1,...,z_N)^{\top} \in \mathcal{Z} \subseteq \mathcal{S}^{N}$.
		\end{enumerate}
	\end{definition}
	\begin{lemma}\label{lem:1}
		Suppose Assumptions \ref{ass:7}-\ref{ass:8} are verified. Then the total squared Mahalanobis distance, $d^{2TM}(\mathbf{z},\boldsymbol{\mu},\Sigma)$, never increases between two consecutive iterations of the weighted K-means algorithm if $\Sigma^{(k)}_s$ is positive-definite for every $s \in \mathcal{S}$ and every iteration $k \in \{1,2,...,\bar{k}\}$, with $\bar{k} < \infty$.
	\end{lemma}
	\begin{theorem}\label{th:1}
		Suppose Assumptions \ref{ass:7}-\ref{ass:8} are verified. Let also the Mahalanobis distance classifier $z^M_i(\cdot,\cdot)$ be defined as in Definition \ref{def:1} $(ii)$. Then
		\begin{align*}
			\Pr[z^M_i(\boldsymbol{\mu}^0,\Sigma^0) \ne S_i] = O(p^{-1}),
		\end{align*}
		for any given $i \in [N]$ provided a suitable permutation in the true groups' labels, and where $p$ stands as the number of covariates.
	\end{theorem}
	\begin{theorem}\label{th:2}
		Suppose Assumptions \ref{ass:7}-\ref{ass:9} are verified. Let the total squared Mahalanobis distance $d^{2TM}(\cdot,\cdot,\cdot)$ be defined as in Definition \ref{def:1} $(iii)$. Then
		\begin{align*}
			(\bar{\boldsymbol{\mu}},\bar{\Sigma}) = \arg \min_{\boldsymbol{\mu}\in \Theta^{\bar{S}},\Sigma \in \Omega^{\bar{S}}}   \plim_{N \to \infty} d^{2TM}(\mathbf{z},\boldsymbol{\mu},\Sigma),
		\end{align*}
		for any $\mathbf{z} \in \mathcal{Z}$ such that $\bar{\Sigma} \in \Omega^{\bar{S}}$, where $\bar{\boldsymbol{\mu}} = \boldsymbol{\mu}^0$ and $\bar{\Sigma} = \Sigma^0$ if $\mathbf{z} = \mathbf{z}^* \equiv (z_1^*,...,z_N^*)^{\top}$, with $z_i^* = S_i$ for almost every $i \in [N]$ provided a suitable permutation in the true groups' labels.
	\end{theorem}
	\begin{corollary}\label{cor:11}
		Suppose Assumptions \ref{ass:7}-\ref{ass:9} are verified. Let the Mahalanobis distance classifier $z^M_i(\cdot,\cdot)$ and the total squared Mahalanobis distance $d^{2TM}(\cdot,\cdot,\cdot)$ be defined as in Definition \ref{def:1} $(ii)$ and \ref{def:1} $(iii)$, respectively. Then
		\begin{align*}
			(\mathbf{z}^M({\boldsymbol{\mu}}^0,{\Sigma}^0), {\boldsymbol{\mu}}^0,{\Sigma}^0) = \arg \min_{\mathbf{z}\in \mathcal{Z},\boldsymbol{\mu} \in \Theta^{\bar{S}},\Sigma \in \Omega^{\bar{S}}}   \plim_{N,p \to \infty} d^{2TM}(\mathbf{z},\boldsymbol{\mu},\Sigma),
		\end{align*}
		provided a suitable permutation in the true groups' labels, and where $\mathbf{z}^M(\cdot,\cdot) = ({z}^M_1(\cdot,\cdot),...,{z}^M_N(\cdot,\cdot))^{\top}$.
	\end{corollary}
	\begin{corollary}\label{cor:1}
		Suppose Assumption \ref{ass:1} and Assumptions \ref{ass:3}-\ref{ass:9} are verified. Let the true CATE $\tau_{d}(a,s)$, the estimated CATE $\hat{\tau}_{d}(a,s,\hat{S}_i)$, and the Mahalanobis distance classifier $z^M_i(\cdot,\cdot)$ be defined as in Assumption \ref{ass:1}, Definition \ref{def:0} $(iii)$, and Definition \ref{def:1} $(ii)$, respectively. Then
		\begin{align*}
			\hat{\tau}_{d}(a,s,z^M_i({\boldsymbol{\mu}}^0,{\Sigma}^0)) - \tau_{d}(a,s) = O_p\left(1/\sqrt{N^0_{s}}\right),
		\end{align*}
		if $p/N \to \infty$ as $N \to \infty$ for any triple $(z^M_i(\cdot,\cdot), d, a) \in \mathcal{S} \times \mathcal{D} \times [0,1]$ provided a suitable permutation in the true groups' labels.
	\end{corollary}
	Lemma \ref{lem:1} shows that using the weighted K-means algorithm never increases the TSMD between two consecutive iterations. This lemma is practical for searching the global minimum of the TSMD across the entire parameter space. However, it requires every $\Sigma^{(k)}_s$ to be positive-definite at each iteration of the algorithm, which will pose a problem if $p > N^{(k)}_s$, where $N^{(k)}_s$ denotes the number of observations assigned to the $s^{th}$ group at the $k^{th}$ iteration. The matrix $\Sigma^{(k)}_s$ will also be singular if some covariates show no within-group variation or near multicollinearity among each other, which is likely to occur in small groups and/or if $X_i$ contains binary or discrete elements.
	\par
	To solve this type of issue, it is possible to use eigenvalue decomposition to enforce a minimum, strictly positive eigenvalue in every covariance matrix at every iteration of the weighted K-means algorithm. Doing so guarantees that the resulting covariance matrix will be invertible and positive-definite. This can however invalidate the conclusion of Lemma \ref{lem:1} if the constrained covariance matrices are too different from the unconstrained ones. In the empirical application, the objective function decreases between two consecutive iterations for all sets of initial parameter values, except when the enforced minimum eigenvalue, denoted by $\lambda_{min}$, is set at a relatively large value (i.e., larger than 2). On the other hand, setting $\lambda_{min}$ too close to zero increases the number and the sharpness of local minima in the TSMD objective function, which makes the search for the global minimum more challenging. This is why, in the empirical analysis, I use a grid of different $\lambda_{min}$ values bounded between 0.01 and 1.5 to limit the impact of this restriction on the convergence behavior of the algorithm.\footnote{Note also that the weighted K-means algorithm reduces to the standard K-means algorithm when $\lambda_{min} \ge \max \{\lambda_{1,max},...,\lambda_{S,max}\}$ at each iteration, where $\lambda_{s,max}$ is the maximum eigenvalue of $\Sigma_s$ for a given $s \in \mathcal{S}$. In this case, it is easy to show that $d_i^{2M}(\boldsymbol{\mu}_{s}, \Sigma_s) = \lambda_{min}^{-1} (X_i - \boldsymbol{\mu}_{s})^{\top} (X_i - \boldsymbol{\mu}_{s}) + p\log (\lambda_{min})$, and that minimizing $d^{2TM}(\mathbf{z},\boldsymbol{\mu}, \Sigma)$ with respect to $\mathbf{z}$, $\boldsymbol{\mu}$, and $\Sigma$ will yield the same results as the standard K-means algorithm. Such an equivalence also highlights the fact that $\lambda_{min}$ will determine the maximum weight that is given to any covariate introduced in the weighted K-means algorithm.}
	\par
	The main benefit of using the minimum eigenvalue approach is that it allows the use of all available covariates to estimate all group memberships. As shown in Theorem \ref{th:1}, the probability of misclassifying any observation in the sampled population when using the Mahalanobis distance classifier goes to zero as the number of covariates goes to infinity under Assumptions \ref{ass:7}-\ref{ass:8}. As shown in \cite{langevin_bias-reduced_2026}, similar guarantees do not exist for the standard K-means algorithm with Euclidean distance, unless every vector of (scaled) true means get increasingly distant from each other as the number of covariates goes to infinity. Consequently, it is crucial to include as many covariates as possible in the weighted K-means algorithm to bring down the misclassification rate to zero in finite samples, even though this might lead to $p > N^{(k)}_s$ for several iterations of the algorithm.
	\par
	Theorem \ref{th:2} states that the global minimum of the TSMD is asymptotically located at the true means and true covariance matrices if almost every observation in the sample is correctly classified in the limit. Note that the asymptotic global minimum of the TSMD is unique under Assumptions $\ref{ass:7}$-$\ref{ass:9}$, unless two groups (or more) share the same vector of mean values and the same covariance matrix, which is (very) unlikely when $p$ is (very) large. Combining the results of Theorems \ref{th:1} and \ref{th:2} implies that globally minimizing the TSMD leads to consistent estimates of all mean values, covariance matrices, and group memberships, regardless of the rate at which $p$ grows relative to $N$ (Corollary \ref{cor:11}).
	\par
	Unlike Corollary \ref{cor:11}, Corollary \ref{cor:1} requires $p$ to increase at a faster rate than $N$ to guarantee that estimated CATEs are $\sqrt{N^0_s}$-consistent under Assumption \ref{ass:6} for every $s \in \mathcal{S}$. If $p$ does not increase at a faster rate than $N$, Theorem 2 of \cite{langevin_bias-reduced_2026} shows that some observations may remain misclassified (even asymptotically), which could invalidate Assumption \ref{ass:6} and reduce the convergence rate for all CATEs. To avoid this, one could also modify Assumption \ref{ass:6} such that the outcome models still converge to the truth at the $1/\sqrt{N^0_{s}}$ rate in the presence of misclassification, but it is unclear if this is a realistic assumption for most outcome models in cross-sectional settings \citep{dzemski_convergence_2021}. On the other hand, $p$ increasing at a faster rate than $N$ will lead to $p > N^{(k)}_s$ for any pair $(s,k) \in \mathcal{S} \times \{1,2,...,\bar{k}\}$ when $N$ is large, hence hampering the ability of the weighted K-means algorithm to find the optimal solution if $\lambda_{min}$ is too large. Nonetheless, it can be that $z_i^M({\boldsymbol{\mu}}^0,{\Sigma}^0) = (S_i)$ for every $i \in [N]$ even if $p << N$ in finite samples under Assumptions \ref{ass:7}-\ref{ass:8}. Under these two assumptions, the ability of the minimum TMSD to correctly classify every observation in finite sample mostly depends on the differences in $\Sigma^0_s$ across groups.
	\par
	\par
	Overall, the above results show that minimizing the TSMD via the weighted K-means algorithm recovers all true group memberships under Assumptions \ref{ass:7}-\ref{ass:8}. Predicted group memberships can then be used to obtain consistent estimates of all CATEs provided that all baseline functions, $B_e(X_i,S_i)$, are known. The general credibility and precision of the estimates will depend on both the ``quality and quantity" of the covariates introduced in the weighted K-means algorithm, as well as on the value of $\lambda_{min}$, and on how the algorithm is initialized.
	\par
	The proposed methodology does not rely on DR scores $\hat{\Gamma}_{i,d}$. Because DR scores rely on observed outcome values $Y_i$, they cannot be used to build policy rules that would guide treatment choice for novel, yet untreated observations. Indeed, any policy rule based on doubly robust scores requires $\hat{\Gamma}_{i,d}$ to be predicted using only $X_i$ (or a subset of $X_i$).\footnote{This additional prediction step is sometimes referred to as the \textit{best linear pojection} of the DR scores, and is represented, for instance, by Equation (38) in \cite{athey_policy_2021}, and by step 7 in Algorithm 1 of \cite{zhou_offline_2023}. This additional prediction step is not strictly required with the proposed methodology, although using a small subset of $X_i$ to predict $\hat{\Gamma}_{i,d}$ relies on a similar logic to what is done in Section \ref{sec:34} to estimate feasible policy rules.} Although this will lead to a feasible policy rule, it may entail a reduction in welfare if the ranking between the ``predicted" scores and the ``estimated" ones, $\hat{\Gamma}_{i,d}$, is not identical. The probability that such a situation arises increases with the number of treatment options. In contrast, the methodology proposed in this paper does not rely on doubly robust scores nor on the inclusion of observed outcome values in the weighted K-means algorithm, thereby reducing the risk of overall welfare loss while improving external validity.
	\par
	Other approaches based on nonparametric estimators and statistical learning, such as kernel regression, sieve-based methods, random forests, and neural networks, can certainly be used to estimate CATEs \citep{foster_subgroup_2011,athey_policy_2021}. However, some of these algorithms do not offer strong theoretical guarantees in small/medium samples, while others are not necessarily well-suited for CATE estimation in a multi-action framework with semicontinuous treatments. Despite this, empirical results obtained with the proposed approach are compared to those obtained by the generalized random forest (GRF) approach of \cite{athey_generalized_2019}. Additional comparisons with other ML techniques are left to further research.
	\subsection{Best Linear Predictor of the Estimated CATEs}\label{sec:32}
	As with all modern machine learning techniques, there are several reasons why the CATEs estimated by the weighted K-means algorithm may be biased and/or imprecise. First, Assumption \ref{ass:1} may not be sufficiently close to the truth, in which case within-group CATE estimates may suffer from contamination bias \citep{goldsmith-pinkham_contamination_2024}. Second, Assumption \ref{ass:3} may also be violated because anticipations form over time or because of selection on unobservables, which is likely to bias the estimated CATEs upward. Third, Assumption \ref{ass:6} is likely to fail if the outcome model is misspecified, leading to biased or imprecise CATE estimates. It may also be that heterogeneity in the covariates is completely different from TEH, which would violate Assumption \ref{ass:7}. Note that, in such a case, CATE estimates should not be significantly different across groups.
	\par
	Due to all the aforementioned reasons, I compute the best linear predictor (BLP) of every CATE across groups and treatment options, thereby extending the BLP method to the multi-action framework. The BLP method is a post-processing approach that combines sample splitting and the estimated CATE ``proxies" to yield a second set of CATEs (i.e., the ``BLPs" of the true CATEs), which are known to improve over the estimated CATEs in the mean squared error sense \citep{chernozhukov_fisher-schultz_2025}. Using multiple splits and quantile aggregation also allows for valid inference on the BLPs (or any other target) provided certain regularity conditions are satisfied. However, computing the variance of the BLP for each group and treatment option is not trivial due to the presence of generated regressors. Since inference on the BLPs is a second-order issue with respect to the identification of policy rules based on expected values, I only perform approximate inference on the BLPs by treating all demeaned CATEs $(\hat{\tau}_{i,d}(1) - |\mathcal{M}|^{-1}\sum_{i\in \mathcal{M}}\hat{\tau}_{i,d}(1))$ in Eq. (\ref{eq:BLP1}) as fixed quantities. Performing asymptotically valid inference on the BLPs, either through bootstrapping or the use of bias-corrected confidence intervals, is left to further research.
	\par
	To estimate the BLP for each CATE in the multi-action framework, I estimate the following regression equation
	\begin{align}\label{eq:BLP}
		Y_i = {\alpha}_{0}\tilde{X}_{i} + \sum_{­d=1}^{\bar{D}}[{\alpha}_{1,d}(A_{i,d}-\hat{p}_{i,d}) + {\alpha}_{2,d}(A_{i,d}-\hat{p}_{i,d})(\hat{\tau}_{i,d}(1) - \mathbb{E}_{\mathcal{M}}[\hat{\tau}_{i,d}(1)])] + \epsilon_i,
	\end{align}
	where $i \in \mathcal{A}$, with $\mathcal{A}$ being a random subset of $[N]$, $\hat{p}_{i,d} = \text{Pr}[A_{i,d}>0|X_i,\hat{S}_i]$ is the estimated probability of being assigned to treatment $d$ (based on the entire sample), $\hat{\tau}_{i,d}(1) \equiv \hat{\tau}_{d}(1,\hat{S}_i,\hat{S}_i)$ is the estimated CATE for the $i^{th}$ observation given treatment option $d$ and estimated group membership $\hat{S}_i$ under perfect adherence, $\mathbb{E}_{\mathcal{M}}[\hat{\tau}_{i,d}(1)]$ is the average of $\hat{\tau}_{i,d}(1)$ for all $i \in {\mathcal{M}}$ with ${\mathcal{M}} ={\mathcal{A}}^c$ under perfect adherence, and where $\tilde{X}_{i} = [1,\hat{B}(X_i),\hat{p}_{i,1},...,\hat{p}_{i,\bar{D}},\hat{\tau}_{i,1}(1)\times\hat{p}_{i,1},...,\hat{\tau}_{i,\bar{D}}(1)\times\hat{p}_{i,\bar{D}}]^{\top}$ with $\hat{B}(X_i)$ representing the estimated baseline function for patient $i$ in the outcome equation (if any). Both $\hat{B}(X_i)$ and $\hat{\tau}_{i,d}(1)$ with $d=1,...,\bar{D}$ are estimated using the main sample $\mathcal{M}$ such that $\hat{\tau}_{i,d}(1) - \mathbb{E}_{\mathcal{M}}[\hat{\tau}_{i,d}(1)]$ measures the deviation of the $i^{th}$ estimated CATE from the average effect in the main sample for the $d^{th}$ treatment option. Note also that all CATEs are estimated under perfect adherence to avoid ``double-counting" the effect of imperfect adherence on the last regressor, $(A_{i,d}-\hat{p}_{i,d})(\hat{\tau}_{i,d}(1) - \mathbb{E}_{\mathcal{M}}[\hat{\tau}_{i,d}(1)])$.
	\par
	The BLP of the $d^{th}$ treatment option for observation $i$ under perfect adherence is then defined as follows 
	\begin{align}\label{eq:BLP1}
		\hat{\tau}^{BLP}_{d}(1,\hat{S}_i,\hat{S}_i) \equiv \hat{\tau}^{BLP}_{i,d}(1) :=  \hat{\alpha}_{1,d} + \hat{\alpha}_{2,d}(\hat{\tau}_{i,d}(1) - |\mathcal{M}|^{-1}\sum_{i\in \mathcal{M}}\hat{\tau}_{i,d}(1)),
	\end{align}
	where $\hat{\alpha}_{1,d}$ and $\hat{\alpha}_{2,d}$ are obtained from Eq. (\ref{eq:BLP}) using ordinary least squares, and where $|\mathcal{M}|^{-1}\sum_{i\in \mathcal{M}}\hat{\tau}_{i,d}(1)$ is an estimate of $\mathbb{E}_\mathcal{M}[\hat{\tau}_{i,d}(1)]$. If $\hat{\alpha}_{2,d}$ is larger than zero, this indicates the presence of heterogeneity in the true CATEs $\tau_{d}(1,S_i)$, and also that $\hat{\tau}_{i,d}(1)$ is a relevant predictor of this heterogeneity. If $\hat{\alpha}_{2,d} = 1$, then $\hat{\tau}_{i,d}(1)$ is a perfect predictor of $\tau_{d}(1,S_i)$. In general, $\hat{\alpha}_{2,d} \ne 1$, correcting for bias or noise in $\hat{\tau}_{i,d}(1)$ \citep{chernozhukov_fisher-schultz_2025}. Using different splits to estimate Eq. (\ref{eq:BLP}) allows to recover the distribution of both $\hat{\alpha}_{0,d}$ and $\hat{\alpha}_{1,d}$, which can then be used to test whether $\hat{\alpha}_{1,d}$ is significantly different than zero. Finally, note that the BLP of treatment option $d$ given adherence value $a$ can be obtained by assuming a specific functional form $g(\cdot)$ such that $\hat{\tau}^{BLP}_{i,d}(a) = g(a,\hat{\tau}^{BLP}_{i,d}(1))$. For the rest of the paper, all estimated CATEs will refer to the BLP estimates, with $\hat{\tau}_{i,d}(a) = \hat{\tau}^{BLP}_{i,d}(a) \equiv \hat{\tau}^{BLP}_{d}(a,\hat{S}_i,\hat{S}_i)$ where $\hat{\tau}^{BLP}_{d}(a,\hat{S}_i,\hat{S}_i)$ refers to the BLP estimate of the $d^{th}$ treatment option conditional on estimated group $\hat{S}_i$ and adherence level $a$.
	\subsection{Policy Rules with Predicted Adherence}\label{sec:33}
	If treatment adherence is not always perfect in practice, it may depend on the chosen treatment option. In this case, using the estimated CATEs with perfect adherence to estimate any policy rule may lead to suboptimal treatment allocations. To account for the effect of imperfect adherence on policy rules, it is possible to replace the fixed adherence value $a$ by the potential adherence value $A_i(d)$, where $A_i(d)$ is defined as in Section \ref{sec:21} and corresponds to the adherence value we would have observed for patient $i$ had she received treatment option $d$. We can thus rewrite Eq. (\ref{eq:1}) as follows
	\begin{align}
		\pi^*_C &= \arg \max_{\pi \in \Pi} \sum_{s \in \mathcal{S}} \mathbbm{1}[S_i=s] \mathbb{E}[Y_i(\pi,A_i(\pi))|S_i=s],\nonumber \\ \label{eq:adh1}
		&\equiv \arg \max_{\pi \in \Pi} \sum_{s \in \mathcal{S}} \mathbbm{1}[S_i=s] \tau_{\pi}(A_i(\pi),s), 
	\end{align}
	where the equivalence is a direct consequence of Assumption \ref{ass:1}.
	\par
	In its current form, Eq. (\ref{eq:adh1}) only allows adherence to vary across treatment option and latent state value. To allow adherence to vary across patients, I apply the empirical welfare maximization principle of \cite{kitagawa_who_2018} and solve the sample analog of Eq. (\ref{eq:adh1}) under the assumption that the joint distribution $\mathcal{P}$ of the target population is identical to that of the sampled population. The sample analog of Eq. (\ref{eq:adh1}) is defined as follows
	\begin{align}\label{eq:in_sample}
		\hat{\pi}^*_C &:= \arg \max_{\pi \in \Pi} \sum_{i \in [N]} \sum_{s \in \mathcal{S}}  \mathbbm{1}[\hat{S}_i=s] \hat{\tau}_{\pi}(\hat{A}_i(\pi),s,\hat{S}_i),
	\end{align}
	where $\hat{\tau}_{d}(a,s,\hat{S}_i)$ is as in Definition \ref{def:0} $(iii)$. To compute the estimated policy rule $\hat{\pi}^*_C$, we can replace $\hat{S}_i$ by any (consistent) classifier, and $\hat{\tau}_{\pi}(\cdot,\cdot,\cdot)$ by its corresponding BLP estimate. To estimate the potential adherence values, I assume that $A_i(d) = a(W_i,\tau^e_{i,d}(1))$, where $a(\cdot,\cdot)$ is a parametric adherence function with $a(\cdot,\cdot)=0$ if $d=0$, $W_i \in \mathcal{W} \subseteq \mathbb{R}^{p'}$ is a set of individual and environmental features (not necessarily a subset of $X_i$), and where $\tau^e_{i,d}(1) \in \mathbb{R}$ is the anticipated effectiveness of treatment option $d$ under perfect adherence for patient $i$.
	\par
	Introducing anticipated treatment effectiveness in the model allows adherence to depend on the patients' and/or the physicians' anticipations regarding treatment effectiveness, which is natural if adherence to treatment is \textit{unconditionally} endogenous with respect to patient-specific outcomes.\footnote{As opposed to being conditionally endogenous, which would violate either Assumptions \ref{ass:2} or \ref{ass:3}. Eqs. (\ref{eq:linear}) and (\ref{eq:linear2}) provide an example of conditional absence of endogeneity since $\epsilon_i$ is uncorrelated with observed adherence $A_{i,D_i}$ conditional on $S_i$. Not conditioning on $S_i$ however leads to endogeneity bias if $\xi_{s,d}$ is correlated with $\tau_{d}(1,s)$ for any $d \in \{1,...,\bar{D}\}$ across values of $s$.} Because the anticipated effectiveness $\tau^e_{i,d}(1)$ is unobserved for all $i \in [N]$, it can be estimated by $\hat{\tau}_{i,d}(1)$ as defined in the previous section. Although $\hat{\tau}_{i,d}(1)$ may be a biased estimator of $\tau^e_{i,d}(1)$, it is difficult to obtain other estimators of $\tau^e_{i,d}(1)$ without relying on survey data about expectations of patients and/or physicians.
	\par
	In practice, this results in a two-step estimation procedure where the first step uses observed adherence $A_{i,D_i}$ to estimate all CATEs via the BLP approach described in the previous section. The second step models the observed adherence $A_{i,d}$ as a function of $W_i$ and $\hat{\tau}_{i,d}(1)$ for each treatment option separately using treated observations only. Given that adherence is a semicontinuous variable with a point mass at the unit value, it can be modeled via a two-part model where the first part models the probability of observing perfect adherence, and the second part models observed adherence when it is imperfect \citep{leung_choice_1996}. The estimated adherence value for every observation in the sample can then be plugged into the BLP estimates, allowing us to compute the estimated policy rule $\hat{\pi}^*_C$ with predicted adherence.
	\par
	It is also possible to adapt the above estimated policy rule to any of the other decision criteria described in Section \ref{sec:23}. For instance, assuming Assumption \ref{ass:1} is verified, I obtain the following estimated policy rule under the Bayesian decision criterion and predicted adherence to treatment
	\begin{align*}
		\hat{\pi}^*_B &= \arg \max_{\pi \in \Pi}  \sum_{i \in [N]} \sum_{s \in \mathcal{S}} \text{Pr}[S_i=s|X_i] \hat{\tau}_{\pi}(\hat{a}(W_i,\hat{\tau}_{i,\pi}(1)),s,\hat{S}_i),
	\end{align*}
	where $\text{Pr}[S_i=s|X_i]$ is a subjective prior over $S_i$ conditional on $X_i$, which differs from the unconditional subjective prior used in Eq. (\ref{eq:2}).\footnote{It is possible to estimate such a conditional prior distribution by using, for instance, $\widehat{\mathbb{P}[S_i=s|X_i]} = \frac{((X_i - \hat{\boldsymbol{\mu}}_{s})^{\top} \{\hat{\Sigma}_s\}^{-1} (X_i - \hat{\boldsymbol{\mu}}_{s}))^{-1}}{\sum_{s \in \mathcal{S}} ((X_i - \hat{\boldsymbol{\mu}}_{s})^{\top} \{\hat{\Sigma}_s\}^{-1} (X_i - \hat{\boldsymbol{\mu}}_{s}))^{-1}}$, where $(\hat{\boldsymbol{\mu}}_{s},\hat{\Sigma}_s)$ is provided by the weighted K-means algorithm for every $s \in \mathcal{S}$.} For simplicity, the adherence model $a(\cdot,\cdot)$ is assumed to be independent of the state variable, but this can be relaxed via the estimation of within-group adherence models, provided sufficient (exogenous) variation in observed adherence across estimated groups.
	\subsection{Feasible Policy Rules}\label{sec:34}
	If we assume there exists no global constraint over the class of feasible policies $\Pi$, the infeasible policy rule is what we would obtain if the true group membership $S_i$ were known for all observations. In this case, I define the estimated infeasible policy rule as follows
	\begin{align}\label{eq:infeasible}
		\hat{\pi}^{I} \equiv (d^*_{1},...,d^*_{N})  &:= \arg \max_{d_1,...,d_{N} \in \mathcal{D}^{N}} \sum_{i \in [N]} \sum_{s \in \mathcal{S}} \mathbbm{1}[S_i=s] \hat{\tau}_{d_{i}}(A_i(d_i),s,S_i),
	\end{align}
	by Assumption \ref{ass:1}, where the infeasible (and unrestricted) policy rule $\hat{\pi}^I$ amounts to selecting the treatment option $d^*_i$ that maximizes the estimated CATE $\hat{\tau}_{d}(\cdot,\cdot,\cdot)$ for every patient in the sample given potential adherence $A_i(d_i) \in [0,1]$ for patient $i$ being assigned treatment option $d_i$. The potential adherence value $A_i(d_i)$ can be replaced by either predicted, perfect, or any level of adherence that the econometrician deems realistic to impose. For simplicity, I will use the potential adherence notation throughout the rest of the paper, although it can be changed to any other predetermined value.
	\par
	Since $S_i$ is unobserved, I replace $S_i$ by the Mahalanobis distance classifier $z^M_i(\hat{\boldsymbol{\mu}},\hat{\Sigma})$ for all $i \in [N]$, where $(\hat{\boldsymbol{\mu}},\hat{\Sigma})$ corresponds to the estimated mean vectors and covariance matrices obtained via the minimization of the TSMD. If Assumptions \ref{ass:7}-\ref{ass:9} are satisfied, Corollary \ref{cor:11} implies that the estimated policy rule $\hat{\pi}^*_C$ with $\hat{S}_i = z^M_i(\hat{\boldsymbol{\mu}},\hat{\Sigma})$ should be similar to $\hat{\pi}^I$ provided $p$ and $N$ are sufficiently large. In this case, it is not unrealistic to treat the estimated group membership $z^M_i(\hat{\boldsymbol{\mu}},\hat{\Sigma})$ as the true one for all $i \in [N]$.
	\par
	However, using $z^M_i(\hat{\boldsymbol{\mu}},\hat{\Sigma})$ typically requires a large number of covariates to ensure that the number of misclassified observations is close to zero in finite samples. This is why estimated policy rules based on $z^M_i(\hat{\boldsymbol{\mu}},\hat{\Sigma})$ are still ``infeasible" in the sense that they are based on too much information to be practical in most clinical contexts. This is why, in this paper, the class of feasible policy rules $\Pi$ is restricted to the class of rules that optimally predict $z^M_i(\hat{\boldsymbol{\mu}},\hat{\Sigma})$ using only a subset $V \in \mathcal{V}_K \subset \mathcal{X}$, where $\mathcal{V}_K$ is the space of all subsets of $\mathcal{X}$ such that $|V| = {1,...,K}$, with $K$ being a small positive integer (often less than 10).
	\par
	This type of restriction over $\Pi$ is convenient because the potential welfare loss associated with such class of feasible policy rules depends on the ability of $V$ to correctly predict $z^M_i(\hat{\boldsymbol{\mu}},\hat{\Sigma})$ for all $i \in [N]$, under the assumption that $z^M_i(\hat{\boldsymbol{\mu}},\hat{\Sigma})$ corresponds to the ground truth. To quantify the welfare loss from such a restriction over $\Pi$, I consider the class $\mathcal{H}$ of all real-valued functions $h: \mathcal{V}_K \to \mathcal{S}$ and the associated \textit{misclassification rate}
	\begin{align}\label{eq:misc}
		\text{Mis}_N(h(\cdot),V) := N^{-1}\sum_{i\in [N]} \mathbbm{1}[z^M_i(\hat{\boldsymbol{\mu}},\hat{\Sigma}) \ne h(V_i)],
	\end{align}
	where $z^M_i(\hat{\boldsymbol{\mu}},\hat{\Sigma})$ is considered as the true group membership for all $i \in [N]$, and where $V = (V_1,...,V_N)^{\top}$ with $V_i$ corresponding to the $K$-sized vector of selected covariates for observation $i$. It is then possible to define the optimal function $h^*(\cdot) \in \mathcal{H}$ with optimal subset $V^*$ as follows
	\begin{align}\label{eq:13}
		\hat{h}(V^*) := \arg \min_{h \in \mathcal{H}, V \in \mathcal{V}_K} N^{-1}\sum_{i\in [N]} \mathbbm{1}[ z^M_i(\hat{\boldsymbol{\mu}},\hat{\Sigma}) \ne h(V_i)].
	\end{align}
	\par
	Plugging the above definition into Eq. (\ref{eq:infeasible}) yields the following estimated \textit{feasible} policy rule under the certainty criterion
	\begin{align}\label{eq:feasible_c}
		\hat{\pi}^F_{C}(V^*) := \arg \max_{d_1,...,d_{N} \in \mathcal{D}^{N}} \sum_{i\in [N]}  \sum_{s \in \mathcal{S}} \mathbbm{1}[\hat{h}(V_i^*)=s] \hat{\tau}_{d_i}(A_i(d_i),s,\hat{h}(V_i^*)),
	\end{align}
	where the estimated feasible policy rule selects the treatment option that maximizes the estimated CATE for every patient in the sample given adherence $A_i(d_i)$ and predicted group membership $\hat{h}(V_i^*)$. Such a policy rule entails a trade-off between feasibility and welfare loss, where the use of the function $\hat{h}(\cdot)$ and the subset $V^*$ will lead to a suboptimal treatment allocation unless the misclassification rate $\text{Mis}_N(\hat{h}(V^*),V^*)$ is equal to zero for any given value of $N$.
	\par
	It is also possible to adapt the above feasible policy rule to any of the decision criteria described in Section \ref{sec:23}. For instance, applying the Bayesian decision criterion to Eq. (\ref{eq:feasible_c}) yields the following Bayesian version of the estimated feasible policy rule
	\begin{align}\label{eq:feasible_b}
		\hat{\pi}^F_{B}(V^*) := \arg \max_{d_1,...,d_{N} \in \mathcal{D}^{N}} \sum_{i\in [N]} \sum_{s \in \mathcal{S}} \text{Pr}[z^M_i(\hat{\boldsymbol{\mu}},\hat{\Sigma})=s|V^*_i] \hat{\tau}_{d_i}(A_i(d_i),s,\hat{h}(V_i^*)),
	\end{align}
	where $\Pr[z^M_i(\hat{\boldsymbol{\mu}},\hat{\Sigma}) = s|V^*_i]$ is a prior distribution over $z^M_i(\hat{\boldsymbol{\mu}},\hat{\Sigma})$ given $V^*_i$. The main benefit of this last policy rule is that it accounts for the uncertainty that arises from using $\hat{h}(\cdot)$ and $V^*$ to identify group membership, which can reduce welfare loss compared to the certainty criterion shown above.
	\par
	In practice, it is quite difficult to determine which function $h(\cdot) \in \mathcal{H}$ and subset $V \subset \mathcal{X}$ minimize the misclassification rate Mis$(h,V)$. Because of their flexibility and interpretability, several authors have relied on decision trees \citep{sverdrup_policytree_2020,athey_policy_2021,zhou_offline_2023}. In the empirical application, I also use decision trees to predict estimated group memberships with a small subset of $X_i$. Other methods, such as simple linear regression or multinomial logistic regression, can also be used, but these methods tend to be less flexible than decision trees and are not necessarily easy to use for routine decision-making processes. 
	\subsection{Cost-Effectiveness Analysis}\label{sec:35}
	The objective of cost-effectiveness analysis (CEA) in PL is to show how the optimal treatment allocation changes in the sampled population according to the policymaker's WTP, where WTP is expressed in terms of treatment costs (\$) per additional percentage point (pp) of probability of reaching SVR. To perform any of these CEAs, I define the incremental cost-effectiveness ratio (ICER) of treatment option $d$ for patient $i$ with given adherence level $A_i(d) \in [0,1]$ as follows
	\begin{align}\label{eq:ICER}
		\text{ICER}_{i,d} := \frac{c_d}{\hat{\tau}_{i,d}(A_i(d))},
	\end{align}
	where $c_d$ refers to the (fixed) cost of prescribing treatment option $d$, and where $\hat{\tau}_{i,d}(A_i(d)) \equiv \hat{\tau}_{d}(A_i(d),\hat{S}_i,\hat{S}_i)$ with $\hat{S}_i = \{z^M_i(\hat{\boldsymbol{\mu}},\hat{\Sigma}), \hat{h}(V^*_i)\}$ and $A_i(d) = \{\hat{a}(W_i,\hat{\tau}_{i,d}(1)),1\}$, which leads to four different CEAs. Note that Eq. (\ref{eq:ICER}) assumes that the default comparator is the null treatment regimen with zero cost (i.e., $c_0 = 0$) and no marginal effect on SVR (i.e., $\hat{\tau}_{i,0}(0) =0$).
	\par
	Optimal treatment allocation can then be obtained by selecting the treatment option that maximizes the chance of reaching SVR for every patient in the sample, and for which the corresponding ICER value remains below the policymaker's WTP. Formally, this can be represented by the following constrained maximization program (under the certainty criterion; the $C$ subscript is dropped for brevity)
	\begin{align}\label{eq:17}
		\hat{{\pi}}^*(\omega) = &\arg \max_{d_1,...,d_N\in \mathcal{D}^N}  \sum_{i\in [N]} \hat{\tau}_{i,d_i}(A_i(d_i)),\\
		&\text{s.t.} \ \ \text{ICER}_{i,d_i} \le \omega, \ \forall \ i \in [N], \nonumber
	\end{align}
	where $\hat{{\pi}}(\omega)$ represents the set of $N$ treatment choices (i.e., a treatment allocation) as a function of the policymaker's WTP $\omega \ge 0$. It is then possible to illustrate how the in-sample optimal treatment allocation changes as a function of the policymaker's WTP. I show in Appendix \ref{app:F} a simple algorithm that can be used to obtain $\hat{{\pi}}^*(\omega)$ for any $\omega \in \mathbb{R}_{\ge0}$ given potential adherence $A_i(d_i)$ for every pair $(i,d) \in [N]\times \mathcal{D}$. Note that it is also possible to adapt the above program to any other decision criterion by modifying the maximand of Eq. (\ref{eq:17}). 
	\par
	Finally, I also compare the total costs and predicted outcome values provided by the treatment allocation $\hat{{\pi}}^*(\omega)$ to a ``benchmark" allocation. This comparison is performed by computing an aggregate ICER value, which varies with adherence, WTP, and the choice of benchmark allocation. This aggregate ICER value is defined as follows
	\begin{align}\label{eq:171}
		\text{ICER}_{A}(\mathbf{b},\omega) := \frac{ \sum_{i \in \mathcal{B}} c_{\hat{\pi}^*_i(\omega)} - c_{b_i}}{\sum_{i \in \mathcal{B}} \hat{\tau}_{i,\hat{\pi}^*_i(\omega)}(A_i(\hat{\pi}^*_i(\omega))) - \hat{\tau}_{i,b_i}(A_i(b_i))},
	\end{align}
	where $\mathbf{b} = \{b_i\}_{i \in \mathcal{B}}$ is the benchmark treatment allocation for all patients in the subset $\mathcal{B} \subseteq [N]$ with $|\mathcal{B}| = N_b$. Note that all estimated CATEs in the denominator of Eq. (\ref{eq:171}) are computed using $\hat{S}_i = z^M_i(\hat{\boldsymbol{\mu}},\hat{\Sigma})$ as the ground truth for all $i \in \mathcal{B}$, while treatment choices $\hat{\pi}^*_i(\omega)$ and $b_i$ are based on either $z^M_i(\hat{\boldsymbol{\mu}},\hat{\Sigma})$ or $\hat{h}(V^*_i)$. 
	\par
	The value of $\text{ICER}_A(\mathbf{b},\omega)$ can then be plotted against WTP and/or adherence values for a given benchmark allocation $\mathbf{b}$ to identify the situations in which $\hat{{\pi}}^*(\omega)$ dominates (or is dominated) by the benchmark at the aggregate level. Note that if $\text{ICER}_A(\mathbf{b},\omega)$ is positive, this implies that there is a trade-off between total health benefits and total costs, and neither $\hat{\pi}(\omega)$ nor the benchmark allocation dominates the other one. If $\text{ICER}_A^*(\mathbf{b},\omega)$ is negative, then one of the two allocations dominate the other at the aggregate level. The dominance of one allocation over another \textit{does not} imply that one of the two allocations is Pareto optimal because the comparison is at the aggregate level. Nonethelesss, Pareto optimality can be achieved by enforcing an additional positivity constraint such that every term in the sum of the denominator of Eq. (\ref{eq:171}) is positive for any pair $(\mathbf{b},\omega) \in \mathcal{D}^{N_b} \times \mathbb{R}_{\ge0}$.
	\section{Empirical Application: Hepatitis C Treatment in Co-Infected Patients}\label{sec:4}
	This section is divided into four subsections. The first subsection briefly describes the cohort used for the empirical application. The second subsection describes the health condition shared by the cohort's participants and the different treatment options currently used to treat this condition. The third subsection briefly describes the clinical guidelines concerning HepC treatment in both mono- and HIV/HCV co-infected patients. Finally, the fourth subsection describes how the various policy rules are estimated from the cohort data.
	\subsection{The Canadian Co-Infection Cohort}\label{sec:41}
	\begin{figure}[t!]
		\centering
		\includegraphics[width=12cm]{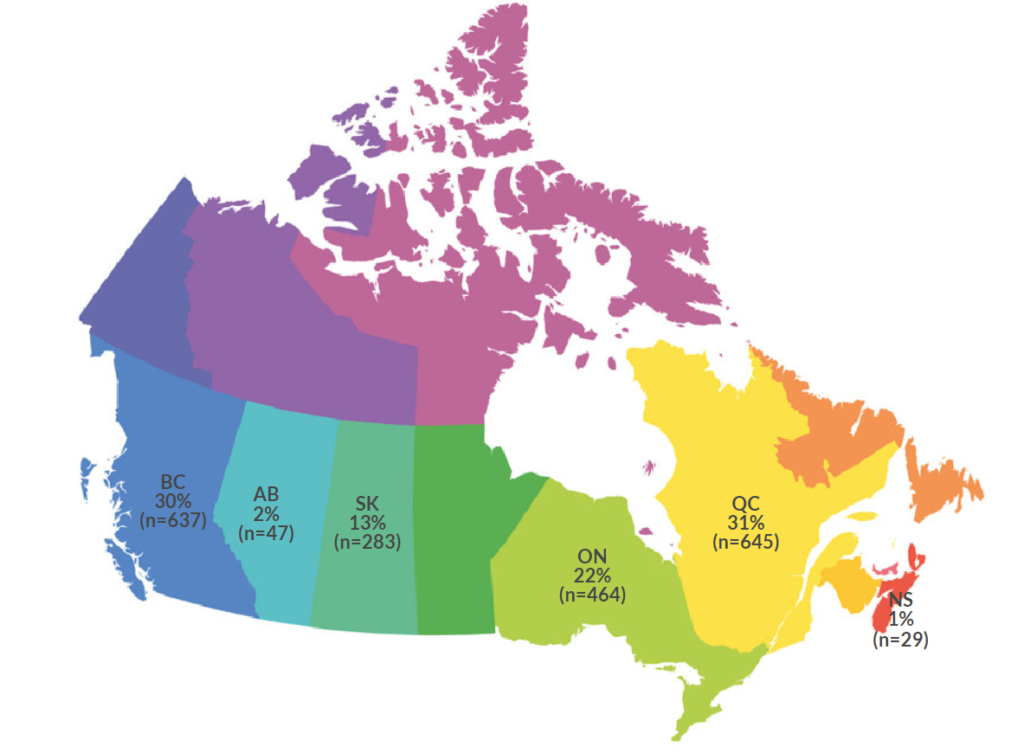}
		\caption{Geographical distribution of the participants enrolled in the Canadian Co-infection Cohort, 2003-2023}
		\label{fig:1}
	\end{figure}
	The Canadian Co-Infection Cohort (CCC) prospectively follows around 2,200 patients with documented HIV infection and chronic Hepatitis C (HepC) infection or evidence of HCV exposure across 6 Canadian provinces between April 2003 and December 2023.\footnote{A chronic HepC infection, as opposed to an acute HepC infection, refers to an infection that does not \textit{spontaneously clear} 6 months after a first positive HCV blood test.} Inclusion criteria for the cohort are very broad~: $>$16 years old, documented HIV seropositive infection, documented HepC infection or evidence of exposure, and capacity to provide informed consent. All eligible patients were approached to participate to avoid selection bias \citep{klein_cohort_2010}. A total of 19 clinical sites are participating in the CCC as of the end of 2023. Figure \ref{fig:1} shows the geographical distribution of all participants enrolled in the CCC between 2003 and 2023.
	\par
	The primary objective of the CCC is to evaluate the effect of highly active antiretroviral therapy (HAART) on the progression of end-stage liver disease (ESLD) in HIV/HCV co-infected patients, while the secondary objective is to evaluate the ``role of HCV treatment in the evolution of liver disease with a particular emphasis on evaluating access to treatment" \citep{klein_cohort_2010}. To meet these objectives, socio-demographic, medical, behavioral, and patient-reported quality of life information is collected for all participants through questionnaires at baseline and during follow-up visits. Follow-up visits are planned every 6 months to update information on risk behaviors, medical treatments, and diagnoses. Blood tests are also performed at each follow-up visit, including liver profile (i.e., concentrations of various liver enzymes), blood cell count, hepatitis A and B serology, plasma HIV and HCV RNA, and HCV genotype. Causes of death, withdrawals from the study, and losses to follow-up (defined as missing more than 2 consecutive follow-up visits) are also systematically recorded.
	\par
	It is estimated that the CCC has enrolled 23\% of the total HIV/HCV co-infected population in care in Canada \citep{saeed_frequent_2024}. The CCC is representative of the co-infected population engaged in clinical care \citep{young_rate_2023}. Figure \ref{fig:2} presents the socioeconomic characteristics of CCC participants at enrollment. Note the large proportion of Indigenous people and low-income patients (less than CAN\$1,500 per month) compared to the overall Canadian population.\footnote{According to Statistics Canada's Canadian Income Survey, around 15\% of the Canadian population lives on less than CAN\$1,500 per month \citep{statistics_canada_table_2024}. The 2021 Census data also revealed that Indigenous people comprise about 5\% of the overall Canadian population \citep{government_of_canada_canadas_2023}.} Because HCV and HIV are blood-borne diseases that spread mostly through injection drug use, low-income Indigenous people are at increased risk of being co-infected. In Canada, evidence suggests that Indigenous peoples account for 70\% to 80\% of new HepC infections among persons who inject drugs \citep{saeed_frequent_2024}.
	\begin{figure}[t]
		\centering
		\includegraphics[width=12cm]{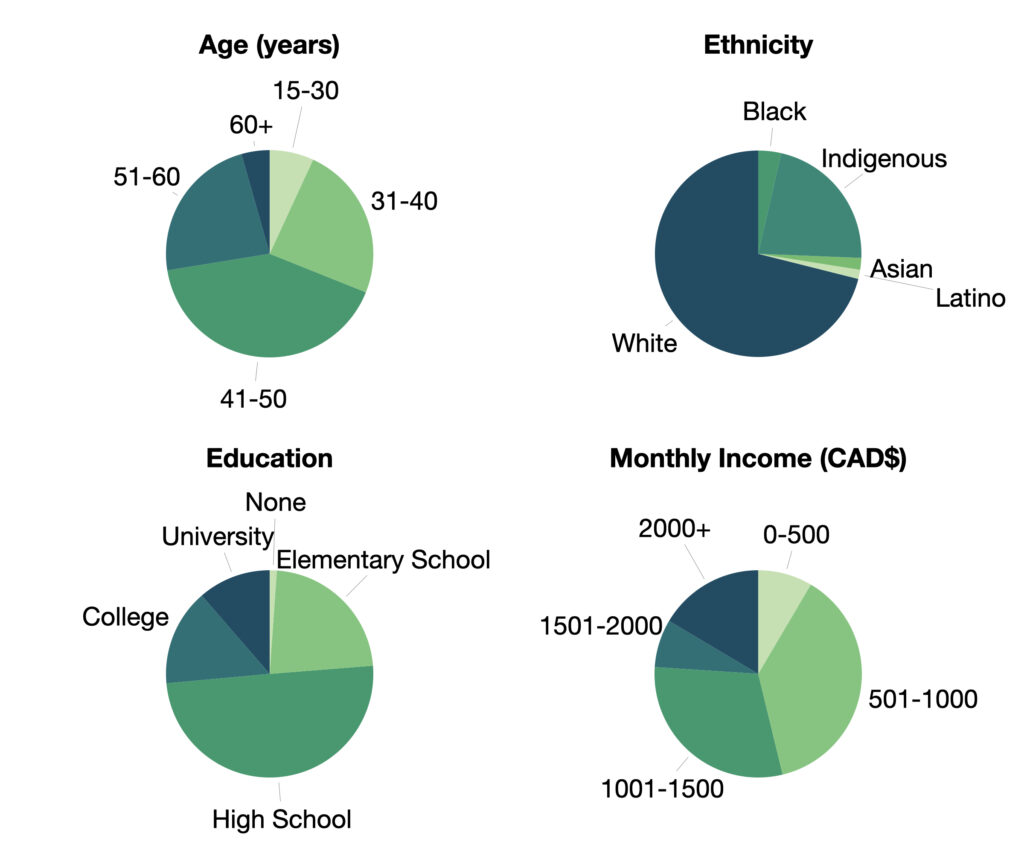}
		\caption{Socioeconomic characteristics at enrollment, Canadian Co-Infection Cohort, 2003-2023}
		\label{fig:2}
	\end{figure}
	\subsection{HIV, HCV, and Direct-Acting Antiviral Agents}\label{sec:42}
	Around 6 million people are believed to be co-infected with HIV and HCV worldwide \citep{soriano_changing_2013}. Since the approval of the first antiretroviral therapy (ART) in 1987, several other antiretroviral drugs have been approved and used to treat HIV infection. This has led to a dramatic reduction in HIV morbidity and mortality in the developed world \citep{tseng_evolution_2015}.
	\par
	HepC is a viral infection that causes liver inflammation, which can lead to lifelong illnesses such as liver cirrhosis and cancer. In Canada, it is estimated that more than 200,000 persons were living with chronic Hepatitis C infection as of the end of 2021, with an incidence of nearly 8,000 new cases per year. HepC-related mortality has been steadily decreasing since the approval of the new direct-acting antiviral (DAA) agents by Health Canada in 2014 \citep{public_health_agency_of_canada_canadas_2024}. Although DAAs have been covered by public health insurance plans in all Canadian provinces since 2018, recent evidence suggests that treatment initiation rates are too low to reach the HCV reduction targets in Canada by 2030 \citep{young_rate_2023, feld_timing_2022, tian_feasibility_2024}.\footnote{The Canadian HCV elimination targets aim to reduce chronic hepatitis C infections by 80\% and related deaths by 65\% between 2016 and 2030.} Systemic barriers such as incarceration and discrimination toward Indigenous people are believed to represent some of the main causes for current low DAA uptake despite universal health insurance coverage in Canada \citep{jiang_hcv_2023, hosein_increased_2024}. Therefore, extensive scale-up in screening, treatment, and harm-reduction strategies may be necessary to reach the Canadian targets \citep{tian_feasibility_2024}. Note that, unlike Hepatitis B, there is no vaccine against HepC, which implies that patients can be reinfected after being cleared of HepC viral load.
	\par
	In HCV mono-infected patients, clinical trials have shown that DAAs are more than 90\% effective in achieving SVR 12 weeks after treatment, regardless of HCV genotype \citep{esteban_efficacy_2018,laurain_safety_2019}. As mentioned in the introduction, reaching SVR is equivalent to curing the disease and is the health outcome of interest in most studies on the efficacy of HCV treatments. Similar results have been obtained in trials with both HCV mono-infected and HIV/HCV co-infected patients \citep{ford_expanding_2012, rockstroh_efficacy_2015}. However, the generalizability of those results outside the clinical trial setting is not obvious due to the very restrictive eligibility criteria in most DAA trials. Specifically, uptake of certain antiretroviral therapy regimens, active illicit drug use, low CD4 cell count, and detectable HIV RNA are the reasons that lead to the exclusion of co-infected patients from most DAA trials \citep{saeed_how_2016}. In parallel, observational studies that estimate the effectiveness of DAAs in co-infected patients obtain SVR rates that are between 10 and 17 pp lower for co-infected persons compared to HCV mono-infected persons \citep{alam_real-world_2017, wilton_real-world_2020}.
	\par
	The lower effectiveness of DAAs in co-infected patients can be attributed to different factors. A recent study looking at 4098 co-infected patients across 6 different countries found that unsuccessful DAA treatment was more likely among patients with evidence of higher immunodeficiency and among those who reported recently injecting drugs \citep{harney_unsuccessful_2025}. Other evidence suggests that non-achievement of SVR for persons who initiate DAA therapies is mainly driven by nonvirological factors such as treatment discontinuation, nonattendance for viral load testing after treatment initiation, or death \citep{darvishian_loss_2020, wilton_real-world_2020}. Other elements, such as drug-drug interactions (DDIs) with current and/or past ART, high degree of hepatic fibrosis, and HCV genotype, may also affect DAA effectiveness in co-infected patients \citep{chen_treatment_2015, kwo_regimens_2015, martel-laferriere_prise_2022}.
	\subsection{DAA Regimens and Clinical Guidelines}\label{sec:43}
	Three main DAAs are currently used to treat HepC infection in co-infected patients, with the following commercial names: Epclusa (EPC), Mavyret (MAV), and Harvoni (HAR). Both the Epclusa and Harvoni treatment regimens consist of a single pill to be taken daily for 12 weeks. The Mavyret treatment regimen consists of three pills to be taken daily for 8 weeks. The American Association for the Study of Liver Diseases (AASLD), which provides guidelines and practical recommendations to health practitioners in the United States regarding HCV treatment, indicates that Harvoni features fewer suspected DDIs with ARTs than Epclusa, which in turn features fewer suspected DDIs with ARTs than Mavyret \citep{american_association_for_the_study_of_liver_diseases_patients_2022}. This partly explains why the price negotiated by the Pan-Canadian Pharmaceutical Alliance for Harvoni ($\approx$ CAN\$798/pill) is higher than the price of Epclusa ($\approx$ CAN\$714/pill), which is in turn higher than the price of Mavyret ($\approx$ CAN\$238/pill).
	\par
	For patients who do not have any of the following: a current Hepatitis B infection, an ongoing pregnancy, a current hepatocellular carcinoma (HCC) diagnosis, a history of liver transplant, and who were never treated with DAAs or older interferon-based treatment regimens, Canadian guidelines recommend prescribing either Epclusa or Mavyret for their respective lengths, regardless of the HCV genotype \citep{martel-laferriere_prise_2022-1}. On the other hand, the AASLD guidelines recommend extending Mavyret's regimen length to 12 weeks in case of a co-infection with HIV. Unlike the AASLD guidelines, Canadian guidelines do not distinguish between HCV mono-infected and HIV/HCV co-infected patients \citep{american_association_for_the_study_of_liver_diseases_patients_2022}. Given the generally lower SVR rates for HIV/HCV co-infected patients, the absence of specific Canadian guidelines for co-infected patients might lead to suboptimal health outcomes within the Canadian co-infected population. Finally, the AASLD recommends using the Harvoni treatment regimen only against certain HCV genotypes (1a, 1b, 4, 5, and 6), whereas Canadian guidelines recommend using Harvoni only in patients where interferon-based treatment regimens have not been successful in the past \citep{martel-laferriere_prise_2022-1}.
	\par
	Since there is no vaccine for HepC, it is not possible to prevent the spread of the virus through mass vaccination campaigns. Moreover, it is estimated that 20\% to 30\% of Canadians infected with HIV are co-infected with HCV. This percentage is likely to rise in the future due to higher reinfection rates in co-infected patients, especially among those who inject drugs \citep{young_rate_2023,hosein_increased_2024}. It is therefore essential to treat every co-infected patient with the most effective option to maintain the HCV transmission rate as low as possible in the future. Given that the marginal effect of each DAA regimen on the probability of reaching SVR in co-infected patients is uncertain, the proposed estimation strategy can be used to derive a policy rule that will identify which DAA among the three listed above is the most effective for any given patient with observed characteristics $X_i$ subject to a set of predetermined (cost) constraints. The next section details such a strategy.
	\subsection{Estimation Strategy}\label{sec:44}
	\subsubsection{Sampled Population}\label{sec:441}
	The sampled population consists of all HCV-infected CCC participants (regardless of the enrollment date) who were still alive and not lost to follow-up before January 1st, 2014, and who were prescribed either Epclusa, Mavyret, or Harvoni after January 1st, 2014, or who were never prescribed any treatment against HCV between January 1st, 2010, and September 1st, 2024. These never-treated patients will act as the set of ``untreated patients" to which every treatment option will be compared, and allow for the fact that any patient infected with HepC can spontaneously clear the virus. The outcome $Y_i$ corresponds to whether a patient has reached SVR between January 1st, 2014, and September 1st, 2024. Patients who were treated multiple times with DAAs or with older interferon-based treatment regimens and at least one DAA regimen between January 1st, 2014, and September 1st, 2024, are not excluded from the sample. This means that the same treated patient can be observed multiple times in the sampled population.\footnote{Because this situation concerns only 39 patients out of 1,187, all within-patient outcome values were assumed to be independent from each other after conditioning on a set of dummy variables indicating all previously assigned treatment options. Observations with missing or unknown outcome values were also discarded from the analysis (less than 20 observations out of 745 remaining treated observations).} SVR outcomes for both treated and untreated patients are provided by the participating centers and are confirmed via a careful examination of the laboratory test results. Additional details about the identification of spontaneous clearance in untreated patients and SVR outcomes in treated patients are described in Appendix \ref{app:C}.
	\subsubsection{Classification Procedure and Outcome Models}\label{sec:442}
	I use a total of 146 pre-treatment variables to perform the classification procedure, including standard individual characteristics (e.g., age, gender, weight, etc.), socioeconomic indicators (e.g., income brackets, educational attainment, etc.), medical diagnoses, past and current ART medications, previous HCV treatment regimens (and the respective SVR outcome), and other health-related information. Additional information about the construction of the set of variables entering the classification procedure is provided in Appendix \ref{app:C}.
	\par
	I introduce this set of variables into the weighted K-means algorithm to minimize the TSMD. Because the TSMD is a non-convex function of the estimated parameters and group memberships, I use a set of 3,000 different initial assignment values (i.e., $\mathbf{z}^{(0)}$) to search for the global minimum of the TSMD for each combination of hyperparameters, where $\lambda_{min}$ varies from 0.01 to 1.5 and the total number of groups $\bar{S}$ varies from 2 to 6. To limit the influence of missing values on the classification procedure and given the large number of missing values in the CCC data, I impute missing covariate values with the group mean values at each iteration of the algorithm. This avoids dropping any observations during the classification procedure.\footnote{The main caveat of such an imputation scheme is that it underestimates (co)variance parameters, which is compensated through the enforcement of the minimum eigenvalue $\lambda_{min}$. Direct adjustment to appropriately ``inflate" the estimated (co)variance parameters led to several convergence issues. Note also that the number of missing values varies across covariates, from a few missing values to several hundred.}
	\par
	I model the within-group SVR outcome $Y_i$ with a linear probability model (LPM) using observed adherence for each treatment option and each $i \in [N]$. LPMs are convenient because Assumption \ref{ass:1} is more likely to hold with LPMs than with nonlinear binary choice models, and because they can identify treatment effects even if all treated patients have a non-null outcome in any given group.\footnote{Identification of the treatment effect of the $d^{th}$ option is not possible in nonlinear binary choice models if $A_{i,d}>0$ implies $Y_i =1$, which often occurs during the cross-validation procedure.} Moreover, information criteria (e.g., AIC and BIC) indicate better in-sample goodness-of-fit with LPMs compared to Probit binary choice models in most cases (results not shown). I computed observed adherence $A_{i,d} \in [0,1]$ for each treated patient using the recorded DAA start and end dates and the standard treatment length for each treatment option (8 or 12 weeks). Participating centers also confirmed complete adherence to treatment for each patient in the cohort.
	\par
	I use 2-fold cross-validation (CV) with 200 different data splittings to select the total number of groups $\bar{S}$ and the minimum eigenvalue $\lambda_{min}$. Specifically, this means that, for each combination of $\bar{S}$ and $\lambda_{min}$, both the weighted K-means algorithm and the estimation of all outcome models are performed using half of the sample (the ``training set") for each fold and each data splitting. OOS predicted outcome values are then computed using the other half of the sample (the ``test set") for each fold and each data split. I then select the optimal hyperparameter values based on the maximum balanced accuracy produced by the entire set of OOS predicted outcome values.\footnote{Balanced accuracy corresponds to the average of specificity and sensitivity, and its maximum value refers to the point that is closest, in terms of the L1 distance, to the the perfect classifier point in the receiver operating characteristic (ROC) curve. It complements the area under the curve of the ROC (AUC-ROC) curve by focusing on the probability cutoff (or threshold) value that generates the smallest false positive/negative rates rather than all possible cutoff values, many of them being useless for prediction purposes.} To avoid identification issues and ensure that Assumption \ref{ass:5} is satisfied in practice, I exclude from the CV procedure any initial assignment values and data splittings that generated groups with no variation among any of the three treatment options. To facilitate the convergence of the weighted K-means algorithm, I use the optimal assignment values obtained with the entire sample after minimization of the TMSD, $\mathbf{z}^M(\hat{\boldsymbol{\mu}},\hat{\Sigma})$, as the initial assignment values for every training set and every combination of hyperparameter values.
	\par
	To limit computational burden, each within-group baseline function $B_s(X_i)$ is introduced in the outcome equation only after the optimal hyperparameter values $\bar{S}$ and $\lambda_{min}$ are identified. The variables entering each baseline function (in an additive linear specification) are selected using a standard LASSO procedure. The final choice of within-group baseline linear function is validated using the CV procedure described above. The BLP of each resulting CATE is then computed using Eq. (\ref{eq:BLP1}). Finally, inference on the selected within-group CATEs is performed by increasing the total number of splits to 1,000 and using quantile aggregation to generate median values and confidence intervals at the 95\% level \citep{chernozhukov_fisher-schultz_2025}.
	\subsubsection{Adherence Model}\label{sec:443}
	As mentioned in Section \ref{sec:33}, adherence is modeled with a two-part model, with the first part modeling the probability of complete adherence to treatment and the second part modeling the continuous adherence value if adherence is incomplete. For simplicity, each part is modeled using a linear model with a small set of covariates that do not fully overlap across each part. This namely includes~: predicted treatment effectiveness under full adherence ($\hat{\tau}_{i,d_i}(1)$; first part only), clinical center fixed effects (first part only), the degree of hepatic cirrhosis (both parts), ethnicity indicators (both parts), a fixed address indicator (both parts), and gender (both parts). The final choice of covariates in each part of the model is based on statistical significance, increases in adjusted R-squared values, and mentions in the relevant literature \citep{harney_unsuccessful_2025, wilton_real-world_2020}. The number of covariates included in the second part is much lower than in the first part due to the large fraction of patient completely adhering to treatment. I then incorporate predicted adherence values in each outcome model to yield the estimated CATEs with predicted adherence, $\hat{\tau}_{i,d}(\hat{a}(W_i,\hat{\tau}_{i,d}(1)))$, for every pair $(i,d) \in [N] \times \mathcal{D}$. 
	\subsubsection{Feasible Policy Rule and CEAs}\label{sec:444}
	To estimate feasible policy rules under the certainty criterion, I use Eq. (\ref{eq:feasible_c}) and the \textit{rpart} command from the \textbf{rpart} package in R to predict $z^M_i(\hat{\boldsymbol{\mu}},\hat{\Sigma})$ for every $i \in [N]$. The \textit{rpart} command allows the econometrician to fix the maximum depth and the complexity of the tree. For simplicity, I set the complexity of the tree to the smallest degree possible while still being able to predict each group membership in at least one terminal leaf. Missing/unknown values in the matrix of covariates $X$ are handled by surrogate rules following the recommendation of \cite{breiman_classification_2017}.\footnote{This is not the case for other popular packages such as \textbf{policytree}, which does not handle missing values in covariates.} 
	The \textit{rpart} command also directly provides the conditional probability $\text{Pr}[z^M_i(\hat{\boldsymbol{\mu}},\hat{\Sigma}) =s |L_i = k]$, where $L_i$ is the terminal leaf associated with the $i^{th}$ observation, and where $k \in \{1,2,...,|L|\}$ with $|L|$ corresponding to the number of terminal leaves in the tree. Those probabilities are used to obtain the Bayesian version of the feasible policy rule, as shown in Eq. (\ref{eq:feasible_b}). Finally, I use the estimated CATE with either complete or predicted adherence to perform four different CEAs, as described in Section \ref{sec:35}.  Each CEA compares the resulting treatment allocation, $\hat{\pi}^*(\omega)$, to the observed treatment allocation (the \textit{status quo} allocation) for treated patients only using Eq. (\ref{eq:171}). Any treatment allocation obtained with a feasible policy rule is also compared to its ``infeasible" analog (i.e., obtained with $z^M_i(\hat{\boldsymbol{\mu}},\hat{\Sigma})$) using all patients in the cohort.
	\section{Results}\label{sec:5}
	\subsection{Descriptive Statistics}\label{sec:51}
	\begin{table}[t!]
		\centering
		\caption{Descriptive Statistics by Treatment Regimen}
		\label{tab:descriptive_statistics}
		
		\begin{threeparttable}
			
			\setlength{\tabcolsep}{8pt}
			\renewcommand{\arraystretch}{1.05}
			
			\begin{tabular}{lccccc}
				\toprule
				& Untreated & Harvoni & Epclusa & Mavyret & Full Sample \\
				& (1) & (2) & (3) & (4) & (5) \\
				\midrule
				Male
				& 0.65 & 0.80 & 0.68 & 0.62 & 0.70 \\
				& (0.48) & (0.40) & (0.47) & (0.49) & (0.46) \\[0.5mm]
				Age
				& 44.9 & 51.3 & 48.1 & 45.2 & 47.5 \\
				& (9.3) & (8.3) & (10.2) & (10.6) & (9.7) \\[0.5mm]
				Liver Stiffness
				& 9.6 & 10.8 & 8.4 & 7.2 & 9.5 \\
				& (10.7) & (8.2) & (8.2) & (3.8) & (8.8) \\[0.5mm]
				Income Level
				& 2.78 & 3.31 & 3.12 & 3.10 & 3.04 \\
				& (1.41) & (1.76) & (1.55) & (1.76) & (1.58) \\[0.5mm]
				\addlinespace[0.5em]
				Adherence
				& 0.000 & 0.978 & 0.965 & 0.985 & 0.973$^\dagger$ \\
				& (0.000) & (0.090) & (0.149) & (0.121) & (0.123) \\[0.5mm]
				SVR
				& 0.182 & 0.952 & 0.910 & 0.912 & 0.634 \\
				& (0.387) & (0.214) & (0.286) & (0.286) & (0.482) \\
				\midrule
				$N$
				& 487 & 332 & 345 & 68 & 1,232 \\
				\% of Sample
				& 39.5 & 26.9 & 28.0 & 5.5 & 100.0 \\
				\bottomrule
			\end{tabular}
			\begin{tablenotes}[flushleft]
				\item \footnotesize
				\textit{Notes:} Means are reported with standard deviations in parentheses.
				Statistics are based on information collected at treatment for treated patients, or at enrollment for untreated patients (apart cases with spontaneous clearance; see Appendix \ref{app:C}). Male is a binary indicator that is equal to one if the patient is a male. Liver stiffness is measured using FibroScan during the pre-treatment follow-up visit and is expressed in kilopascals (kPa). Monthly income is measured on a categorical scale from 1 to 8, corresponding to monthly income categories ranging from \$0--500 to \$4,001 or more. $^\dagger$Conditional on being treated.
			\end{tablenotes}
			
		\end{threeparttable}
	\end{table}
	Table \ref{tab:descriptive_statistics} shows various descriptive statistics by treatment regimen in the sampled population. All the values presented in Table \ref{tab:descriptive_statistics} are mean values with standard deviations in parentheses. The first row indicates that around 80\% of the patients prescribed Harvoni in the sample are men, compared to less than 70\% for all other treatment regimens. The second row shows that untreated patients are, on average, significantly younger than the average patient in the cohort at the 95\% confidence level (assuming independence between patients). The third row indicates that there exists a slight gradient in terms of average liver stiffness across treatment regimens, with the patient prescribed Harvoni having, on average, a higher liver stiffness than the one prescribed Epclusa, who has a, average higher liver stiffness than the one prescribed Mavyret.\footnote{A normal liver stiffness corresponds to a FibroScan measure that lies between 2 and 7 kilopascals (kPa). A FibroScan measure above 19 kPa is usually associated with liver cirrhosis.} The fourth row also shows that untreated patients have significantly lower income levels than the sample average, which can be explained by financial constraints preventing access to DAA treatments before the advent of public coverage by all Canadian provinces in 2018 (except Newfoundland). Typically, such differences in observed covariates likely indicate the presence of unconditional selection into treatment. Under Assumption \ref{ass:3}, conditioning on the unobserved group membership leads to plausible exogenous treatment assignments for every patient in the sample.
	\par
	The last two lines of Table \ref{tab:descriptive_statistics} show the average adherence level and the proportion of patients reaching SVR by treatment regimen. The fifth row of Table \ref{tab:descriptive_statistics} shows that adherence levels are very high for all treatment regimens. This is because adherence is perfect for more than 90\% of all treated patients for every treatment regimen. Only one patient prescribed Mavyret did not completely adhere to treatment, with an adherence value virtually equal to zero because the patient did not initiate the treatment. The last row of Table \ref{tab:descriptive_statistics} also indicates that SVR is achieved slightly more frequently under Harvoni than the other two regimens, even if patients prescribed Harvoni are somewhat older and have worse liver conditions than the average patient with any other treatment option.  Note also that the first column of Table \ref{tab:descriptive_statistics} indicates a spontaneous clearance rate of 18\% in the untreated population.
	\subsection{Classification Procedure and Outcome Models}\label{sec:52}
	\begin{table}[t!]
		\centering
		\caption{Results of the Classification and Regression Procedures}
		\label{tab:classification_regression}
		
		\begin{threeparttable}
			
			\setlength{\tabcolsep}{8pt}
			\renewcommand{\arraystretch}{1.05}
			
			\begin{tabular}{lccccc}
				\toprule
				& & \multicolumn{2}{c}{AUC-ROC}
				& \multicolumn{2}{c}{Maximum Balanced Accuracy} \\
				\cmidrule(lr){3-4} \cmidrule(lr){5-6}
				\multirow{2}{2.3cm}{Number of Groups, $\bar{S}$}
				& $\lambda^*_{\min}$
				& In-Sample
				& OOS
				& In-Sample
				& OOS \\
				& (1) & (2) & (3) & (4) & (5) \\
				\midrule
				
				1 & --   & 0.9006 & 0.8923 & 0.8955 & 0.8953 \\
				2 & 0.10 & 0.9076 & 0.8958 & 0.8955 & 0.8949 \\
				3 & 0.10 & 0.9334 & 0.9081 & 0.8991 & 0.9000 \\
				4 & 0.75 & 0.9235 & 0.9113 & 0.9018 & 0.9028 \\
				5 & 0.85 & 0.9284 & 0.9138 & 0.9029 & 0.9026 \\
				6 & 0.30 & 0.9345 & 0.9142 & 0.9030 & 0.9024 \\
				
				\bottomrule
			\end{tabular}
			
			\begin{tablenotes}[flushleft]
				\footnotesize
				\item \textit{Notes:} AUC-ROC denotes the area under the receiver operating
				characteristic curve. Maximum Balanced Accuracy is the maximum average of specificity and sensitivity across classification thresholds. OOS denotes out-of-sample.
				OOS performance metrics are based on two-fold cross-validation with
				200 data splits.
			\end{tablenotes}
			
		\end{threeparttable}
	\end{table}
	Table \ref{tab:classification_regression} shows the results of the classification and regression procedures both in- and out-of-sample (OOS), where the latter is based on the CV procedure described in Section \ref{sec:442}. The first row of Table \ref{tab:classification_regression} presents the results of the regression procedure with only one group, therefore without using the weighted K-means algorithm. The first column of Table \ref{tab:classification_regression} shows the optimal minimum eigenvalue $\lambda^*_{min}$, which corresponds to the value of $\lambda_{min}$ that is associated with the largest OOS maximum balanced accuracy (Column (6)). Because it is based on the optimal point in the ROC curve, I consider the OOS maximum balanced accuracy to be a more appropriate performance metric than the OOS AUC-ROC to select the optimal hyperparameters. Nonetheless, the OOS AUC-ROC metric remains informative about the general OOS performance of binary models, which is why it is also shown in Column (3) of Table \ref{tab:classification_regression}.
	\par
	Column (5) of Table \ref{tab:classification_regression} shows that the OOS balanced accuracy is maximized when $\bar{S}=4$ and $\lambda_{min} = 0.75$, with a maximum balanced accuracy of 0.9028. This corresponds to a sensitivity of 0.8972 and a specificity of 0.9084, where both are attained with a classification threshold of 0.646. Although increasing $\bar{S}$ yields larger AUC-ROC values, the differences vanish when the number of splits increases from 200 to 1000, and parsimonious models are generally preferred over more complex ones when OOS prediction performance is similar across models. A more thorough search across the hyperparameter space could eventually yield marginally better performance metrics, but at the cost of greater computational overhead.
	\par
	The results shown in columns (2) to (5) of Table \ref{tab:classification_regression} are all based on the following LPM
	\begin{align}\label{eq:14}
		\text{SVR}_i = \beta_{0,z^M_i(\hat{\boldsymbol{\mu}},\hat{\Sigma})} + \sum_{d \in \{\text{EPC},\text{MAV},\text{HAR}\}} \tau_{d}(1,z^M_i(\hat{\boldsymbol{\mu}},\hat{\Sigma})) \times A_{i,d} + B_{z^M_i(\hat{\boldsymbol{\mu}},\hat{\Sigma})}(X_i) + \epsilon_{i},
	\end{align}
	where $z^M_i(\hat{\boldsymbol{\mu}},\hat{\Sigma})$ represents the group membership of patient $i$ produced by the weighted K-means algorithm, $\beta_{0,s}$ is the intercept value for the $s^{th}$ group (i.e., the probability of spontaneous clearance in the $s^{th}$ group), $\tau_{d}(1,s)$ is defined as in Assumption \ref{ass:1} and represents the CATE of treatment option $d$ for the $s^{th}$ group under complete adherence, $B_s(X_i)$ refers to the baseline function for the $s^{th}$ group, and $\epsilon_{i} \sim N(0,\sigma^2_{\epsilon})$ with $\sigma^2_{\epsilon} > 0$. To reduce computational burden, the results shown in columns (2) to (5) of Table \ref{tab:classification_regression} are all based on the restriction that $B_s(X_i) = 0$ for all values of $s \in \mathcal{S}$. However, to improve OOS goodness-of-fit in the selected model with $\bar{S}=4$ and $\lambda_{min} = 0.75$, I use a standard LASSO procedure to create four different within-group baseline functions. Adding these functions to each group yields an OOS AUC-ROC of 0.9162 and an OOS maximum balanced accuracy of 0.9029 (sensitivity=0.8985; specificity=0.9074) across 200 data splits.\footnote{In fact, only the largest group is augmented with a baseline function; adding a baseline function to the other groups yields no improvements in any of the OOS performance metrics.} Because both OOS performance metrics increase when adding the baseline functions, all further results presented in this section are based on such ``augmented" models.
	\par
	To compare these results to those obtained with the typical approach used in the PL literature, I use the \textit{multi\_arm\_causal\_forest} command from the \textbf{grf} package to generate estimates of CATEs in a multi-action PL setting based on causal random forests, which is one of the most popular approaches for PL in the econometric literature \cite{wager_estimation_2018,athey_policy_2021,tibshirani_generalized_2026}. Specifically, it uses a multivariate extension of the ``R-learner" approach suggested by \cite{nie_quasi-oracle_2021} to estimate a CATE for every treatment option and every observation in the sample. However, two elements limit the comparison between the proposed approach and the one using causal random forests. First, the \textit{multi\_arm\_causal\_forest} command does not handle semi-continuous treatment effect variables. To bypass this issue, I enforce $A_{i,D_i} = 1$ if and only if the observed adherence for assigned treatment $D_i$ is higher than 0.5. Second, it is not possible to generate OOS predicted outcome values with the current \textit{multi\_arm\_causal\_forest} command. Instead, the command only produces out-of-bag (OOB) predictions, which correspond to predictions of $Y_i$ using only the trees that did not use the $i^{th}$ observation for training. It remains unclear under which conditions these two measurements can be compared to each other for different methods \citep{krupkin_prediction_2024}.
	\par
	Notwithstanding these differences, applying the \textit{multi\_arm\_causal\_forest} command to the CCC dataset yields a maximum OOB balanced accuracy of 0.9071, corresponding to a maximum OOB AUC-ROC of 0.9232.\footnote{I obtain these results using 2,000 trees per forest and searching across 2,000 different seeds for the random number generator. The number of candidate variables tried for each split in each tree is left at the default, theoretically optimal value of $\sqrt{p} + 20$. Changing this hyperparameter induces substantial variation in both OOB performance metrics.} These large values are caused by the increased sensitivity of the random forest approach; the maximum OOB balanced accuracy of 0.9071 corresponds to a sensitivity of 0.9449 and a specificity of 0.8692 at a classification threshold of 0.315.\footnote{At a threshold of 0.5, which is a natural threshold for binary classification, the weighted K-means features an OOS sensitivity of 0.9040 and an OOS specificity of 0.9004, while the OOB sensitivity and specificity of the causal random forest equal 0.9078 and 0.8803, respectively.} If the cost (in terms of both welfare and financial costs) of false positives is high, then using the causal random forest approach may lead to suboptimal treatment allocation compared to a more specific approach. However, the two approaches yield (unconstrained) welfare-maximizing treatment choices that are identical across 96.3\% of the sample, as shown in Table \ref{tab:confusion_matrix}.
	\par
	\begin{table}[t!]
		\centering
		\caption{Unconstrained Welfare-maximizing Treatment Choices Based on Estimated CATEs}
		\label{tab:confusion_matrix}
		\begin{threeparttable}
			\setlength{\tabcolsep}{12pt}
			\renewcommand{\arraystretch}{1.15}
			\begin{tabular}{lcc}
				\toprule
				& \multicolumn{2}{c}{Weighted K-means} \\
				\cmidrule(lr){2-3}
				Causal Random Forest
				& Harvoni & Mavyret \\
				\midrule
				Harvoni & 1114 & 45 \\
				Mavyret & 0 & 73 \\
				\bottomrule
			\end{tabular}
		\end{threeparttable}
	\end{table}
	\par 
	Table \ref{tab:confusion_matrix} also shows that, without cost or feasibility constraints on the policy rule, the Epclusa treatment regimen is dominated by the other two regimens. This can also be seen in Figure \ref{fig:res1}, which shows the estimated CATEs $\hat{\tau}_{i,d}(1)$ using Eq.(\ref{eq:14}) for the selected combination of hyperparameters (i.e., $\bar{S}=4$ and $\lambda_{min}= 0.75$). For comparison, the last set of columns in Figure \ref{fig:res1} shows the estimated ATEs obtained with the entire sample (i.e., with a single group) and no baseline function in the specification. The error bars correspond to heteroskedasticity-robust confidence intervals at the 95\% level. The first set of bars on the left indicates that patients in the first group have around an 80\% chance of spontaneously clearing the HepC virus. This result could be leveraged to reduce total prescription costs in the future, especially if waiting to confirm spontaneous clearance induces no substantial costs (which is unlikely in practice for communicable disease). Within-group descriptive statistics also show that the average age in Group 1 is significantly lower than the average age in Groups 3 and 4, and the entire sample. However, there is no significant difference between the average values observed in Group 1 and those in other groups for male proportion, liver stiffness, income level, and degree of adherence (see Table \ref{tab:A1} in Appendix \ref{app:E}).
	\par
	\begin{figure}[t!]
		\centering
		\includegraphics[width=13.5cm]{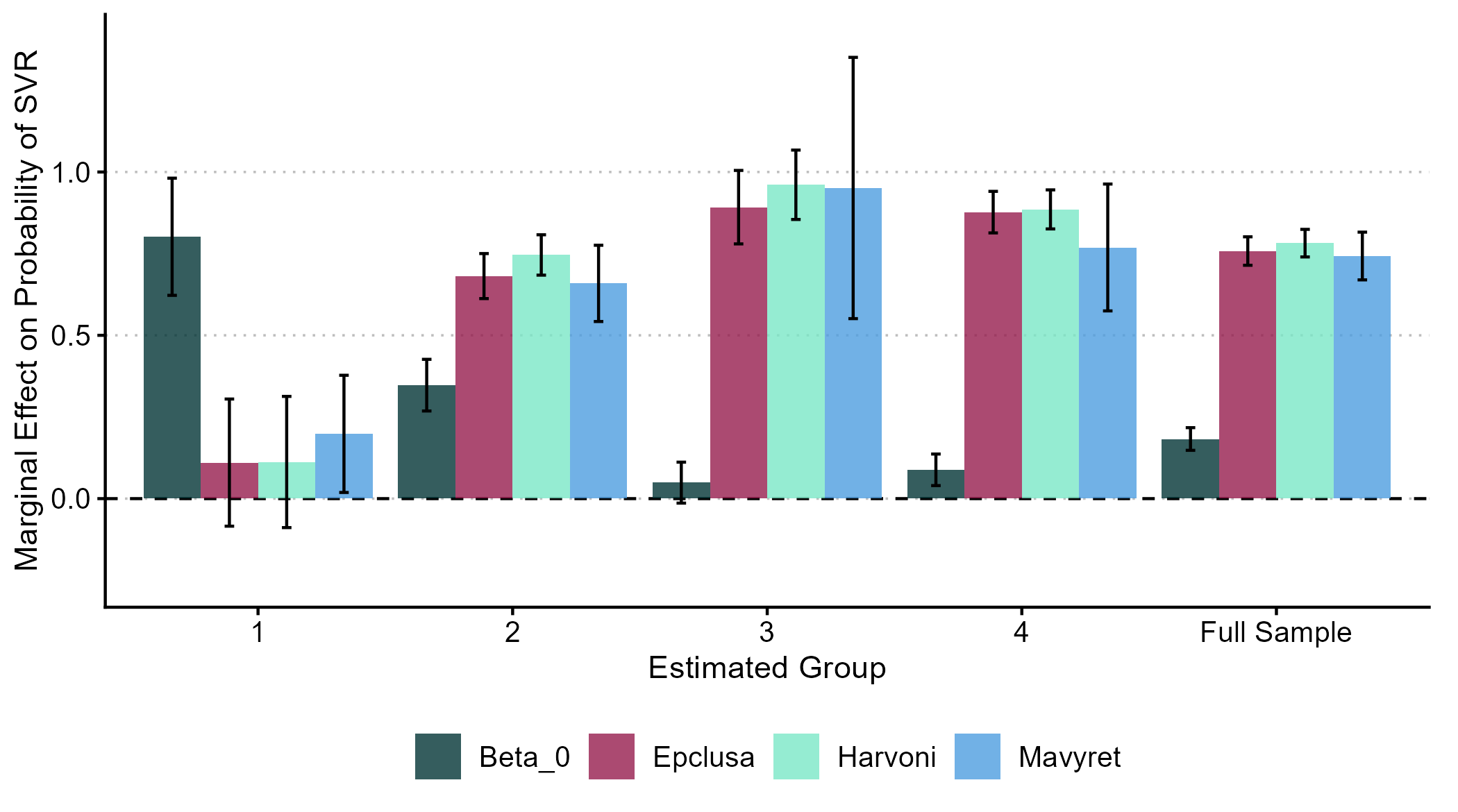}
		\caption{Estimated coefficients and confidence intervals for each within-group LPM and the full sample. Beta\_0 corresponds to the within-group intercept.}\label{fig:res1}
	\end{figure}
	\par
	For comparison, I also compute DR scores of the estimated CATEs for both approaches, which is a common practice in the PL literature \citep{athey_policy_2021}. For the causal random forest approach, I obtain the DR scores via the \textit{get\_scores} command of the \textbf{grf} package. Welfare-maximizing treatment choices based on DR scores differ from the CATEs' welfare-maximizing treatment choices across 38.1\% of the sample. Among these 472 ``divergent" DR treatment choices, 93 (19.7\%) of them indicate that the null treatment regimen is welfare-maximizing, although all corresponding CATEs are significantly greater than zero. Moreover, all 93 DR null treatment choices correspond to either spontaneous clearance (89 patients) or clearance with adherence below 0.5 (4 patients). Reciprocally, all spontaneous clearance cases observed in the sample are associated with DR scores that imply the null treatment regimen is welfare-maximizing. Such results indicate that every patient who experienced spontaneous clearance would have been less likely to reach SVR had they initiated any of the considered treatment options. The probability of such counterfactuals being all true is obviously very slim, unless one is willing to assume that the true marginal effects of the considered DAAs are negative only for patients who experienced spontaneous clearance.
	\par
	Similar counterfactual results can also be observed with DR scores based on the CATEs produced by the weighted K-Means algorithm. To compute these scores, I use the augmented inverse probability weighting (AIPW) estimator\footnote{See \cite{zhou_offline_2023} and \cite{sverdrup_qini_2025} for the formal expression of the doubly-robust scores in the multi-action framework.}, where the propensity score $\mathbb{P}[D_i = d|S_i=s]$ is estimated by the within-group average of observed adherence for every $d \in \mathcal{D}$. The resulting DR scores indicate that Mavyret is the welfare-maximizing treatment for every patient who experienced spontaneous clearance, even though the estimated CATEs imply that Mavyret is less effective than Epclusa and Harvoni among 80.1\% of the spontaneous clearance cases. Moreover, the DR scores indicate that the null treatment regimen is the welfare-maximizing treatment option for every patient who fully adhered to any DAA regimen but did not reach SVR. As with the causal random forest, such a result implies that the marginal effect of all three DAAs in this group of patients (which is different from the group experiencing spontaneous clearance) would have to be negative in order for such a result to be valid.\footnote{This arises because all DAA failures happen in groups where the probability of not being treated is greater than the probability of being treated with any treatment option, thereby leading to negative DR scores for all treatment options. Although such differences in propensity scores can be caused by positive selection into treatment, this is unlikely to be the case here since there is no common contraindication that would prevent Canadian physicians from prescribing all three considered DAAs to any given patient \citep{martel-laferriere_prise_2022-1}.} Therefore, such results question the applicability and benefits of using DR scores for PL in finite-sample, multi-action frameworks where both true propensity scores and outcome models are unknown.
	\subsection{Best Linear Predictor of the Estimated CATEs}\label{sec:521}
	\begin{figure}[t!]
		\centering
		\includegraphics[width=13.5cm]{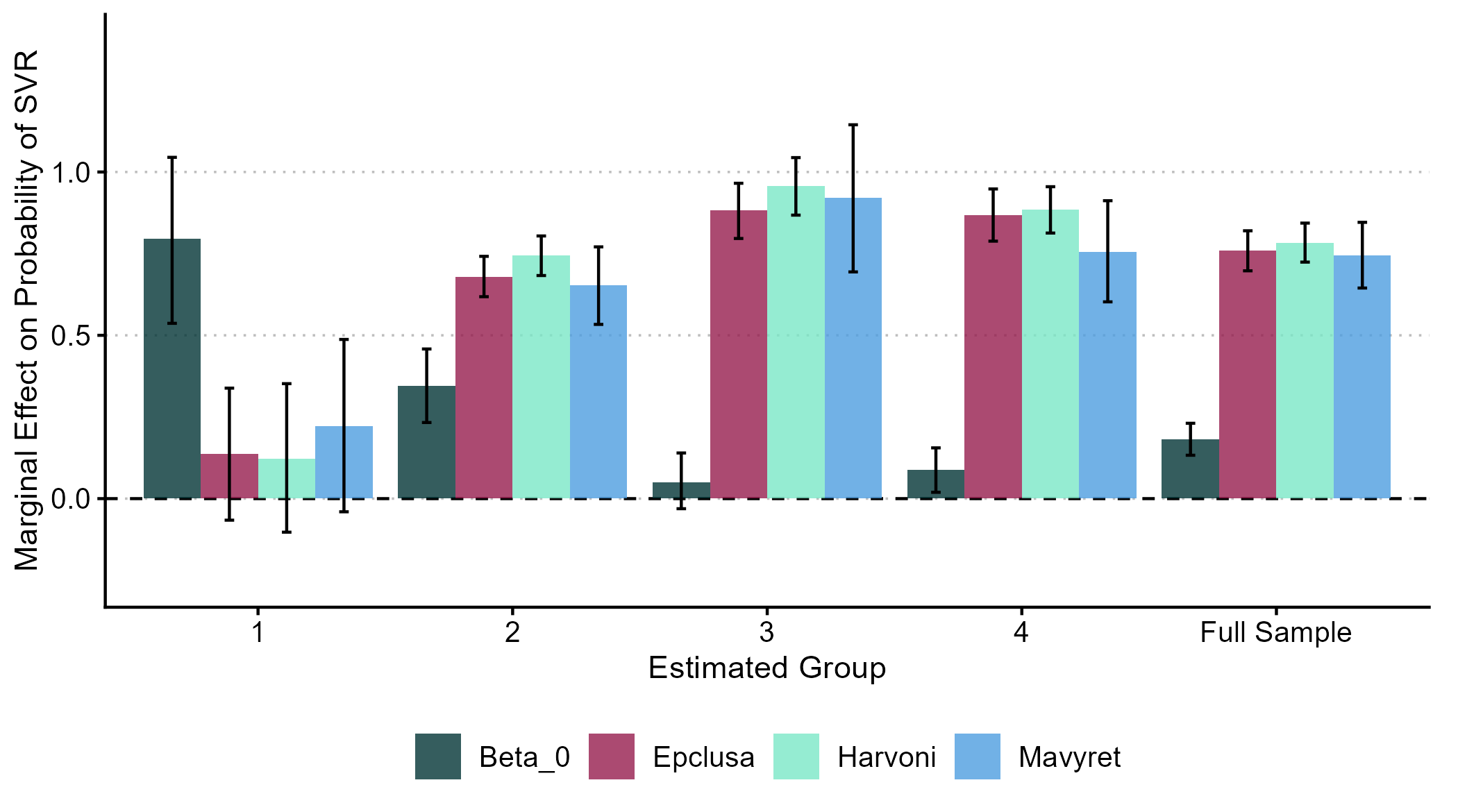}
		\caption{BLP median coefficients and confidence intervals for each within-group LPM. Beta\_0 corresponds to the within-group intercept.}\label{fig:res2}
	\end{figure}
	Instead of using DR scores, I compute the BLP of every estimated CATE using Eqs. (\ref{eq:BLP}) and ­(\ref{eq:BLP1}) and 2-fold sample splitting with 1,000 data splits. I then use quantile aggregation across splits to obtain the median-aggregated BLP and confidence intervals of every estimated CATE, as described in \cite{chernozhukov_fisher-schultz_2025}. Because the weighted K-means algorithm provides CATE estimates that are often similar but more easily interpretable than causal random forests, I focus exclusively on the BLPs of the CATEs produced by the former algorithm, which are shown in Figure \ref{fig:res2}. Comparing Figures \ref{fig:res1} and \ref{fig:res2}, one can easily notice that the BLPs of the CATEs are almost identical to the CATEs themselves~: the largest absolute variation between each estimated CATE and its BLP equals 3.1 pp and corresponds to the marginal effect of the Mavyret treatment regimen in Group 3 that decreases from 0.951 to 0.920. All other absolute variations between each estimated CATE and its BLP are below 3 pp. Note also that the confidence intervals shown in Figure \ref{fig:res2} may be too narrow because of the use of approximate inference on the BLPs, as explained in Section \ref{sec:32}.
	\par
	Table \ref{tab:table_BLP} shows the median value of the $\hat{\alpha}_{2,d}$ coefficient from Eq. (\ref{eq:BLP1}) across all data splits, which can be interpreted as a bias/noise-reducing coefficient for each treatment option other than the null. If this coefficient is significantly greater than zero, this indicates the presence of heterogeneity in the true CATEs and that the estimated CATEs are relevant predictors of this heterogeneity. Moreover, if $\hat{\tau}_{d}(a,\hat{S}_i,\hat{S}_i)$ is an unbiased estimator the true CATE $\tau_{d}(a,s)$ under Assumption \ref{ass:1}, then ${\alpha}_{2,d}=1$ \citep{chernozhukov_fisher-schultz_2025}. The results portrayed in Table \ref{tab:table_BLP} clearly show that there is heterogeneity in the true CATEs and that the CATEs estimated by the K-means algorithm are relevant predictors of it.
	\par
	\begin{table}[t!]
		\centering
		\caption{BLP Bias-reducing Coefficient and p-Value}
		\label{tab:table_BLP}
		\begin{threeparttable}
			\setlength{\tabcolsep}{10pt}
			\renewcommand{\arraystretch}{1.05}
			\begin{tabular}{lcc}
				\toprule
				&Median $\hat{\alpha}_{2,d}$ &Median p-value \\
				\midrule
				Harvoni
				& 0.994 & 7.11E-09 \\
				Epclusa
				& 0.956 & 2.62E-09 \\
				Mavyret
				& 0.931 & 0.00130 \\
				\bottomrule
			\end{tabular}
			\begin{tablenotes}[flushleft]
				\footnotesize
				\item \textit{Notes:} All median values are based on a thousand 2-fold sample splittings.
			\end{tablenotes}
		\end{threeparttable}
	\end{table}

	\subsection{Adherence to Treatment}\label{sec:53}
	Table \ref{tab:adherence_model} shows the estimated coefficients for each part of the adherence model described in Section \ref{sec:443}, using only treated patients. The first column of Table \ref{tab:adherence_model} shows the estimated coefficient for the binary part of the model, where I removed the predicted effectiveness of the Epclusa treatment regimen because of its very high collinearity with the predicted effectiveness variables. I assessed the degree of multicollinearity among the included variables in each part of the model with the variance inflation factor (VIF $<$ 5). The first column of Table \ref{tab:adherence_model} shows that predicted treatment effectiveness of Mavyret (when prescribed), the degree of liver stiffness, and quality of life are significantly and positively associated with the probability of completely adhering to the assigned treatment. This is also true for people with a fixed address (+13.8 pp), who are either Metis (+5.8 pp) or Asian (+5.7 pp), and living in a long-term care facility or in a rehabilitation center (+5.3 pp). The fixed effects associated with the clinical centers 3, 4, 7, 10, 11, 12, 16, 18, and 19 are also positive and statistically significant at the 95\% confidence level. Besides, those associated with clinical centers 5, 8, and 15 are \textit{negative} and significant at the 95\% level (results not shown). Such results warrant further investigation to better understand why patients treated in different clinical centers are more or less likely to fully adhere to treatment. On the other hand, living alone is the only variable that is negatively and significantly associated with the probability of completely adhering to treatment at the 95\% confidence level.
	\begin{table}[t!]
		\small
		\centering
		\caption{Estimated Coefficients from the Two-Part Adherence Model}
		\label{tab:adherence_model}
		
		\begin{threeparttable}
			
			\setlength{\tabcolsep}{10pt}
			\renewcommand{\arraystretch}{1.05}
			
			\begin{tabular}{lcc}
				\toprule
				& Binary Part & Continuous Part \\
				& (1) & (2) \\
				\midrule
				
				Intercept
				& 0.7513*** & 0.5942*** \\
				& (0.0574) & (0.1275) \\
				
				Predicted effectiveness, Harvoni
				& 0.0315 & -- \\
				& (0.0280) & \\
				
				Predicted effectiveness, Mavyret
				& 0.0924** & -- \\
				& (0.0311) & \\
				
				Male (Y/N)
				& -0.0366 & 0.1497$^\dagger$ \\
				& (0.0260) & (0.0789) \\
				
				Liver stiffness (FibroScan score)
				& 0.0028* & -- \\
				& (0.0012) & \\
				
				History of liver cirrhosis (Y/N)
				& -- & -0.2031* \\
				& & (0.0858) \\
				
				Quality of life (EQ-VAS score)
				& 0.0016** & -- \\
				& (0.0006) & \\
				
				CD4 lymphocytes Count
				& -0.0001$^\dagger$ & -- \\
				& (0.00004) & \\
				
				Sexually-transmitted HIV infection (Y/N)
				& -0.0628$^\dagger$ & -- \\
				& (0.0334) & \\
				
				Drug-transmitted HCV infection (Y/N)
				& -- & -0.1325* \\
				& & (0.0578) \\
				
				Having a fixed address (Y/N)
				& 0.1378** & 0.1931$^\dagger$ \\
				& (0.0455) & (0.1012) \\
				
				Living in a house, room or an apartment (Y/N)
				& -- & -0.1456* \\
				& & (0.0614) \\
				
				Living in long-term care or rehabilitation center (Y/N)
				& 0.0530* & -- \\
				& (0.0231) & \\
				
				Living alone (Y/N)
				& -0.0777* & -- \\
				& (0.0308) & \\
				
				Sharing a room or an appartment (Y/N)
				& -0.0666$^\dagger$ & -- \\
				& (0.0397) & \\
				
				Asian ethnicity (Y/N)
				& 0.0572*** & -- \\
				& (0.0110) & \\
				
				Metis ethnicity (Y/N)
				& 0.0583* & -- \\
				& (0.0248) & \\
				
				Hispanic ethnicity (Y/N)
				& -0.0975 & -- \\
				& (0.0638) & \\
				
				Black ethnicity (Y/N)
				& -0.0806 & -- \\
				& (0.0747) & \\
				
				Ever been in prison (Y/N)
				& -0.0447 & -0.1616$^\dagger$ \\
				& (0.0302) & (0.0830) \\
				
				Ever injected drugs (Y/N)
				& -0.0229 & -- \\
				& (0.0267) & \\
				
				Working part-time (Y/N)
				& -- & -0.3409*** \\
				& & (0.0793) \\
				
				\midrule
				
				Clinical Center Fixed Effects
				& Yes & No \\
				$N$
				& 745 & 63 \\
				$R^2_a$
				& 0.072 & 0.439 \\
				\bottomrule
			\end{tabular}
			
			\begin{tablenotes}[flushleft]
				\footnotesize
				\item \textit{Notes:} Heteroskedasticity-robust standard errors are reported
				in parentheses ($HC1$), clustered at the clinical-center level for the
				binary component. Binary covariates are indicated by (Y/N). Clinical center fixed effects are not reported for brevity. $\dagger$ $p<0.10$; * $p<0.05$;** $p<0.01$; *** $p<0.001$.
			\end{tablenotes}
		\end{threeparttable}
	\end{table}
	\par
	The second column of Table \ref{tab:adherence_model} shows the estimated coefficients for the continuous part of the two-part model, where being male and having a fixed address have a positive and significant effect on the degree of adherence to treatment at the 90\% confidence level. In contrast, all other included variables (apart from the intercept) in the continuous part present a negative effect on the degree of adherence at the 90\% confidence level. While the second column of Table \ref{tab:adherence_model} shows that having a history of liver cirrhosis is associated with a significant drop in observed adherence to treatment, the first column shows that a high degree of liver fibrosis increases the probability of fully adhering to treatment. Such an apparent contradiction may be the result of healthier patients having less incentive to fully adhere to treatment, especially if the treatment is more effective in patients with no current liver cirrhosis. The large negative effect of working part-time on adherence is because only three patients with adherence levels below 0.2 work part-time (all three with Epclusa). The other variables are included in the continuous part of the model because their respective p-values are lower than 0.1 and because those variables are suspected to affect the level of adherence according to the literature on DAAs \citep{harney_unsuccessful_2025, darvishian_loss_2020, wilton_real-world_2020, jiang_hcv_2023}. Although potentially relevant, no ethnicity indicators were included in the continuous part because of a lack of variation among the 63 observations with imperfect adherence.
	\subsection{Feasible Policy Rule and Minimax-Regret Treatment}\label{sec:54}
	The structure of the decision tree used to predict group memberships is presented in the appendix and was provided directly by the \textit{rpart.plot} command of the \textbf{rpart} package in R. The selected decision tree contains 25 terminal leaves and has a depth of 7 (shown in Appendix \ref{app:E}). This tree shows that group memberships can be predicted using the results of various blood tests (e.g., blood levels of gamma-glutamyl Transferase (GGT), alpha-fetoprotein (AFP), total and direct bilirubin, insulin, hemoglobin, aspartate transferase (AST), absolute neutrophils, and a few others) and the degree of liver stiffness (based on FibroScan measures). Using this tree to predict the K-means group memberships $z^M_i(\hat{\boldsymbol{\mu}},\hat{\Sigma})$ yields a misclassification rate of 19.72\%. Increasing the tree depth to 8 and 9 reduces the misclassification rate to 18.18\% and 17.37\%, respectively. For simplicity, all feasible policy rules are based on the 7-depth tree because it is the least complex one that can predict every K-means group membership.
	\par
	Using the conditional probabilities $\text{Pr}[z^M_i(\hat{\boldsymbol{\mu}},\hat{\Sigma}) =s |L_i = k]$ that are provided by the \textit{rpart} command for every pair $(s,k) \in \mathcal{S} \times \{1,...,25\}$, it is possible to compute the Bayesian version of the feasible policy rule induced by the tree shown in Appendix \ref{app:E}. Because there is only one group where the largest median BLP coefficient corresponds to the Mavyret regimen and only one terminal leaf corresponding to this group, the treatment allocation generated by Eq. (\ref{eq:feasible_b}) implies that only two patients in the sample should be treated with Mavyret, whereas all others should be treated with Harvoni. Increasing the complexity of the tree (through its depth and/or complexity parameter) increases the number of patients for whom the Mavyret treatment regimen is welfare-maximizing up to 118, which corresponds to the total number of patients in Group 1 based on $z^M_i(\hat{\boldsymbol{\mu}},\hat{\Sigma})$.
	\par
	Finally, based on the median BLP coefficients shown in Figure \ref{fig:res2}, it is possible to show that Epclusa corresponds to the minimax-regret optimal treatment regimen, even though it is dominated by the other regimens. This means that the Epclusa treatment regimen performs relatively well across all groups, with a maximum regret (or welfare loss) of 8.6 pp -- compared to 10.1 and 12.8 pp for Harvoni and Mavyret, respectively --, which occurs when the patient truly belongs to Group 1 and Epclusa is prescribed instead of Mavyret. However, minimax-regret optimality depends on the effectiveness of every considered alternative, implying that Epclusa may not remain minimax-regret optimal if other treatment options are added to the analysis.
	\subsection{Cost-Effectiveness Analyses}\label{sec:55}
	\subsubsection{K-means Group Memberships under Perfect Adherence}\label{sec:551}
	\begin{figure}[t!]
		\centering
		\subfigure[Total health benefits (dashed line; left axis) and total prescription costs (solid line; right axis)]{
			\includegraphics[width=13.75cm]{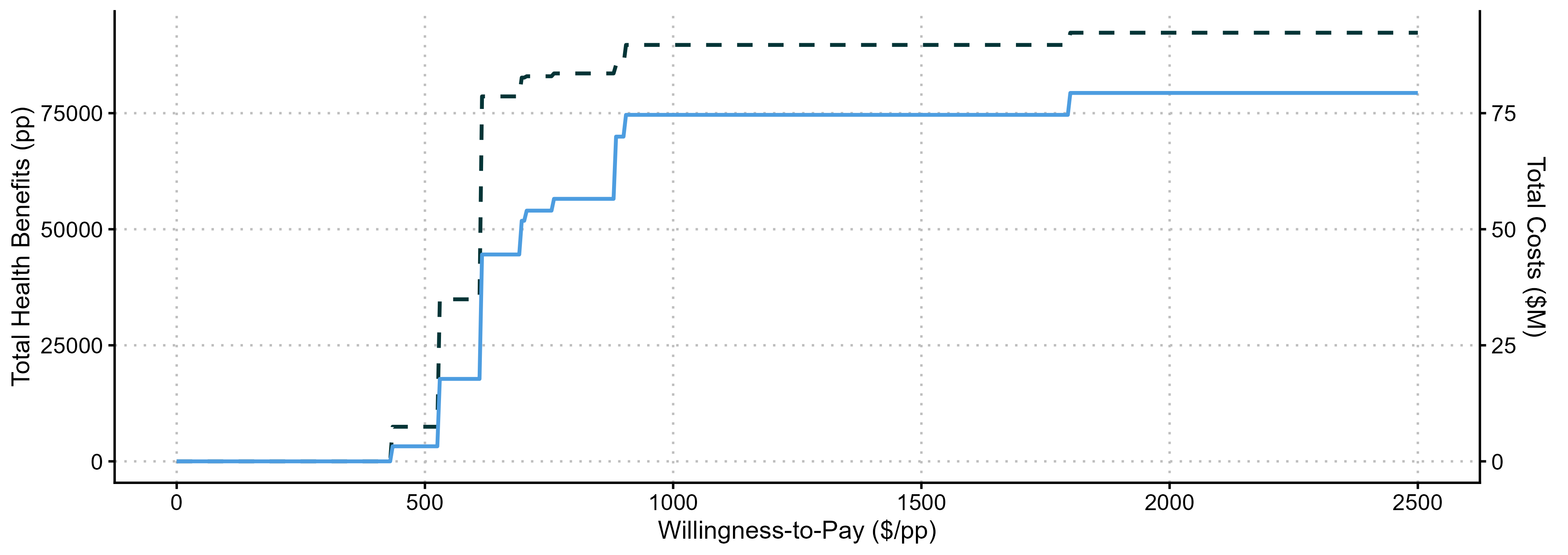}
		}
		\subfigure[Optimal treatment allocation]{
			\includegraphics[width=13.5cm]{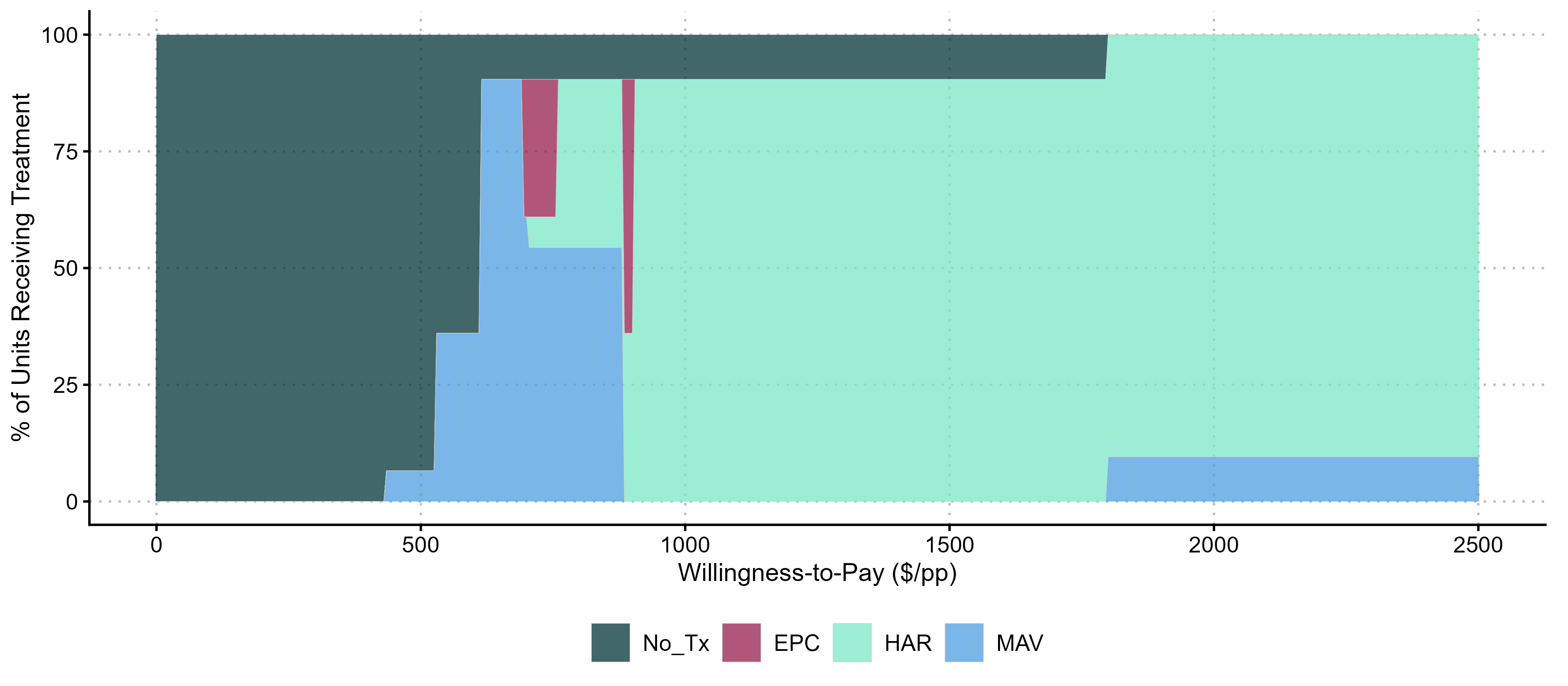}
		}
		\subfigure[Comparison with the status quo allocation]{
			\includegraphics[width=13.75cm]{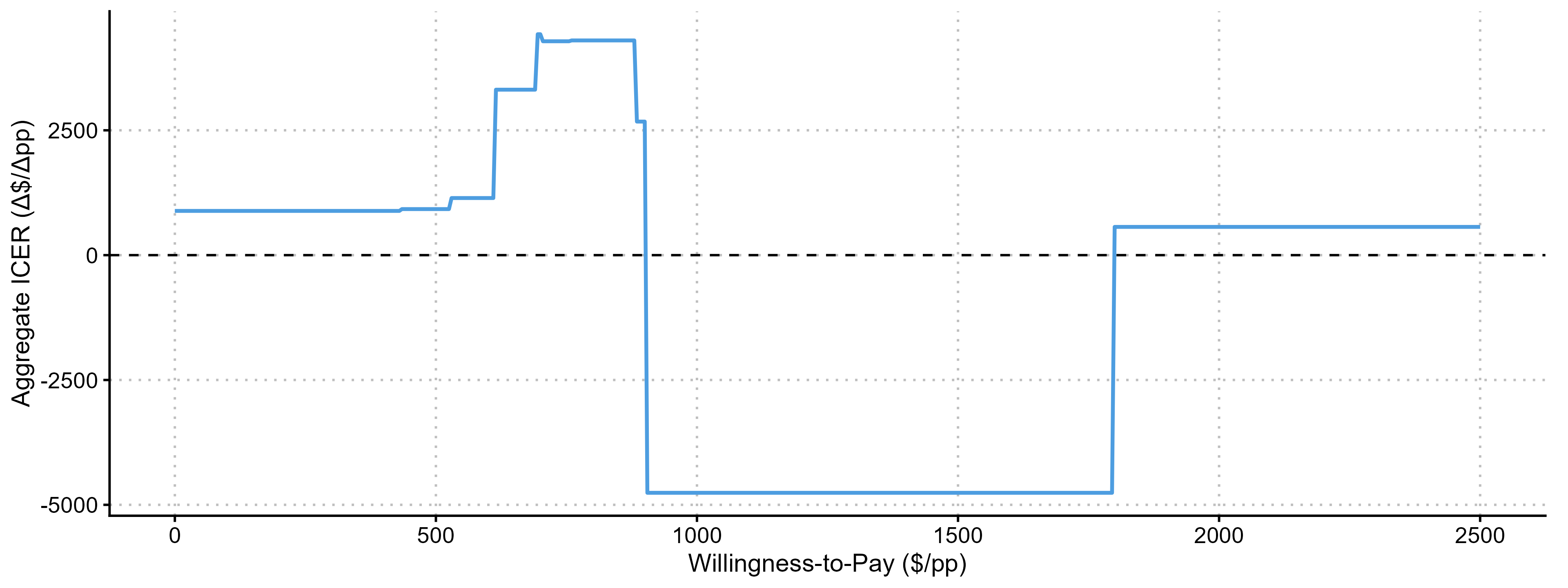}
		}
		\caption{Results of the cost-effectiveness analysis with perfect adherence to treatment and K-means group memberships, $\mathbf{z}^M(\hat{\boldsymbol{\mu}},\hat{\Sigma})$.\label{fig:3}}
	\end{figure}
	Figure \ref{fig:3} shows the results of the CEA described in Section \ref{sec:35} under perfect adherence for every treatment option and every patient in the sample when $\hat{S}_i = z^M_i(\hat{\boldsymbol{\mu}},\hat{\Sigma})$. Panel (a) of Figure \ref{fig:3} shows the total health benefits on the left axis, represented by the sum of the marginal effects (in total pp) across every patient's optimal treatment choice, and the corresponding total prescription costs on the right axis (in CAN\$M), both as a function of the policymaker's WTP. Panel (b) of Figure \ref{fig:3} shows the corresponding optimal treatment allocation as the percentage of patients in the sample that receive each treatment option as a function of WTP. Finally, panel (c) of Figure \ref{fig:3} shows the value of the aggregate ICER when the benchmark corresponds to the status quo, as described by Eq. (\ref{eq:171}). Recall that Panel (c) of Figure \ref{fig:3} is based exclusively on treated patients, given the nature of the benchmark allocation.
	\par
	Figure \ref{fig:3} shows that most of the health benefits are generated when the WTP surpasses CAN\$615 per additional pp. Indeed, 85.1\% of the maximum health benefits are generated when the WTP is equal to CAN\$615/pp, with the corresponding costs amounting to 56.1\% of the maximum total prescription costs (CAN\$M44.6 vs CAN\$M79.3). Panel (b) of Figure \ref{fig:3} also shows that 90.4\% of all patients in the sample are optimally treated with Mavyret when WTP is between CAN\$615/pp and CAN\$690/pp. Panel (b) also shows that all 118 patients belonging to Group 1 are left untreated when the WTP is below CAN\$1,800/pp. When WTP goes beyond CAN\$1,800/pp, then optimal allocation implies that 90.4\%, 9.6\%, and 0.0\% of all patients in the sample are prescribed Harvoni, Mavyret, and Epclusa, respectively. Such an allocation leads to the maximum health benefits achievable under perfect adherence. At last, panel (c) of \ref{fig:3} shows that optimal treatment allocation leads to reduced total costs (savings of CAN\$M2.45) and higher health benefits (+514.6 pp) compared to the status quo when WTP lies between CAN\$905/pp and CAN\$1,800/pp. Here, the negative values of the aggregate ICER correspond to negative total costs (i.e., savings) divided by positive total health benefits, therefore corresponding to an allocation that dominates the status quo at the aggregate level.
	\par
	To give a better sense of the magnitude of the health benefits compared to the cost of treatment, it is possible to compare the total cost of a given treatment option per expected gains in quality-adjusted life years (QALYs). For instance, if we assume that prescribing Mavyret induces a marginal effect of 75 percentage points on the probability of reaching SVR, and that being cured of HCV leads to an increase in quality of life between 1.3 and 2 QALYs \citep{mattingly_changes_2020}, such a hypothetical scenario translates into a total cost that ranges from CAN\$28,500 to CAN\$43,850 per additional QALY given the total cost of the Mavyret treatment regimen (CAN\$40,000.00 for 8 weeks). The upper bound is below the standard conservative threshold of CAN\$50,000 per QALY. In contrast, the list price of Harvoni is CAN\$67,000.00 for 12 weeks. This leads to a total cost that lies between CAN\$41,875 and CAN\$64,425 per additional QALY if we assume a marginal effect of 80 percentage points on the probability of reaching SVR, which is still relatively low. Obviously, such estimates represent averages that cannot be applied to single patients given the presence of TEH.
	\begin{figure}[t!]
		\centering
		\subfigure[Optimal treatment allocation]{
			\includegraphics[width=13.4cm]{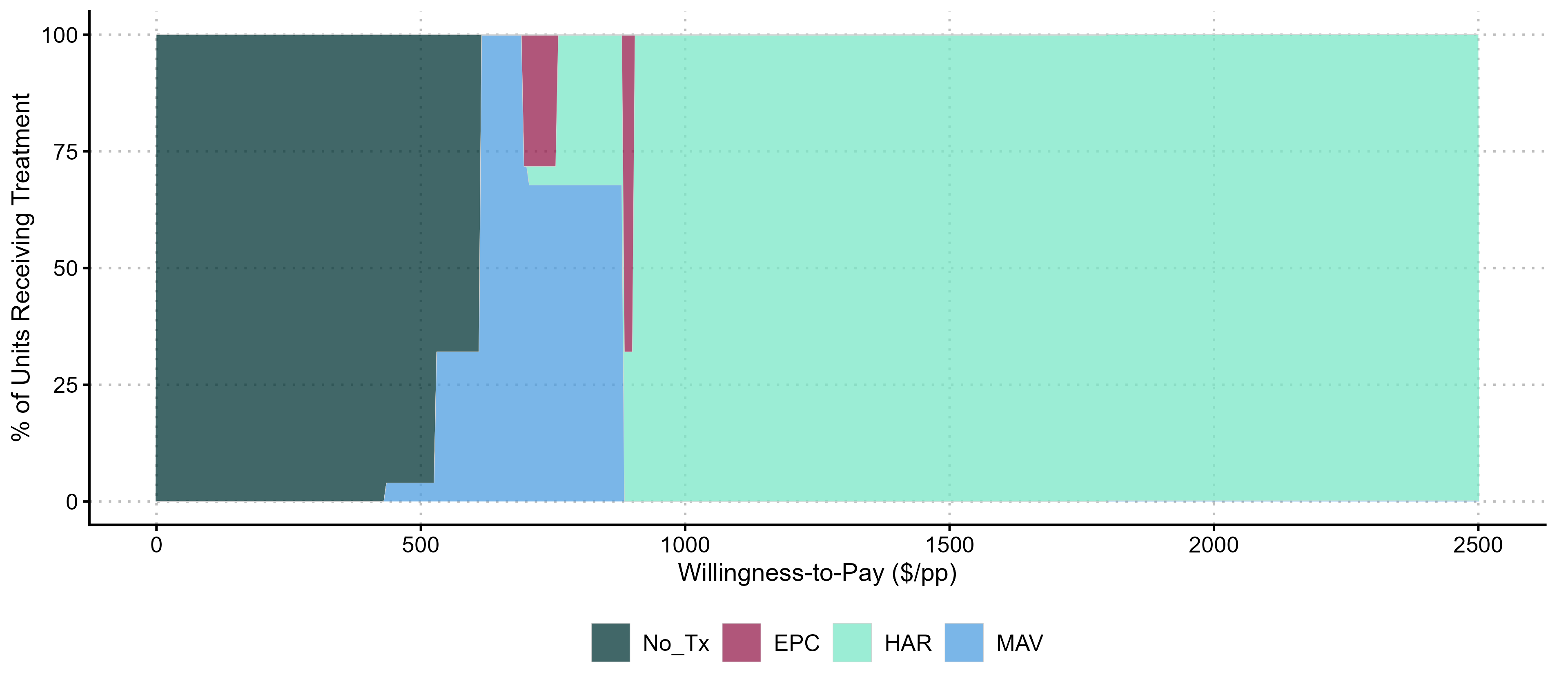}
		}
		\subfigure[Comparison with the K-means group memberships $\mathbf{z}^M(\hat{\boldsymbol{\mu}},\hat{\Sigma})$]{
			\includegraphics[width=13.75cm]{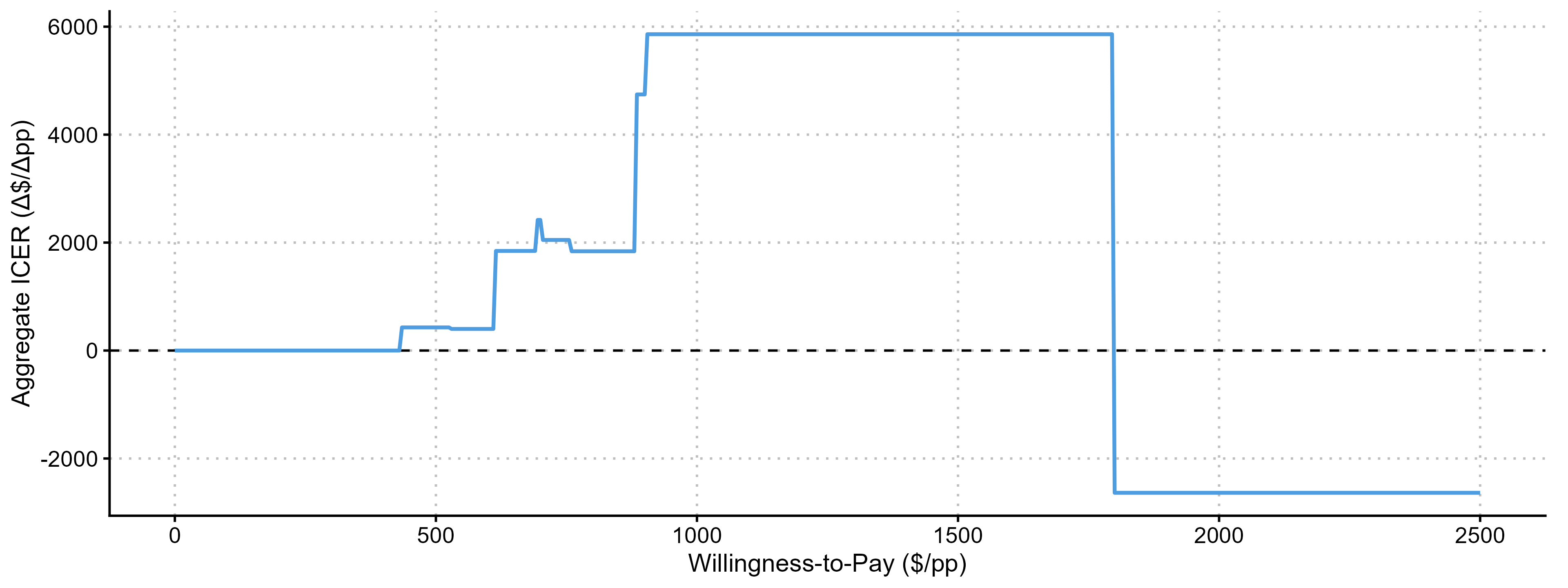}
		}
		\subfigure[Comparison with the status quo allocation]{
			\includegraphics[width=13.75cm]{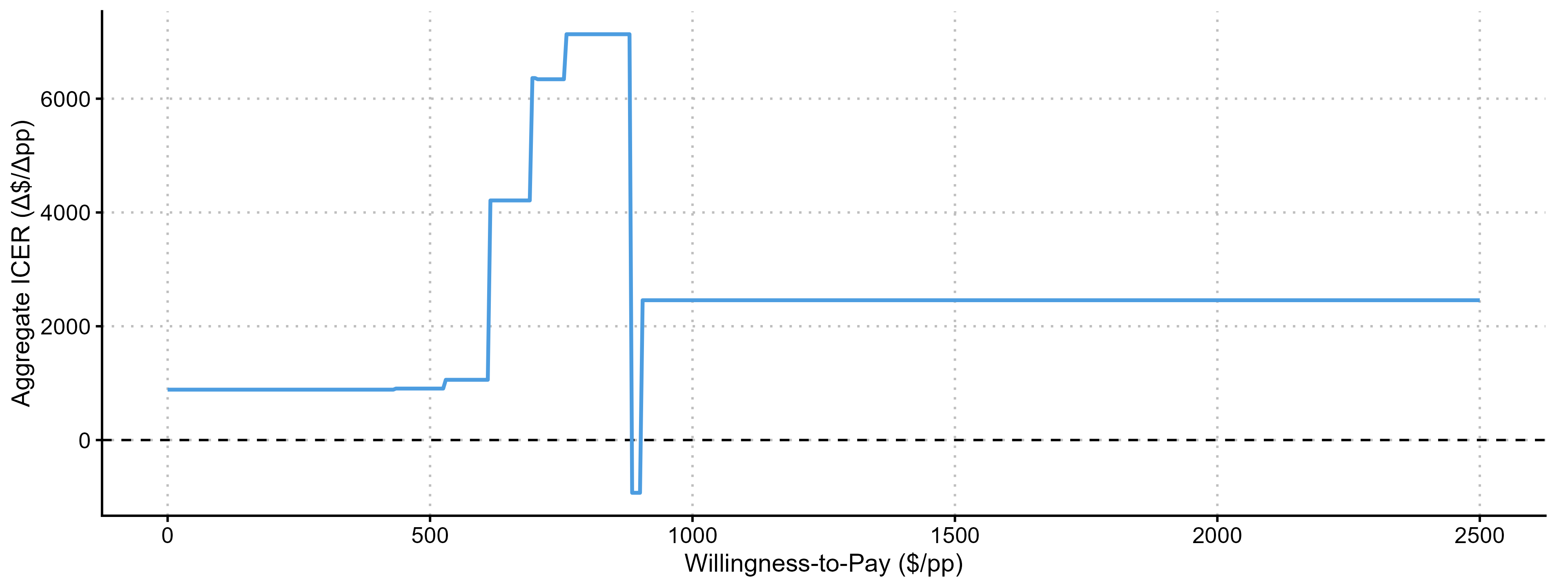}
		}
		\caption{Results of the cost-effectiveness analysis with perfect adherence to treatment and tree-based predicted group memberships $\hat{h}(V^*)$.}\label{fig:4}
	\end{figure}
	\begin{figure}[t!]
		\centering
		\subfigure[Optimal treatment allocation]{
			\includegraphics[width=13.4cm]{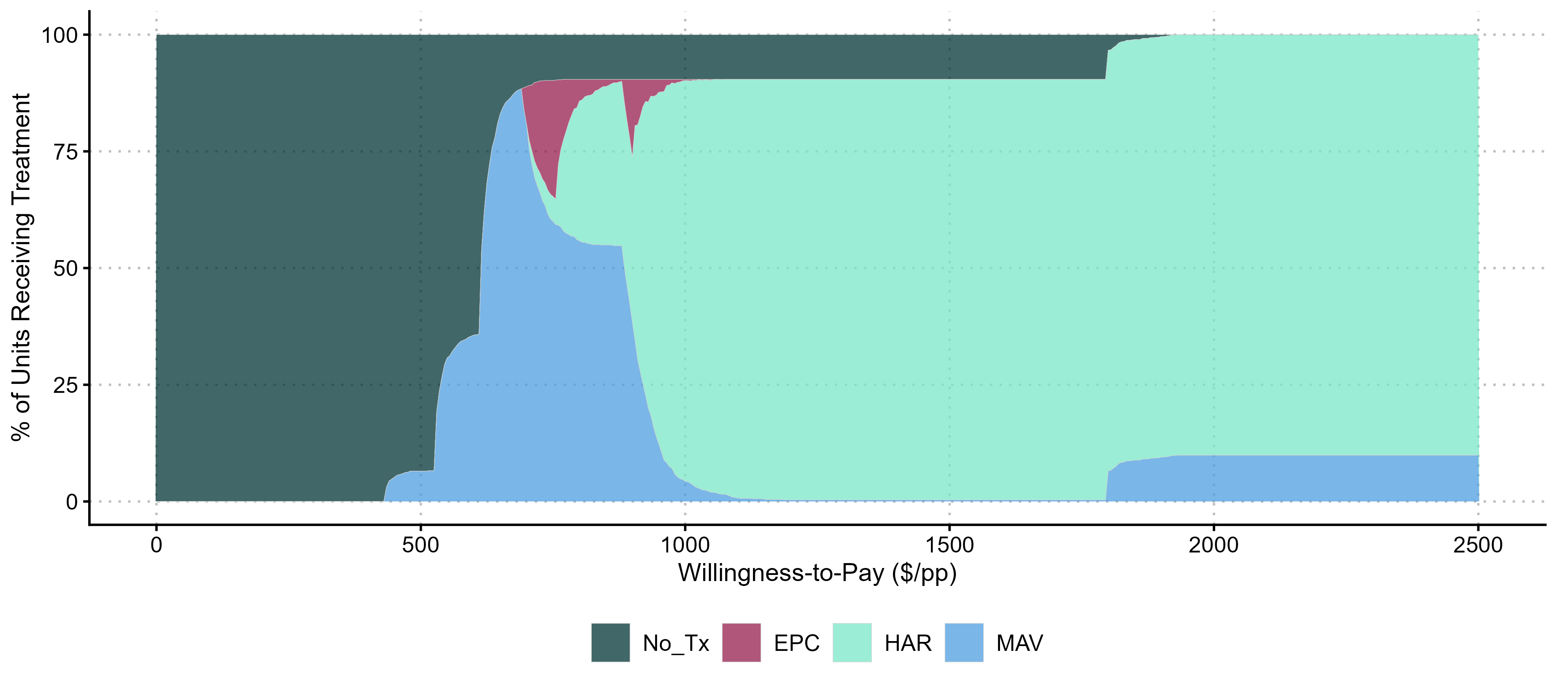}
		}
		\subfigure[Comparison with perfect adherence]{
			\includegraphics[width=13.75cm]{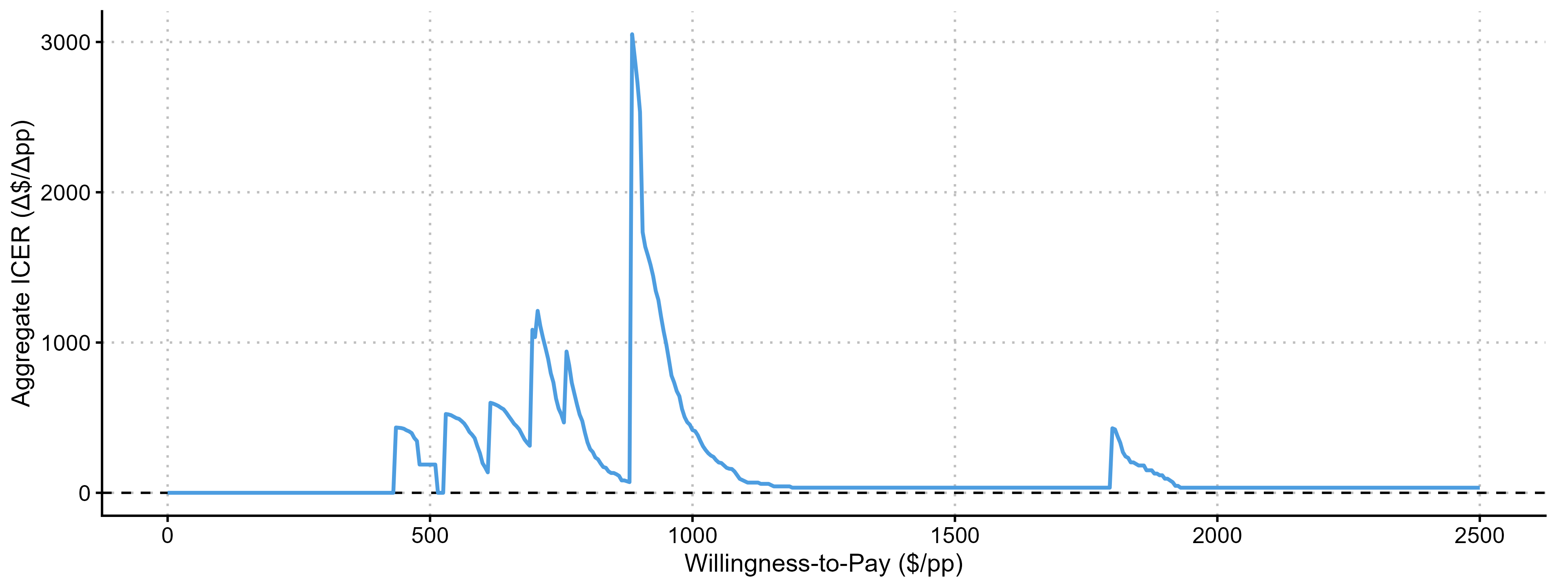}
		}
		\subfigure[Comparison with the status quo allocation (the y-axis is cropped at -100,000 to ease visualization)]{
			\includegraphics[width=13.75cm]{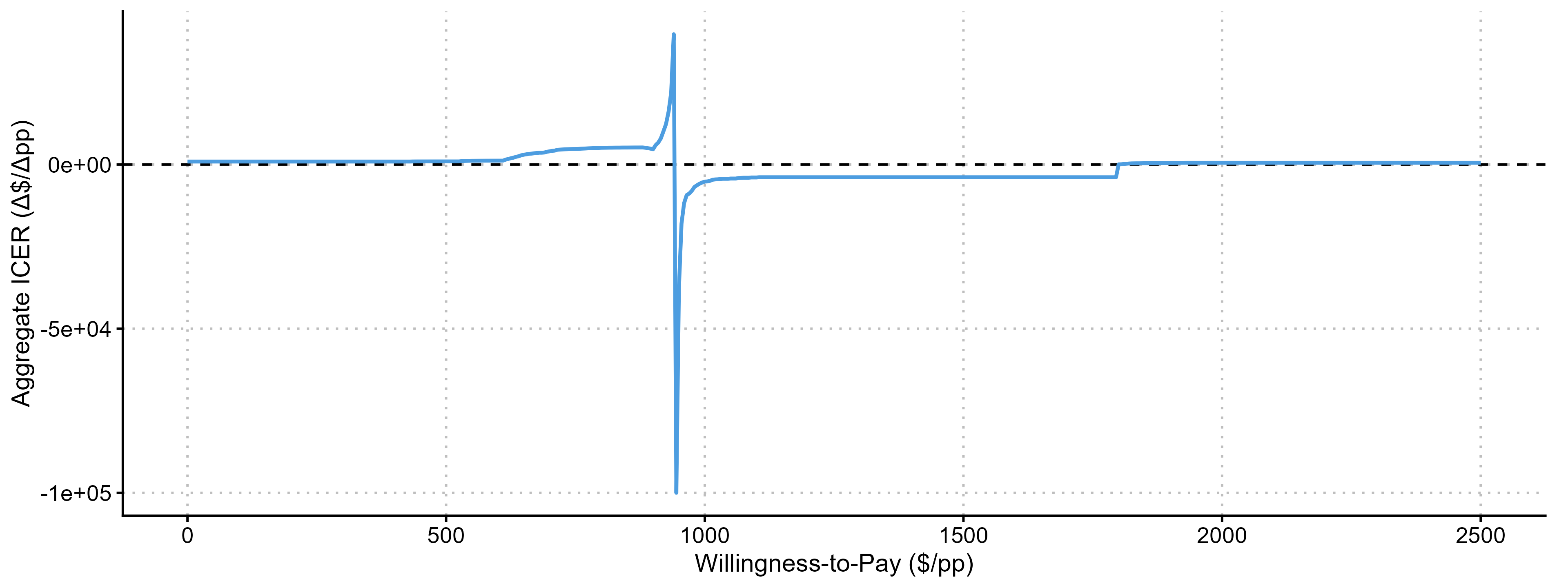}
		}
		\caption{Results of the cost-effectiveness analysis with predicted adherence to treatment and K-means group memberships, $\mathbf{z}^M(\hat{\boldsymbol{\mu}},\hat{\Sigma})$.}\label{fig:5}
	\end{figure}
	\subsubsection{Tree-Based Group Memberships under Perfect Adherence}
	Panel (a) of Figure \ref{fig:4} shows the optimal treatment allocation under tree-based predicted group memberships $\hat{h}(V^*)$ and perfect adherence. For brevity, I do not show the total health benefits and costs associated with this allocation. As mentioned in Section \ref{sec:54}, only two patients are treated with Mavyret rather than Harvoni when using tree-based group memberships and WTP is sufficiently high (above CAN\$1,800/pp; not clearly visible in Panel (a)). Because tree-based group memberships misclassify 19.72\% of all patients in the sample, it is possible to compare the allocation shown in Panel (a) of Figure \ref{fig:4} to the treatment allocation shown in Panel (b) of Figure \ref{fig:3} (both under perfect adherence). This is what is shown in Panel (b) of Figure \ref{fig:4}, indicating that using $\hat{h}(V^*)$ rather than $\mathbf{z}^M(\hat{\boldsymbol{\mu}},\hat{\Sigma})$ to allocate treatment options entails a trade-off between total health benefits and total costs for any WTP below CAN\$1,800/pp. When WTP $>$ CAN\$1,800/pp, using $\hat{h}(V^*)$ for treatment allocation leads to higher total costs (+CAN\$M3.13) and \textit{lower} total health benefits (-1,188.3 pp), which explains the negative aggregate ICER values.
	\par
	If the status quo allocation is used as the benchmark, Panel (c) of Figure \ref{fig:4} shows that using $\hat{h}(V^*)$ to allocate treatment options also entails a trade-off between total health benefits and total costs, except when WTP lies between CAN\$885/pp and CAN\$900/pp. In the latter case, the allocation shown in Panel (a) of Figure \ref{fig:4} is dominated by the status quo allocation due to greater total costs (+CAN\$M0.62) and lower total health benefits (-666.3pp). Panels (b) and (c) of Figure \ref{fig:4} therefore show that using the (7-depth) tree-based group memberships never leads to a dominant situation compared to the status quo allocation or the one shown in Panel (a) of Figure \ref{fig:3}.
	\par
	\begin{figure}[t!]
		\centering
		\subfigure[Optimal treatment allocation]{
			\includegraphics[width=13.4cm]{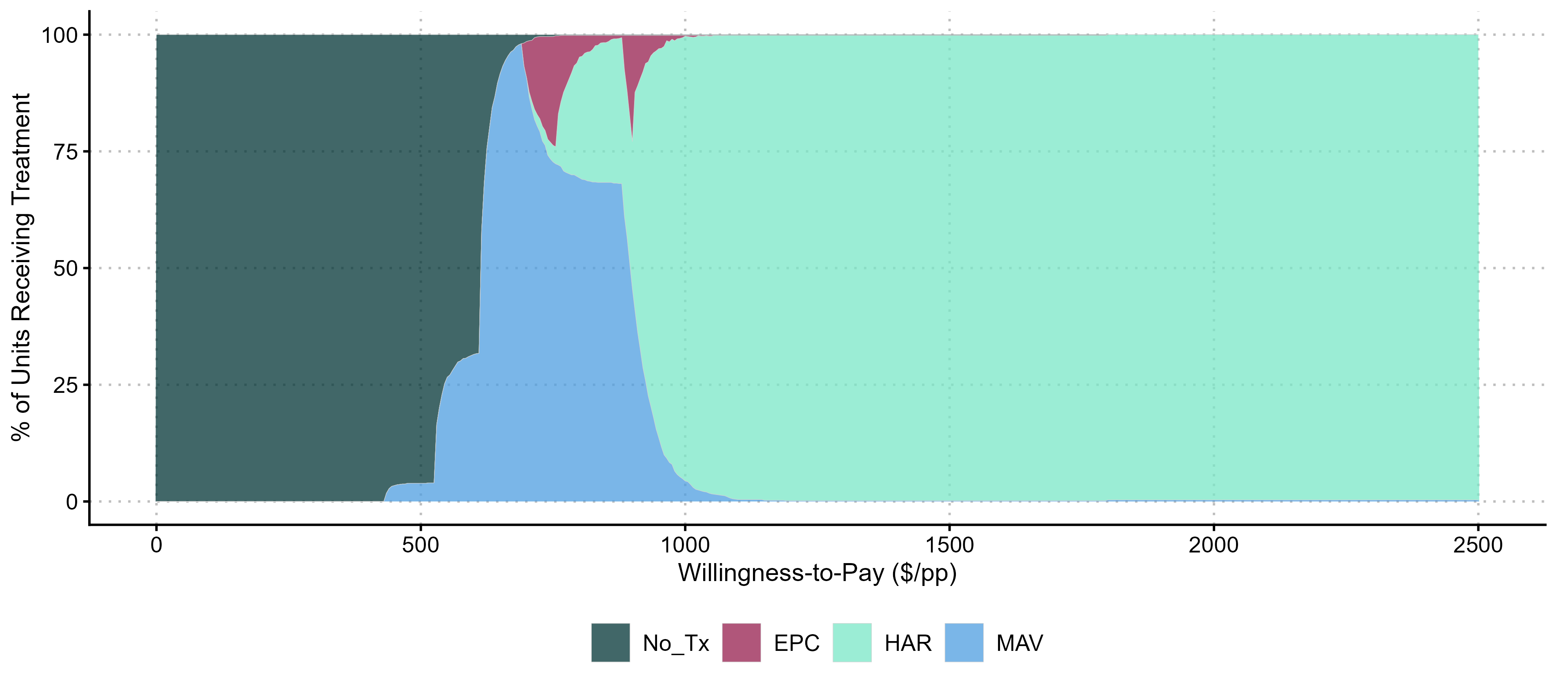}
		}
		\subfigure[Comparison with the K-means group memberships $\mathbf{z}^M(\hat{\boldsymbol{\mu}},\hat{\Sigma})$]{
			\includegraphics[width=13.75cm]{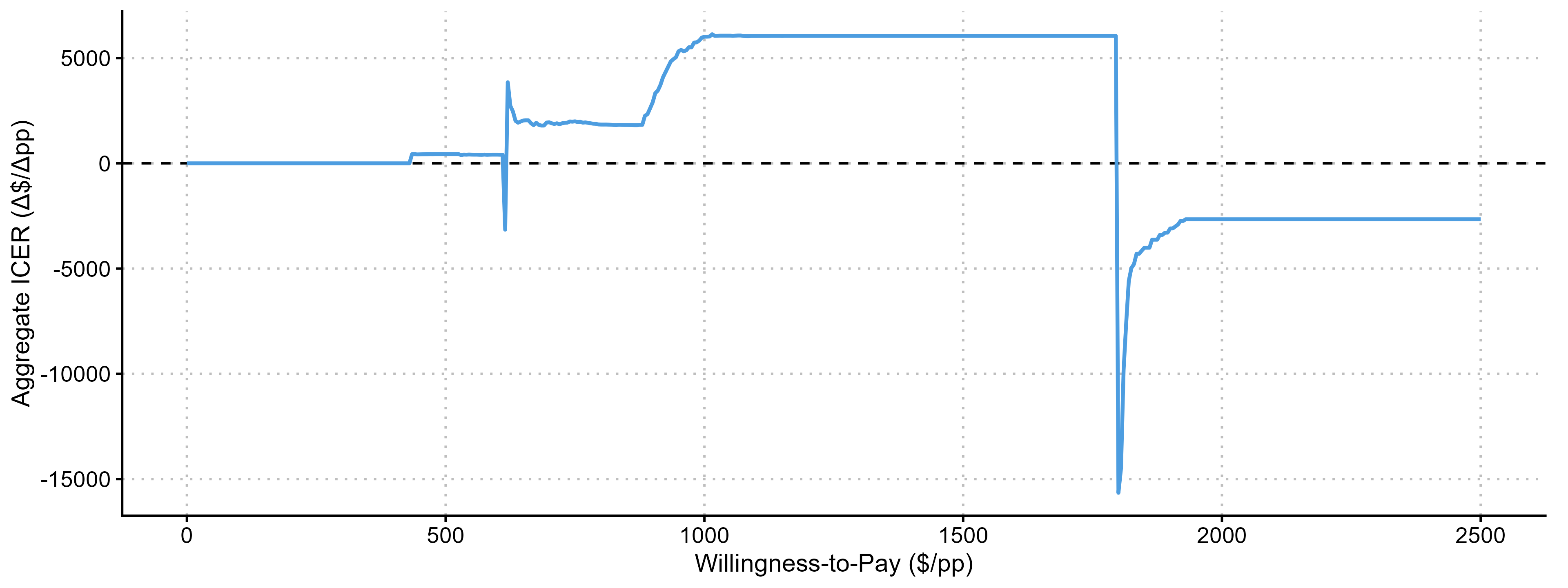}
		}
		\subfigure[Comparison with the status quo allocation (the y-axis is cropped at 100,000 to ease visualization)]{
			\includegraphics[width=13.75cm]{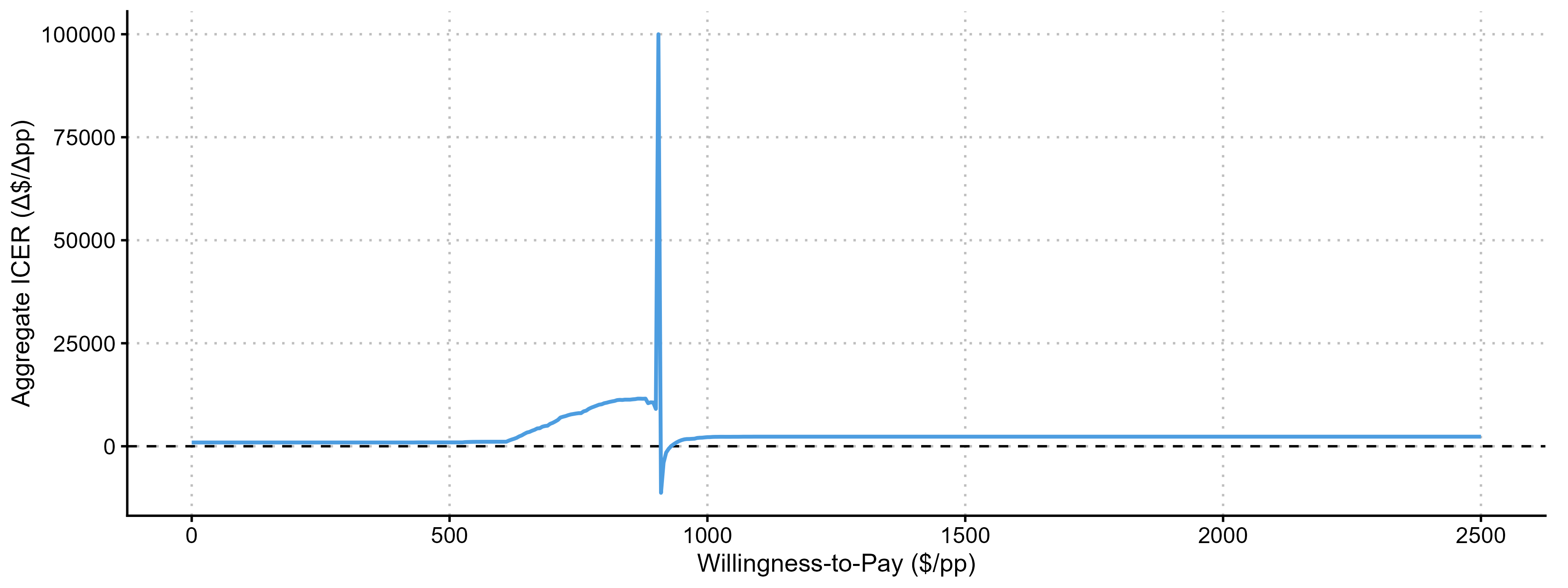}
		}
		\caption{Results of the cost-effectiveness analysis with predicted adherence to treatment and tree-based predicted group memberships $\hat{h}(V^*)$.}\label{fig:6}
	\end{figure}
	\subsubsection{K-means Group Memberships under Predicted Adherence}
	Figure \ref{fig:5} shows the result of the CEA with K-means group memberships under \textit{predicted} adherence, where every curve in each panel are now ``smooth" because every CATE use for the CEA is now the product of the multiplication between the estimated CATE under perfect adherence, $\hat{\tau}_{i,d}(\hat{a}(W_i,\hat{\tau}_d(1)))$, and the continuous adherence predicted from the two-part model. The upper panel of \ref{fig:5} shows the optimal treatment allocation obtained with predicted adherence, which leads to slight increases in the share of patients optimally treated by Mavyret when WTP is above CAN\$1,930/pp compared to perfect adherence (122 vs. 118 patients). 
	\par
	Because using predicted adherence might change the optimal treatment allocation -- and therefore change total health benefits and costs -- it is possible to use the treatment allocation under perfect adherence (Panel (a) of Figure \ref{fig:3}) as a benchmark for comparison. This is what is shown in Panel (b) of Figure \ref{fig:5}, where the aggregate ICER values show that using predicted adherence instead of perfect adherence always reduces total health benefits and/or total cost for any WTP above CAN\$435/pp. This is not surprising given that predicted adherence can never lead to greater health benefits than perfect adherence, but can reduce total costs by reallocating treatment choices to options with the largest health benefits after accounting for imperfect adherence. Nonetheless, predicted adherence is believed to yield more credible estimates of health benefits since observed adherence is imperfect for a non-negligible number of patients in the sample.
	\par
	Panel (c) of Figure \ref{fig:5} shows that the status quo allocation under predicted adherence is dominated by the allocation shown in Panel (a) of Figure \ref{fig:5} for any value of WTP that lies between CAN\$945/pp and CAN\$1,800/pp. When the policymaker's WTP belongs in this interval, the allocation shown in Panel (a) of Figure \ref{fig:5} implies increased total health benefits that go from +12.7pp to +639.9pp, whereas total costs decrease from an amount between CAN\$4.49M and CAN\$2.48M. The generated savings are larger than those shown in Figure \ref{fig:3} because predicted adherence is applied to both the status quo and the optimal allocation, which widens the gap between the two allocations compared to perfect adherence.
	\subsubsection{Tree-based Group Memberships under Predicted Adherence}
	Finally, Figure \ref{fig:6} shows the result of the CEA with both predicted adherence and tree-based group memberships, analogous to Figure \ref{fig:4} but with predicted adherence. As in Figure \ref{fig:4}, Panels (b) and (c) of Figure \ref{fig:6} show the aggregate ICER that uses the allocation obtained with the K-means group memberships and the status quo allocation as the benchmark, both under predicted adherence. This comparative exercise shows that the allocation based on tree-based group memberships dominates the status quo under predicted adherence when WTP lies between CAN\$910/pp and CAN\$925/pp, with total savings ranging from CAN\$M2.92 to CAN\$M0.53 and increased total health benefits going from +258.2pp to +868.2pp. On the other hand, Panel (b) of Figure \ref{fig:6} shows that the treatment allocation shown in Panel (a) is dominated by the treatment allocation obtained from the K-means group memberships (under predicted adherence) when WTP is between CAN\$610/pp and CAN\$620/pp and also when WTP $>$ CAN\$1,800/pp. These results show that, although using the tree-based predicted group memberships can be quite practical, it increases the risk of suboptimal treatment allocations compared to using the K-means group memberships, especially when imperfect adherence is taken into account.
	\section{Discussion and Conclusion}\label{sec:6}
	Decision-making processes are pervasive in daily life. In this paper, I developed a relatively simple methodology to help policymakers build policy rules that can guide decision-makers using observational data. Commonly known policy rules are often not mandatory; this can take the form of guidelines, evaluation grids, or even personal criteria for daily consumption choices. Even if they are not strictly enforced, policy rules remain important in practice if they serve as a baseline decision rule that agents can improve upon using contextual information and personal judgment.
	\par
	In health care and health sciences in general, clinical judgment of health practitioners is often opposed to ``evidence-based medicine", where the former is assumed to be often unreliable by the proponents of the latter. In contrast, the latter is often criticized for not accounting for contextual elements of any single choice problem \citep{alpert_why_2010,knaapen_evidence-based_2014}. As several authors have recently mentioned, those two visions are not dichotomous but should rather coexist for the benefit of patients \citep{riviere_can_2020, subbiah_next_2023, einav_limitations_2024}. This includes leveraging observational data to better identify the real-world elements that influence treatment effectiveness at the population level, and then using those (\textit{macro}) elements to guide the decision-making process at the patient's (\textit{micro}) level without abstracting away from contextual information.
	\par
	Better clinical guidelines that are based on real-world evidence and that acknowledge the usefulness of discretionary decisions under specific circumstances are therefore warranted and necessary to improve the quality of care and reduce the incidence of medical errors in practice. The case of DAA treatment choice for HIV/HCV co-infected patients is one of the many examples where observational data can be used to inform policymakers about the trade-offs that exist in the real world when formulating policy rules, especially when those policy rules can be used in a legal context to defend the various choices made by the implied decision-makers.
	\par
	The methodology developed in this paper to design (optimal) policy rules has several benefits~: it relies more practical assumptions than what is typically assume in the literature in policy learning, it advocates for the use of a specific clustering algorithm that leads to the consistent estimation of all CATEs in the sample, it easily deals with missing values in covariates, it can handle semi-continuous treatment variable in multi-action framework, it is more parsimonious than typical ML methods -- which leads to more precise estimates of most CATEs in the sample (depending on the group's size) while reducing overfitting at the same time --, and it yields interpretable results that can be easily used to formulate updated guideline recommendations.
	\par
	Nonetheless, the proposed methodology has several limitations. First, the credibility and usefulness of the estimated policy rules will depend mostly on whether group-wise unconfoundedness and group-wise heterogeneity (Assumptions \ref{ass:1} and \ref{ass:3}) are both satisfied in practice. If they are not, the estimated policy rules might lead to suboptimal treatment allocations, which may then lead to worse health outcomes and/or higher costs than the status quo. Making the proposed estimation methodology robust to violations of these two assumptions is left to future research. Second, the computational burden of the proposed approach is relatively high, especially when a minimum eigenvalue $\lambda_{min}$ has to be enforced for the weighted K-means algorithm to yield reliable results. To reduce the computational burden, it might be a good idea to exclude all discrete variables from the set of covariates used to perform the weighted K-means algorithm. By doing so, all within-group estimated variances will be bounded away from zero, unless groups can be formed by a single observation. Discarding initial values that lead to such a degenerate result thus allows for the minimum eigenvalue $\lambda_{min}$ to be very close to zero and greatly alleviates the computational burden.
	\par
	Third, it requires enough within-group variation in order to identify all CATEs. This is, however, an issue with all ML techniques used to identify TEH, and the fact that the method is more parsimonious than most other ML techniques makes this issue less relevant in practice. Fourth and last, the method cannot identify contraindications that are being strictly followed by most prescribing physicians in their current practice. Although this limitation is not related to the proposed estimation strategy \textit{per se}, it is important to keep in mind that contraindications that are well known by physicians are likely to be absent from observational datasets. This may then lead to estimated ``optimal" treatment allocations that do not account for these contraindications, which may worsen health outcomes compared to the status quo.
	\par
	Results of the empirical application show that it is possible to improve upon the current practice when treating co-infected HIV/HCV patients in Canada. Specifically, my results show that it would have been possible to reduce prescription costs in the CCC by an amount between CAN\$2.48 million and CAN\$4.49 million while still increasing total health benefits by reallocating treatment options among treated patients in the cohort. However, such important savings are mostly due to Group 1, where it is estimated that 118 patients in the sample have approximately an 80\% chance of spontaneously clearing HCV without taking any treatment. Because acute infection is temporary but still lasts around 6 months on average, not treating these patients may not be optimal if this leads to an important rise in the reinfection rate within the target population \citep{young_rate_2023}. Results from the empirical application also show that treatment reallocation could have substantially increased the probability of reaching SVR for every treated patient in the sample (+2,847.5 pp among 745 treated patients under predicted adherence), but this would have come at the expense of larger prescription costs (+CAN\$1.52 million).
	\par
	Although this paper focused exclusively on SVR, other outcomes are also important for both mono-infected HCV and co-infected HIV/HCV patients: probability of reinfection, degree of liver stiffness, mortality, etc. Still, it is well known that curing HCV improves patients' quality of life and other nonvirological outcomes \citep{simmons_long-term_2015, park_impact_2021}. This is why SVR remains the main outcome when assessing the effectiveness of DAA regimens, although additional outcomes will be included in future analyses.
	\par
	At last, it is important to recall that, according to the obtained results in the context of HIV/HCV co-infection, no DAA regimen is dominated by any other single DAA regimen or set of DAA regimens when the policymaker's WTP is below a certain level. In other words, the optimal treatment choice for any HIV/HCV co-infected patient is not trivial and depends on the characteristics of the patient, on the chosen decision criterion, and the policymaker's WTP. Such a result warrants caution when physicians prescribe an anti-HCV treatment for a never-treated, co-infected patient, especially when costs vary greatly across treatment options. Another interesting feature of the result is that the optimal treatment choice for the physician does not seem to depend on the ART taken by the patient, although this could be a consequence of contraindications in the prescribing patterns of physicians regarding HIV/HCV co-infection, as mentioned above. More research is required to validate such an intuition.
	\bibliographystyle{chicago}
	\newpage
	\bibliography{Langevin_JMP}
	\newpage
	\appendix
	\setstretch{1.15}
	\section{Appendix - Proofs}\label{app:A}
	\subsection{Proof of Lemma \ref{lem:1}}
	\begin{proof}
		The proof is divided into two parts. The first part is when the algorithm goes from the Estimation step to the Classification step. In this case, we have
		\begin{align*}
			d^{2TM}(\mathbf{z}^{(k-1)},\boldsymbol{\mu}^{(k)}, \Sigma^{(k)}) \to d^{2TM}(\mathbf{z}^{(k)},\boldsymbol{\mu}^{(k)}, \Sigma^{(k)}),
		\end{align*}
		for any given value of $\mathbf{z}^{(k-1)} \in \mathcal{Z}$, $\boldsymbol{\mu}^{(k)} \in \Theta^{\bar{S}}$, and $\Sigma^{(k)} \in \Omega^{\bar{S}}$. This step never increases $d^{2TM}(\cdot,\cdot,\cdot)$ given that    
		\begin{align*}
			d^{2TM}(\mathbf{z}^{(k-1)},\boldsymbol{\mu}^{(k)},\Sigma^{(k)}) &= \sum_{g\in \mathcal{S}}  \sum_{i\in [N]} z^{(k-1)}_{i,g} ((X_i - \boldsymbol{\mu}^{(k)}_g)^{\top} \{\Sigma_g^{(k)}\}^{-1} (X_i - \boldsymbol{\mu}^{(k)}_g) + \log \det(\Sigma_g^{(k)})),\\
			&\ge \sum_{g\in \mathcal{S}} \sum_{i\in [N]} z^{(k)}_{i,g} d_i^{2M}(\boldsymbol{\mu}^{(k)},\Sigma^{(k)}),\\
			&= d^{2TM}(\mathbf{z}^{(k)},\boldsymbol{\mu}^{(k)},\Sigma^{(k)}),
		\end{align*}
		given that $z^{(k)}_{i,g} = \mathbbm{1}[z_i^{(k)}=g]$ with $z_i^{(k)} = \arg \min_{g\in\mathcal{S}} d_i^{2M}(\boldsymbol{\mu}^{(k)}_g, \Sigma^{(k)}_g)$ for each $i \in [N]$.
		\par
		When going from the Classification step to the Estimation step, we have
		\begin{align*}
			d^{2TM}(\mathbf{z}^{(k)},\boldsymbol{\mu}^{(k)},\Sigma^{(k)}) \to d^{2TM}(\mathbf{z}^{(k)},\boldsymbol{\mu}^{(k+1)},\Sigma^{(k+1)}).
		\end{align*}
		In this case, it is convenient to define
		\begin{align*}
			d^{2TM}_g(\mathbf{z}^{(k)},\boldsymbol{\mu}_g^{(k)},\Sigma_g^{(k)}) &:= \sum_{i\in [N]} z^{(k)}_{i,g}d_i^{2M}(\boldsymbol{\mu}_g^{(k)},\Sigma_g^{(k)}).
		\end{align*}
		Therefore, minimizing $d^{2TM}(z,\boldsymbol{\mu},\Sigma)$ with respect to $\boldsymbol{\mu}$ and $\Sigma$ is equivalent to minimizing $d^{2TM}_g(\mathbf{z},\boldsymbol{\mu}_g,\Sigma_g)$ with respect to $\boldsymbol{\mu}_g$ and $\Sigma_g$ separately for each value of $g \in \mathcal{S}$. Thus, we can write
		\begin{align*}
			\frac{\partial d^{2TM}_g(\mathbf{z}^{(k)},\boldsymbol{\mu}_g^{(k)},\Sigma_g^{(k)})}{\partial \boldsymbol{\mu}^{(k)}_g } &= \frac{\sum_{i\in [N]} \partial z^{(k)}_{i,g} (X_i - \boldsymbol{\mu}^{(k)}_g)^{\top} \{\Sigma_g^{(k)}\}^{-1} (X_i - \boldsymbol{\mu}^{(k)}_g)}{\partial \boldsymbol{\mu}^{(k)}_g}.
		\end{align*}
		Using the fact that $\frac{\partial(X_i - \boldsymbol{\mu}^{(k)}_g)^{\top} \{\Sigma_g^{(k)}\}^{-1} (X_i - \boldsymbol{\mu}^{(k)}_g)}{\partial \boldsymbol{\mu}^{(k)}_g} = 2(X_i - \boldsymbol{\mu}^{(k)}_g)\{\Sigma_g^{(k)}\}^{-1}$, it is easy to show that
		\begin{align*}
			\frac{\partial d^{2TM}_g(\mathbf{z}^{(k)},\boldsymbol{\mu}_g^{(k)},\Sigma_g^{(k)})}{\partial \boldsymbol{\mu}^{(k)}_g } &= 0 \Leftrightarrow  \hat{\boldsymbol{\mu}}_g = \frac{\sum_{i\in [N]} z^{(k)}_{i,g} X_i}{\sum_{j=1}^N z^{(k)}_{j,g}} \equiv \boldsymbol{\mu}^{(k+1)}_g.
		\end{align*}
		Note that $\hat{\boldsymbol{\mu}}_g$ is the (unique global) minimizer of $d_g^{2TM}(\mathbf{z}^{(k)},\boldsymbol{\mu}_g)$ since $d_g^{2TM}(\cdot,\cdot,\cdot)$ is unbounded from above when expressed as a function of $\boldsymbol{\mu}_g$. Therefore, we have that
		\begin{align*}
			d^{2TM}_g(\mathbf{z}^{(k)},\boldsymbol{\mu}_g^{(k)},\Sigma_g^{(k)}) \ge d^{2TM}_g(\mathbf{z}^{(k)},\boldsymbol{\mu}_g^{(k+1)},\Sigma_g^{(k)}),
		\end{align*}
		for any $g \in \mathcal{S}$.
		\par
		From the result of Appendix \ref{lem:S4}, we also have that
		\begin{align*}
			d^{2TM}_g(\mathbf{z}^{(k)},\boldsymbol{\mu}_g^{(k+1)},\Sigma_g^{(k)}) \ge d^{2TM}_g(\mathbf{z}^{(k)},\boldsymbol{\mu}_g^{(k+1)},\Sigma_g^{(k+1)}),
		\end{align*}
		where $\Sigma_g^{(k+1)}$ is the unique global minimizer of $d^{2TM}_g(\mathbf{z}^{(k)},\boldsymbol{\mu}_g^{(k+1)},\Sigma_g)$ with respect to $\Sigma_g$  according to the result of Appendix \ref{lem:S4}. However, such a result requires that $\Sigma_g^{(k+1)}$ is positive-definite for all $g \in \mathcal{S}$, which might not always be the case in practice.
	\end{proof}
	\subsection{Proof of Theorem \ref{th:1}}
	\begin{proof}
		If we combine the results of Lemma S.2 and Lemma S.3 from \cite{langevin_bias-reduced_2026} with the fact that $\mathbb{E}[ d_i^{2M}(s_i)] = p + \log \det(\Sigma^0_{s_i})$ with $d_i^{2M}(s_i) \equiv d_i^{2M}(\boldsymbol{\mu}^0_{s_i},\Sigma^0_{s_i})$, this leads to
		\begin{align*}
			\text{Pr}[d^{2M}(s_i) \ge d^{2M}(j)] &\le \frac{p(1+\frac{1}{2}\sqrt{M} + p^{-1}\log \det(\Sigma^0_{s_i}))}{2\sum_{l=1}^p d_{js_i,ll}  + \sum_{l=1}^p\sum_{m\le l} (v_{js_i,lm})^2 - p +  \sum_{l=1}^p (b_{js_i,l})^2 + \log \det(\Sigma^0_{j})}.
		\end{align*}
		Since $d_{js,ll} > 0$ for any $(l,j,s) \in \{1,2,...,p\} \times \mathcal{S} \times \mathcal{S}\backslash j$, there exists a finite $K>0$ such that
		\begin{align*}
			\text{Pr}[d^{2M}(s_i) \ge d^{2M}(j)] &\le \frac{p(1+\frac{1}{2}\sqrt{M}+ p^{-1} \sum_{m=1}^p \log \lambda_{s_i,m})}{2pK + p(p+1)K/2 - p +  pK + \sum_{m=1}^p \log \lambda_{j,m}},\\
			&= \frac{1+\frac{1}{2}\sqrt{M}+ \overline{\log \lambda_{s_i}}}{7K/2 + pK - 1 + \overline{\log \lambda_{j}}},\\
			&= O(p^{-1}),
		\end{align*}
		where $\lambda_{j,m}$ is the $m^{th}$ eigenvalue of $\Sigma_j^0$ (when ranked in ascending order), and where $\overline{\log \lambda_{j}} = p^{-1} \sum_{m=1}^p \log \lambda_{j,m}$ for any $j \in \mathcal{S}$. The first inequality uses the fact that $\det(\Sigma^0_{j}) = \prod_{m=1}^p \lambda_{j,m}$. Note also that $\lambda_{j,m} = O(1)$ as $p \to \infty$ (so $\overline{\log \lambda_{j}}$ will be as well) for any pair $(j,m) \in \mathcal{S} \times \{1,2,...,p\}$ given that
		\begin{align*}
			\lambda_{j,m} &= \frac{|| U_j^{-1}\nu_{jm} ||^2}{|| \nu_{jm} ||^2} = \frac{\sum_{i=1}^p ( \sum_{l=1}^i \tilde{u}_{j,il} \nu_{jm,l})^2}{\sum_{l=1}^p (\nu_{jm,l})^2},\\
			&= \frac{\sum_{i=1}^p \sum_{l=1}^i (\tilde{u}_{j,il}\nu_{jm,l})^2 + 2\sum_{i=1}^p \sum_{l=1}^i \sum_{s<l} \tilde{u}_{j,il}\tilde{u}_{j,is} \nu_{jm,l}\nu_{jm,s} }{\sum_{l=1}^p (\nu_{jm,l})^2} = O(1),
		\end{align*}
		as $p \to \infty$ since $\tilde{u}_{j,il} < \infty$ for any $(j,i,l) \in \mathcal{S}\times \{1,2,...,p\}^2$ by Assumption \ref{ass:8}, where $\nu_{jm} = (\nu_{jm,1},...,\nu_{jm,p})^{\top}$ is the eigenvector associated to $\lambda_{j,m}$, $\tilde{u}_{j,il}$ is the element located at the $i^{th}$ row and $l^{th}$ column of $U^{-1}_j$ with $U_j$ defined as in Assumption \ref{ass:8}, and where $||\cdot||^2$ is the squared Euclidean distance. This also implies that there exists a constant $0 < \tilde{K} \le 1$ such that
		\begin{align*}
			\text{Pr}[\cap_{g=1}^{\bar{S}} \hat{z}_{i,g}(\boldsymbol{\mu}^0,\Sigma^0) = s_{i,g}]
			&= \prod_{j\ne s_i} (1 - \text{Pr}[ d_i^{2M}(s_i) \ge d_i^{2M}(j)]),\\
			&\ge (1 - \tilde{K}p^{-1})^{(S-1)}.
		\end{align*}
		where $\hat{z}_{i,g}(\boldsymbol{\mu}^0,\Sigma^0)  = \mathbbm{1}[\hat{z}_{i}(\boldsymbol{\mu}^0,\Sigma^0)  = g]$.
		Given that both $\tilde{K}$ and $\bar{S}$ are finite constants that do not depend on $p$, we have that $(1 - \tilde{K}p^{-1})^{(S-1)} \to 1$ as $p\to \infty$, which implies that 
		\begin{align*}
			\text{Pr}[\cup_{g=1}^{\bar{S}} \hat{z}_{i,g}(\boldsymbol{\mu}^0,\Sigma^0) \ne s_{i,g}] = O(p^{-1}) \Leftrightarrow \text{Pr}[\hat{z}_{i}(\boldsymbol{\mu}^0,\Sigma^0) \ne s_{i}] = O(p^{-1}) 
		\end{align*}
	\end{proof}
	\subsection{Proof of Theorem \ref{th:2}}
	\begin{proof}
		As in the proof of Lemma \ref{lem:1}, we define
		\begin{align*}
			d^{2TM}_g(\mathbf{z},\boldsymbol{\mu}_g,\Sigma_g) &:= \sum_{i\in [N]} z_{i,g}((X_i-\boldsymbol{\mu}_g)^{\top}\Sigma_g^{-1}(X_i-\boldsymbol{\mu}_g)+ \log \det(\Sigma_g)),
		\end{align*}
		for any $\mathbf{z} = (z_1,...,z_N) \in \mathcal{Z}$ and where $z_{i,g} = \mathbbm{1}[z_i=g]$. Therefore, minimizing $d^{2TM}(\mathbf{z},\boldsymbol{\mu},\Sigma)$ with respect to $\boldsymbol{\mu}$ and $\Sigma$ is equivalent to minimizing $d^{2TM}_g(\mathbf{z},\boldsymbol{\mu}_g,\Sigma_g)$ with respect to $\boldsymbol{\mu}_g$ and $\Sigma_g$ separately for each value of $g \in \mathcal{S}$. First, we know from the proof of Lemma \ref{lem:1} that
		\begin{align*}
			d^{2TM}_g(\mathbf{z},\boldsymbol{\mu}_g,\Sigma_g) \ge \sum_{i\in [N]} z_{i,g}((X_i-\hat{\boldsymbol{\mu}}_g)^{\top}\Sigma_g^{-1}(X_i-\hat{\boldsymbol{\mu}}_g)+ \log \det(\Sigma_g)),
		\end{align*}
		for any $\mathbf{z} \in \mathcal{Z}$ and any $\Sigma_g \in \Omega$, where $\hat{\boldsymbol{\mu}}_g = \frac{\sum_{i\in [N]} z_{i,g}  X_{i}}{N_g}$ with $N_g = \sum_{i\in [N]} z_{i,g}$. Note that we can also define $\hat{\boldsymbol{\mu}}_g$ as follows
		\begin{align*}
			\hat{\boldsymbol{\mu}}_g := \frac{\sum_{j=1}^{\bar{S}} \sum_{i\in [N]} z_{i,g} s_{i,j}   X_{i}}{N_g} \equiv \sum_{j=1}^{\bar{S}}  \alpha_{gj} \bar{X}_{gj},
		\end{align*}
		where $s_{i,j} = \mathbbm{1}[s_i = j]$, $\alpha_{gj} = \frac{N_{gj}}{N_g} \equiv \frac{\sum_{i\in [N]} z_{i,g}s_{i,j} }{N_g}$ and where $\bar{X}_{gj} = \frac{\sum_{i\in [N]} z_{i,g}s_{i,j} X_i}{N_{gj}}$. Using the WLLN and the fact that $\mathbb{E}[\bar{X}_{gj}] = \boldsymbol{\mu}^0_j$ for fixed $z_{i,g}$ and fixed $s_{i,j}$, we can write that
		\begin{align*}
			\hat{\boldsymbol{\mu}}_g  \xrightarrow{p} \sum_{j=1}^{\bar{S}}   \bar{\alpha}_{gj} \times \boldsymbol{\mu}^0_j \equiv \bar{\boldsymbol{\mu}}_g,
		\end{align*}
		as $N \to \infty$, where $\bar{\alpha}_{gj} = \lim_{N\to \infty} \frac{\sum_{i\in [N]} s_{i,j} z_{i,g}}{ \sum_{i\in [N]} z_{i,g}}$. Using the trace operator and the definition of $d^{2TM}_g(\mathbf{z},\boldsymbol{\mu}_g,\Sigma_g)$, we can now write that
		\begin{align*}
			\plim_{N \to \infty} d^{2TM}_g(\mathbf{z},\boldsymbol{\mu}_g,\Sigma_g) &= \plim_{N \to \infty} \text{trace}(\sum_{i\in [N]}  z_{i,g}(X_i-\boldsymbol{\mu}_g)^{\top}\Sigma_g^{-1}(X_i-\boldsymbol{\mu}_g))+ \sum_{i\in [N]}  z_{i,g}\log \det(\Sigma_g),\\
			&\ge \plim_{N \to \infty} \text{trace}(\sum_{i\in [N]} z_{i,g}(X_i-\bar{\boldsymbol{\mu}}_g)^{\top}\Sigma_g^{-1}(X_i-\bar{\boldsymbol{\mu}}_g))+ N_g \log \det(\Sigma_g),\\
			&\equiv \plim_{N \to \infty} \text{trace}(\Sigma_g^{-1} \bar{C}_g)+ N_g \log \det(\Sigma_g),
		\end{align*}
		where $\bar{C}_g =  \sum_{i\in [N]} z_{i,g} (X_i-\bar{\boldsymbol{\mu}}_g)(X_i-\bar{\boldsymbol{\mu}}_g)^{\top}$. Note that the asymptotic minimum $\bar{\boldsymbol{\mu}} = (\bar{\boldsymbol{\mu}}_1,...,\bar{\boldsymbol{\mu}}_S)$ is unique unless $\bar{\boldsymbol{\mu}}_g = \bar{\boldsymbol{\mu}}_j$ for any pair $(g,j) \in \mathcal{S}\times \mathcal{S}\backslash g$. From Appendix \ref{lem:S4}, we also know that
		\begin{align*}
			\text{trace}(\Sigma_g^{-1} \bar{C}_g) +  N_g \log \det(\Sigma_g) &\ge \text{trace}( \{N_g^{-1} \bar{C}_g\}^{-1} \bar{C}_g) + N_g \log \det(N_g^{-1} \bar{C}_g),\\
			&= N_gp + N_g \log \det(N_g^{-1} \bar{C}_g),
		\end{align*}
		where each $\bar{C}_g$ is assumed to lie in the interior of $\Omega$ as $N \to \infty$ and $p$ fixed (this typically requires that $\sum_{i\in [N]} z_{i,g} \to \infty$ as $N \to \infty$). If we sum over all values of $g \in \mathcal{S}$ and divide by $N$, then we have that
		\begin{align*}
			\plim_{N \to \infty} N^{-1}d^{2TM}(\mathbf{z},\boldsymbol{\mu},\Sigma) &\ge \plim_{N \to \infty} N^{-1}\sum_{g \in \mathcal{S}} ( N_gp + N_g \log \det(N_g^{-1} \bar{C}_g) ),\\
			&= p  +  \sum_{g \in \mathcal{S}}  \bar{\alpha}_g \log \det( \bar{\Sigma}_g ),
		\end{align*}
		where $\bar{\alpha}_g = \lim_{N\to \infty} \frac{N_g}{N}$, and where $\bar{\Sigma}_g = \plim_{N \to \infty} N_g^{-1} \bar{C}_g$ for any $g \in \mathcal{S}$, assuming that this limit exists for any $\mathbf{z} \in \mathcal{Z}$ such that $\bar{\Sigma} \in \Omega^{\bar{S}}$. Therefore, we can write that
		\begin{align*}
			(\bar{\boldsymbol{\mu}}, \bar{\Sigma}) = \arg \min_{\boldsymbol{\mu} \in \Theta^{\bar{S}}, \Sigma \in \Omega^{\bar{S}}} \plim_{N\to \infty} d^{2TM}(\mathbf{z},\boldsymbol{\mu},\Sigma),
		\end{align*}
		for any $\mathbf{z} \in \mathcal{Z}$ such that $\bar{\Sigma} \in \Omega^{\bar{S}}$, with $\bar{\boldsymbol{\mu}} = (\bar{\boldsymbol{\mu}}_1,...,\bar{\boldsymbol{\mu}}_S)$ and $\bar{\Sigma} = (\bar{\Sigma}_1,...,\bar{\Sigma}_S)$.
		\par
		When almost every observation is correctly classified, the misclassification error goes to zero as the sample size increases. This implies that
		\begin{align}
			\bar{\alpha}_{g(g)} = \lim_{N\to \infty} \frac{\sum_{i\in [N]} s_{i,(g)} z_{i,g}}{ \sum_{i\in [N]} z_{i,g}} = 1,
		\end{align}
		where $(g)$ is a suitable permutation in the label $g$, and where $\bar{\alpha}_{gj} = 0$ for all $j = \mathcal{S}\backslash (g)$. This directly leads to $\bar{\boldsymbol{\mu}}_g = \boldsymbol{\mu}^0_g$ given that $\bar{\boldsymbol{\mu}}_g = \sum_{j\in \mathcal{S}} \bar{\alpha}_{gj}\times \boldsymbol{\mu}^0_j$. Moreover, we also have that
		\begin{align*}
			\bar{\Sigma}_g &= \plim_{N\to \infty} N_g^{-1} \bar{C}_g,\\
			&= \plim_{N\to \infty} N_g^{-1} \sum_{i\in [N]} z_{i,g} (X_i-\boldsymbol{\mu}^0_g)(X_i-\boldsymbol{\mu}^0_g)^{\top},\\
			&= \plim_{N\to \infty} N_g^{-1} \sum_{i\in [N]} z_{i,g} s_{i,(g)} (X_i-\boldsymbol{\mu}^0_g)(X_i-\boldsymbol{\mu}^0_g)^{\top} + N_g^{-1} \sum_{j\ne (g)} \sum_{i\in [N]} z_{i,g} s_{i,j} (X_i-\boldsymbol{\mu}^0_g)(X_i-\boldsymbol{\mu}^0_g)^{\top},\\
			&= \Sigma^0_g +   \plim_{N\to \infty} \sum_{j\ne (g)}  \frac{\sum_{i\in [N]}  z_{i,g} s_{i,j} (X_i-\boldsymbol{\mu}^0_g)(X_i-\boldsymbol{\mu}^0_g)^{\top}}{\sum_{i\in [N]} z_{i,g}},\\
			&=  \Sigma^0_g.
		\end{align*}
		The fourth equality comes from the fact that $N_g^{-1} \sum_{i\in [N]} z_{i,g} s_{i,(g)} (X_i-\boldsymbol{\mu}^0_g)(X_i-\boldsymbol{\mu}^0_g)^{\top}$ is a consistent estimator of $\Sigma^0_g$ if $z_{i,g} = s_{i,(g)}$ for almost every $i \in [N]$. The last equality comes from the fact that $(X_i-\boldsymbol{\mu}^0_g)(X_i-\boldsymbol{\mu}^0_g)^{\top}$ is a finite quantity for any $i \in [N]$ by assumption \ref{ass:7}, and that $\bar{\alpha}_{gj} = 0$ for any $j \ne (g)$ if $z_{i,g} = s_{i,(g)}$ for almost every $i \in [N]$.
	\end{proof}
	\subsection{Proof of Corollary \ref{cor:11}}
	\begin{proof}
		Following the same logic as in the proof of Theorem \ref{th:2}, we can write that
		\begin{align*}
			d^{2TM}_g(\mathbf{z},\boldsymbol{\mu}_g,\Sigma_g) 
			&=  \text{trace}(\sum_{i\in [N]}  z_{i,g}(X_i-\boldsymbol{\mu}_g)^{\top}\Sigma_g^{-1}(X_i-\boldsymbol{\mu}_g))+ \sum_{i\in [N]}  z_{i,g}\log \det(\Sigma_g),\\
			& = \text{trace}( \sum_{j \in \mathcal{S}} \sum_{i\in [N]} s_{i,j} z_{i,g}(X_i-\boldsymbol{\mu}_g)^{\top}\Sigma_g^{-1}(X_i-\boldsymbol{\mu}_g))+  \sum_{j \in \mathcal{S}} \sum_{i\in [N]}  s_{i,j} z_{i,g}\log \det(\Sigma_g),\\
			&= \sum_{j \in \mathcal{S}}  \text{trace}( \Sigma_g^{-1} C_{gj})+ \sum_{j \in \mathcal{S}}N_{gj} \log \det(\Sigma_g) ,
		\end{align*}
		where $C_{gj} = \sum_{i\in [N]} z_{i,g} s_{i,j} (X_i-\boldsymbol{\mu}_g)(X_i-\boldsymbol{\mu}_g)^{\top}$, and where $\Sigma_g \in \Omega$ for every $g \in \mathcal{S}$. By deriving $C_{gj}$ with respect to $\boldsymbol{\mu}_g$, we have that
		\begin{align*}
			d^{2TM}_g(\mathbf{z},\boldsymbol{\mu}_g,\Sigma_g) \ge  \sum_{j \in \mathcal{S}} \text{trace}( \Sigma_g^{-1} \bar{C}_{gj})+ \sum_{j \in \mathcal{S}}N_{gj} \log \det(\Sigma_g)
		\end{align*}
		for any $\mathbf{z} \in \mathcal{Z}$ and any $\Sigma_g \in \Omega$, where $\bar{C}_{gj} = \sum_{i\in [N]} z_{i,g} s_{i,j} (X_i-\bar{\boldsymbol{\mu}}_{gj})(X_i-\bar{\boldsymbol{\mu}}_{gj})^{\top}$ with $\bar{\boldsymbol{\mu}}_{gj} =  \frac{\sum_{i\in [N]} z_{i,g} s_{i,j} X_{i}}{N_{gj}}$ and $N_{gj} = \sum_{i\in [N]} z_{i,g} s_{i,j}$. From Appendix \ref{lem:S4}, we also have that
		\begin{align*}
			d^{2TM}_g(\mathbf{z},\boldsymbol{\mu}_g,\Sigma_g) &\ge  \sum_{j \in \mathcal{S}} \text{trace}( \{N_{gj}^{-1}\bar{C}_{gj}\}^{-1} \bar{C}_{gj} )+ \sum_{j \in \mathcal{S}} N_{gj} \log \det(N_{gj}^{-1}\bar{C}_{gj}),\\
			&= N_gp + \sum_{j \in \mathcal{S}} N_{gj} \log \det(\bar{\Sigma}_{gj}),
		\end{align*}
		for any $\mathbf{z} \in \mathcal{Z}$ such that $\bar{\Sigma}_{gj} \in \Omega$ for every $(g,j) \in \mathcal{S}^2$, with $\bar{\Sigma}_{gj} = N_{gj}^{-1}\bar{C}_{gj}$. However, the minimizers $(\bar{\boldsymbol{\mu}}_{g1},...,\bar{\boldsymbol{\mu}}_{gS})$ and $(\bar{\Sigma}_{g1},...,\bar{\Sigma}_{gS})$ cannot be obtained by minimizing $d^{2TM}_g(\mathbf{z},\boldsymbol{\mu}_g,\Sigma_g)$ with respect to $\boldsymbol{\mu}_g$ and $\Sigma_g$ since this implies $\bar{S}$ different mean vectors and covariance matrices within the $g^{th}$ estimated group. Moreover, those minimizers are built arbitrarily from partitioning each working group into $\bar{S}$ subgroups, which might lead to singular $\bar{\Sigma}_{gj}$ for a given pair $(g,j) \in \mathcal{S}^2$.
		\par
		Notwithstanding that, let's assume that $N \to \infty$ implies that $N_{gj} \to \infty$ for all pairs $(g,j) \in \mathcal{S}^2$, and then sum over the $\bar{S}$ estimated groups and divide by $N$. In this case, we get that
		\begin{align*}
			\plim_{N \to \infty} N^{-1}d^{2TM}(\mathbf{z},\boldsymbol{\mu},\Sigma) &\ge p + \sum_{g\in \mathcal{S}} \sum_{j \in \mathcal{S}} \pi_{gj} \log \det(\Sigma^0_j),\\
			&= p +  \sum_{j \in \mathcal{S}} \pi_{.j} \log \det(\Sigma^0_j) ,
		\end{align*}
		given that $\bar{\Sigma}_{gj} \xrightarrow{p} \Sigma_j^0$ as $N_{gj} \to \infty$ by construction under the continuous mapping theorem, and where $\pi_{gj} = \lim_{N \to \infty} \frac{N_{gj}}{N}$ and $\pi_{.j} = \lim_{N \to \infty} \frac{\sum_{i\in [N]} s_{i,j}}{N}$. Note that this ``infeasible" asymptotic lower bound does not depend on the estimated groups anymore, but only on the true groups formed by the variable $s_i$. 
		\par
		From the proof of Theorem \ref{th:2}, we have the following ``feasible" asymptotic lower bound
		\begin{align*}
			\plim_{N \to \infty} N^{-1}d^{2TM}(\mathbf{z},\boldsymbol{\mu},\Sigma) &\ge p  +  \sum_{g \in \mathcal{S}}  \bar{\alpha}_g \log \det( \bar{\Sigma}_g ),
		\end{align*}
		where $\bar{\alpha}_g \lim_{N \to \infty} \frac{N_g}{N}$, and where $\mathbf{z}$ is such that $\bar{\Sigma} \in \Omega^{\bar{S}}$. From the same proof, we also know that $\bar{\Sigma}_g \xrightarrow{p} \Sigma^0_g$ if almost every observation in the sample is correctly classified. In this case, we have that $\pi_{.(g)} = \lim_{N \to \infty} \frac{\sum_{i\in [N]} s_{i,(g)}}{N} = \lim_{N \to \infty} \frac{\sum_{i\in [N]} z_{i,g}}{N}= \bar{\alpha}_g$. Therefore, we have that
		\begin{align*}
			\plim_{N \to \infty} N^{-1}d^{2TM}(\mathbf{z},\boldsymbol{\mu},\Sigma) &\ge p  +  \sum_{g \in \mathcal{S}}  \bar{\alpha}_g \log \det( \bar{\Sigma}_g ),\\
			&\equiv p +  \sum_{g \in \mathcal{S}} \pi_{.g} \log \det(\Sigma^0_g)
		\end{align*}
		and the feasible asymptotic lower bound coincides with the infeasible one. We can also write that
		\begin{align*}
			(\mathbf{z}, {\boldsymbol{\mu}}^0,{\Sigma}^0) = \arg \min_{\boldsymbol{\mu} \in \Theta^{\bar{S}},\Sigma \in \Omega^{\bar{S}}} \plim_{N \to \infty} d^{2TM}(\mathbf{z},\boldsymbol{\mu},\Sigma),
		\end{align*}
		if $\mathbf{z}$ is such that $z_{i} = s_{i}$ for almost every observation in the sample provided a suitable permutation in the group labels. Using the definition of the misclassification rate, we also have that
		\begin{align*}
			\frac{\sum_{i\in [N]} \mathbbm{1}[\hat{z}_{i}({\boldsymbol{\mu}}^0,{\Sigma}^0) \ne s_i]}{N} \xrightarrow[]{p} \text{Pr}[\hat{z}_{i}({\boldsymbol{\mu}}^0,{\Sigma}^0) \ne s_i]
		\end{align*}
		as $N \to \infty$ by the WLLN, with $\text{Pr}[\hat{z}_{i}({\boldsymbol{\mu}}^0,{\Sigma}^0) \ne s_i]$ converging almost surely to zero as $p \to \infty$ by Theorem \ref{th:1}. Combining this result with the last equation completes the proof.
	\end{proof}
	\subsection{Proof of Corollary \ref{cor:1}}
	\begin{proof}
		Assumptions \ref{ass:6} and \ref{ass:9} imply that 
		\begin{align*}
			\hat{\tau}_{d,z_i}(a,x) - \tau_{d,(s_i)}(a,x) = O_p(1/\sqrt{N^0_s}),
		\end{align*}
		for any $(d,a,g,x) \in \mathcal{D} \times [0,1] \times \mathcal{S} \times \mathcal{X}$ if $z_i = (s_i)$ for every $i \in [N]$ as $N \to \infty$. Using the conclusion of Theorem 2 from \cite{langevin_bias-reduced_2026}, the minimization of the TSMD when $p/N \to \infty$ as $N,p \to \infty$ will lead to the uniform consistency of the Mahalanobis distance classifier, which implies that every observation in the sample is correctly classified. Therefore, this corollary is a direct consequence of the other theoretical results presented in this paper under Assumptions \ref{ass:1} and \ref{ass:3}-\ref{ass:9} when every observation is correctly classified. However, there is no guarantee that the global minimum of the TSMD will be located close to the true group memberships $\mathbf{s}$ in finite samples when the number of covariates $p$ is small relative to the sample size $N$. This is why all results obtained from finite samples should be validated appropriately.
	\end{proof}
	\subsection{Lemma S.1}\label{lem:S4}
	\textit{Let the total squared Mahalanobis distance be defined as in Definition \ref{def:1}. Then we have that 
		\begin{align*}
			d^{2TM}(\mathbf{z}^{(k)}, \boldsymbol{\mu}^{(k+1)},\Sigma^{(k)}) \ge d^{2TM}(\mathbf{z}^{(k)}, \boldsymbol{\mu}^{(k+1)},\Sigma^{(k+1)}),
		\end{align*}
		provided that every element in $\Sigma^{(k)} = (\Sigma^{(k)}_1,...,\Sigma^{(k)}_S)$ and $\Sigma^{(k+1)} = (\Sigma^{(k+1)}_1,...,\Sigma^{(k+1)}_S)$ is positive-definite.}
	\begin{proof}
		Using the definition of the squared total distance $d^{2TM}$, we can write
		\begin{align*}
			d_g^{2TM}(\mathbf{z}^{(k)}, \boldsymbol{\mu}_g^{(k+1)},\Sigma_g^{(k)}) &= \sum_{i\in [N]}    z^{(k)}_{i,g} (X_i - \boldsymbol{\mu}^{(k+1)}_g)^{\top} \{\Sigma_g^{(k)}\}^{-1} (X_i - \boldsymbol{\mu}^{(k+1)}_g) + N^{(k)}_g \log \det(\Sigma_g^{(k)}),
		\end{align*}
		Since each term on the RHS of the above equation is a scalar, we can apply the trace operator on the first term and write
		\begin{align*}
			d_g^{2TM}(\mathbf{z}^{(k)}, \boldsymbol{\mu}_g^{(k+1)},\Sigma_g^{(k)}) &=  \text{trace}(\{\Sigma_g^{(k)}\}^{-1} C^{(k)}_g) + N^{(k)}_g \log \det(\Sigma_g^{(k)}),
		\end{align*}
		where $C^{(k)}_g = \sum_{i\in [N]} z^{(k)}_{i,g} (X_i - \boldsymbol{\mu}^{(k+1)}_g)(X_i -\boldsymbol{\mu}^{(k+1)}_g)^{\top}$ and $N^{(k)}_g = \sum_{i\in [N]}  z^{(k)}_{i,g}$. Using the properties of the trace and the determinant operators, we can then write 
		\begin{align*}
			d_g^{2TM}(\mathbf{z}^{(k)}, \boldsymbol{\mu}_g^{(k+1)},\Sigma_g^{(k)}) &= \text{trace}(B^{(k)}_g) + N^{(k)}_g \log \det(\Sigma_g^{(k)}),\\
			&= \text{trace}(B^{(k)}_g) + N^{(k)}_g( \log \det(B^{(k)}_g) -  \log \det(B^{(k)}_g) +  \log \det(\Sigma_g^{(k)})),\\
			&= \text{trace}(B^{(k)}_g) - N^{(k)}_g (\log \det(B^{(k)}_g) - \log \det(C^{(k)}_g)),\\
			&= \sum_{j=1}^p \lambda_j - N^{(k)}_g \sum_{j=1}^p \log(\lambda_j)  + N^{(k)}_g\log \det(C^{(k)}_g).
		\end{align*}
		where $B^{(k)}_g= \{\Sigma_g^{(k)}\}^{-1} C^{(k)}_g$, and where $\lambda_j$ is the $j^{th}$ eigenvalue of the matrix $B^{(k)}_g$. Using standard calculus, it is easy to see that the only optimum of $d_g^{2TM}(\cdot,\cdot,\cdot)$ with respect to $B^{(k)}_g$ is attained if and only if $\lambda_j = N^{(k)}_g$ for each value of $j \in \{1,2,...,p\}$. Therefore, the eigendecomposition of the optimal matrix $B^{(k)*}_g$ yields the following optimal covariance matrix~:
		\begin{align*}
			\Sigma_g^{(k)*} &= C^{(k)}_g\{B^{(k)*}_g\}^{-1},\\
			&= C^{(k)}_g (Q I_p N^{(k)}_g Q^{-1})^{-1} ,\\
			&= (N^{(k)}_g)^{-1} C^{(k)}_g \equiv \Sigma_g^{(k+1)},
		\end{align*}
		where $Q$ is the matrix of eigenvectors of $B^{(k)*}_g$. Note that $\Sigma_g^{(k+1)}$ is the unique global minimizer of $d_g^{2TM}(\cdot,\cdot,\cdot)$ given that $d_g^{2TM}(\mathbf{z}^{(k)}, \boldsymbol{\mu}_g^{(k+1)},\Sigma_g^{(k)})$ is unbounded from above when expressed as a function of $\Sigma_g^{(k)}$. Finally, note that if $\Sigma_g^{(k)}$ is not positive-definite, then the last term in the squared Mahalanobis distance will be equal to $-\infty$ (or be undefined) due to the non-positive eigenvalue(s) of $\Sigma_g^{(k)}$. A similar logic holds for $\Sigma_g^{(k+1)}$.
	\end{proof}
	
	\newpage
	\section{The weighted K-means Algorithm}\label{app:B}
	\begin{algorithm}[h]
		\caption{The weighted K-means algorithm}\label{alg:Kmeans}
		\setstretch{1.1}
		Let $\mathbf{z}^{(0)} \in \mathcal{S}^N$ be an initial set of group memberships for every observation in the sample. The weighted K-means algorithm consecutively iterates the following two steps until $\mathbf{z}^{(k)} = \mathbf{z}^{(k+1)}$.
		\begin{enumerate}
			\setstretch{1.1}
			\item The Estimation Step~: 
			\begin{enumerate}
				\item Update $\boldsymbol{\mu}^{(k)}_g$ and $\Sigma^{(k)}_g$ for each $g \in\mathcal{S}$ using
				\begin{align*}
					\boldsymbol{\mu}^{(k+1)}_g &= \frac{\sum_{i\in [N]}  z^{(k)}_{i,g}  X_i}{\sum_{j=1}^N z^{(k)}_{j,g} }, & \Sigma^{(k+1)}_g = 
					\frac{\sum_{i\in [N]} z^{(k)}_{i,g} (X_i - \boldsymbol{\mu}^{(k+1)}_g)(X_i - \boldsymbol{\mu}^{(k+1)}_g)^{\top}}{\sum_{j=1}^N z^{(k)}_{j,g}},
				\end{align*}
				where $z^{(k)}_{i,g} = \mathbbm{1}[z^{(k)}_{i} = g]$, and where the superscript $(k)$ denotes the $k^{th}$ iteration. 
				\item Use eigenvalue decomposition to enforce the minimum eigenvalue $\lambda_{min} > 0$ in $\Sigma^{(k+1)}_g \ \forall g \in \mathcal{S}$.
			\end{enumerate}
			
			\item The Classification Step~:
			\begin{enumerate}
				\setstretch{1.1}
				\item Compute $d_i^{2M}(\boldsymbol{\mu}^{(k+1)}_g,\Sigma_g^{(k+1)})$ for each pair $(i,g) \in [N] \times \mathcal{S}$.
				\item Obtain the classifier $z^{(k+1)}_{i} = \arg \min_{g\in\mathcal{S}} d_i^{2M}(\boldsymbol{\mu}^{(k+1)}_g, \Sigma_g^{(k+1)})$ for each $i \in [N]$.
			\end{enumerate}
			\setstretch{1.1}
			
		\end{enumerate}
	\end{algorithm}
	\newpage
	\section{Details concerning the construction of the outcome and the set of covariates}\label{app:C}
	\begin{enumerate}
		\item Covariates~: For treated patients, pre-treatment variables refer to variables recorded during the follow-up visit preceding the start date of the corresponding DAA regimen. For untreated patients, pre-treatment variables refer to baseline information (i.e., collected at enrollment), except for observations with spontaneous clearance where pre-treatment variables are collected during the follow-up visit preceding the clearance episode. Most of the information collected during the appropriate follow-up visit was then transformed into dummy variables for inclusion in the weighted K-means algorithm. A non-exhaustive list of the variables included in the weighted K-means algorithm is provided below in Table \ref{tab:var}.
		\item SVR in treated patients~: Whether or not a patient has reached SVR after being prescribed any given DAA was recorded by the participating centers and confirmed via a careful examination of the laboratory test results for all treated patients, similar to cases of spontaneous clearance (see below). Specifically, SVR was assumed to occur in any treated patient with at least one negative test result in the 10 weeks following the end of a given DAA treatment regimen. Negative test results occurring after the delay of 10 weeks were considered either as (late) SVR or as spontaneous clearance, depending on the comments made by the treating physician.
		\item Spontaneous clearance~: Whether or not a patient has reached SVR by spontaneous clearance was recorded by the participating center and confirmed via a careful examination of the laboratory test results for all patients. Specifically, spontaneous clearance was assumed to occur in any patient presenting at least two consecutive negative test results after a positive test result without taking any medication against HCV between the dates of the corresponding tests, regardless of the time spent between the tests.
		\item Reinfection and relapse~: Cases of reinfection were identified by the treating physician and confirmed via a careful examination of the laboratory test results for all patients in the sample. Specifically, reinfection was assumed to occur in any patient who previously reached SVR (either through treatment or through spontaneous clearance) and who presented at least two consecutive positive results at least 6 months after reaching SVR, or if the identified HepC genotype was different than the one identified before reaching SVR. If reinfection occurred within 6 months and with the same HepC genotype as before, then the patient was identified as not having reached SVR due to relapse.
	\end{enumerate}
	\
	\begin{table}[t!]
		\centering
		\caption{Variables Included in the Weighted K-Means Algorithm for the Classification Procedure}
		\label{tab:var}
		
		\begin{threeparttable}
			
			\setlength{\tabcolsep}{10pt}
			\renewcommand{\arraystretch}{1.05}
			
			\begin{tabular}{lp{9cm}}
				\toprule
				Types of Variables & Examples of Variables \\
				\midrule
				
				Socioeconomic Indicators
				& Age, Gender, Weight, Height, Waist, Income brackets,
				Living situation, Education level, Ethnicity, Part-/Full-time worker,
				Quality of life, etc. \\
				
				Previous and Current Diagnosis
				& Former spontaneous clearance,
				HCV genotype, HBV/HAV diagnosis, Diagnosis of cirrhosis,
				hypertension, etc. \\
				
				Previous and Current Medication
				& Previous and current ARTs, Previous DAA treatment regimen,
				Previous PEG/RBV treatment regimen (no concomitant medication). \\
				
				Lifestyle Habits and ADLs
				& Alcohol consumption, IDUs/NDUs, Ever been in prison/psychiatric
				institutions, ADL scores for each component, etc. \\
				
				Laboratory Results
				& CD4/CD8 cell counts, Degree of liver fibrosis, Blood levels of insulin, hemoglobin, etc. \\
				
				\bottomrule
			\end{tabular}
			
			\begin{tablenotes}[flushleft]
				\footnotesize
				\item \textit{Notes:} The table lists the types of variables included in the
				weighted K-means algorithm used for the classification procedure. ADLs
				denote activities of daily living; ARTs denote antiretroviral therapies;
				DAA denotes direct-acting antiviral; PEG/RBV denotes pegylated interferon
				and ribavirin; IDUs and NDUs denote injection drug users and non-drug
				users, respectively.
			\end{tablenotes}
			
		\end{threeparttable}
	\end{table}
	\clearpage
	\newpage
	\section{Cost-Effectiveness Algorithm}\label{app:F}
	
	\begin{algorithm}[h!]
		\caption{Optimal treatment allocation under WTP constraint}
		\label{alg:treatment_choice}
		\begin{algorithmic}[1]
			
			\Require $A_i(d) \in [0,1]$ for all $(i,d)\in[N]\times\mathcal{D}$;
			$\omega \in \mathbb{R}_{\geq 0}$
			\Ensure $d^*_i$ for all $i\in[N]$
			
			\ForAll{$(i,d)\in[N]\times\mathcal{D}$}
			\State Compute $ICER_{i,d}(A_i(d))$ using Eq.~(\ref{eq:ICER})
			\If{$d=0$}
			\State $ICER_{i,d}(A_i(d)) \gets 0$
			\EndIf
			\State $Avail_{i,d} \gets
			\mathbbm{1}[
			ICER_{i,d}(A_i(d))\leq\omega
			]$ \Comment{$\mathbbm{1}[\cdot]$ is the indicator function}
			\EndFor
			
			\ForAll{$i\in[N]$}
			\State $M_i \gets
			\displaystyle\max_{d\in\{0,\ldots,\bar D\}}
			\left\{
			\hat{\tau}_{i,d}(A_i(d))\times Avail_{i,d}
			\right\}$
			\State $d^*_i \gets
			\displaystyle\arg\min_{d\in\mathcal{D}}
			\left\{
			c_d :
			\hat{\tau}_{i,d}(A_i(d))\times Avail_{i,d}=M_i
			\right\}$ \Comment{Tie-breaking rule}
			\EndFor
			
			\State \Return $d^*_i$
			
		\end{algorithmic}
	\end{algorithm}
	\clearpage
	\newpage
	\section{Additional Results}\label{app:E}
	\begin{table}[h!]
		\centering
		\caption{Descriptive Statistics by Estimated Group in the Sampled Population}
		\label{tab:A1}
		
		\begin{threeparttable}
			
			\setlength{\tabcolsep}{8pt}
			\renewcommand{\arraystretch}{1.05}
			
			\begin{tabular}{lccccc}
				\toprule
				& Group 1 & Group 2 & Group 3 & Group 4 & Full Sample \\
				& (1) & (2) & (3) & (4) & (5) \\
				\midrule
				
				Male
				& 0.75 & 0.68 & 0.75 & 0.71 & 0.70\\
				& (0.44) & (0.47) & (0.43) & (0.45) & (0.46) \\[0.5mm]
				
				Age
				& 45.4 & 47.3 & 48.9 & 48.3 & 47.5 \\
				& (10.6) & (9.6) & (9.0) & (9.7) & (9.7) \\[0.5mm]
				
				Liver Stiffness
				& 14.8 & 6.5 & 15.6 & 11.9 & 9.5 \\
				& (20.7) & (2.4) & (15.5) & (8.8) & (8.8) \\[0.5mm]
				
				Income Level
				& 3.20 & 3.06 & 2.86 & 2.97 & 3.04 \\
				& (1.68) & (1.59) & (1.68) & (1.52) & (1.58) \\[0.5mm]
				
				\addlinespace[0.5em]
				
				Adherence$^\dagger$
				& 0.983 & 0.973 & 0.974 & 0.967 & 0.973 \\
				& (0.097) & (0.120) & (0.129) & (0.138) & (0.123) \\[0.5mm]
				
				SVR
				& 0.907 & 0.609 & 0.519 & 0.617 & 0.634 \\
				& (0.292) & (0.488) & (0.503) & (0.487) & (0.482) \\[0.5mm]
				
				\midrule
				
				$N$
				& 118 & 670 & 81 & 363 & 1,232 \\
				
				\% of Sample
				& 9.5 & 54.4 & 6.6 & 29.5 & 100.0 \\
				
				\bottomrule
			\end{tabular}
			
			\begin{tablenotes}[flushleft]
				\footnotesize
				\item \textit{Notes:} Means are reported with standard deviations in parentheses.
				Statistics are based on information collected at treatment for treated patients, or at enrollment for untreated patients (apart cases with spontaneous clearance; see Appendix \ref{app:C}). Male is a binary indicator that is equal to one if the patient is a male. Liver stiffness is measured using FibroScan during the pre-treatment follow-up visit and is expressed in kilopascals (kPa). Monthly income is measured on a categorical scale from 1 to 8, corresponding to monthly income categories ranging from \$0--500 to \$4,001 or more. $^\dagger$Conditional on being treated.
			\end{tablenotes}
		\end{threeparttable}
	\end{table}
	\newpage
	\begin{landscape}
		\begin{figure}
			\centering
			\includegraphics[width=26cm]{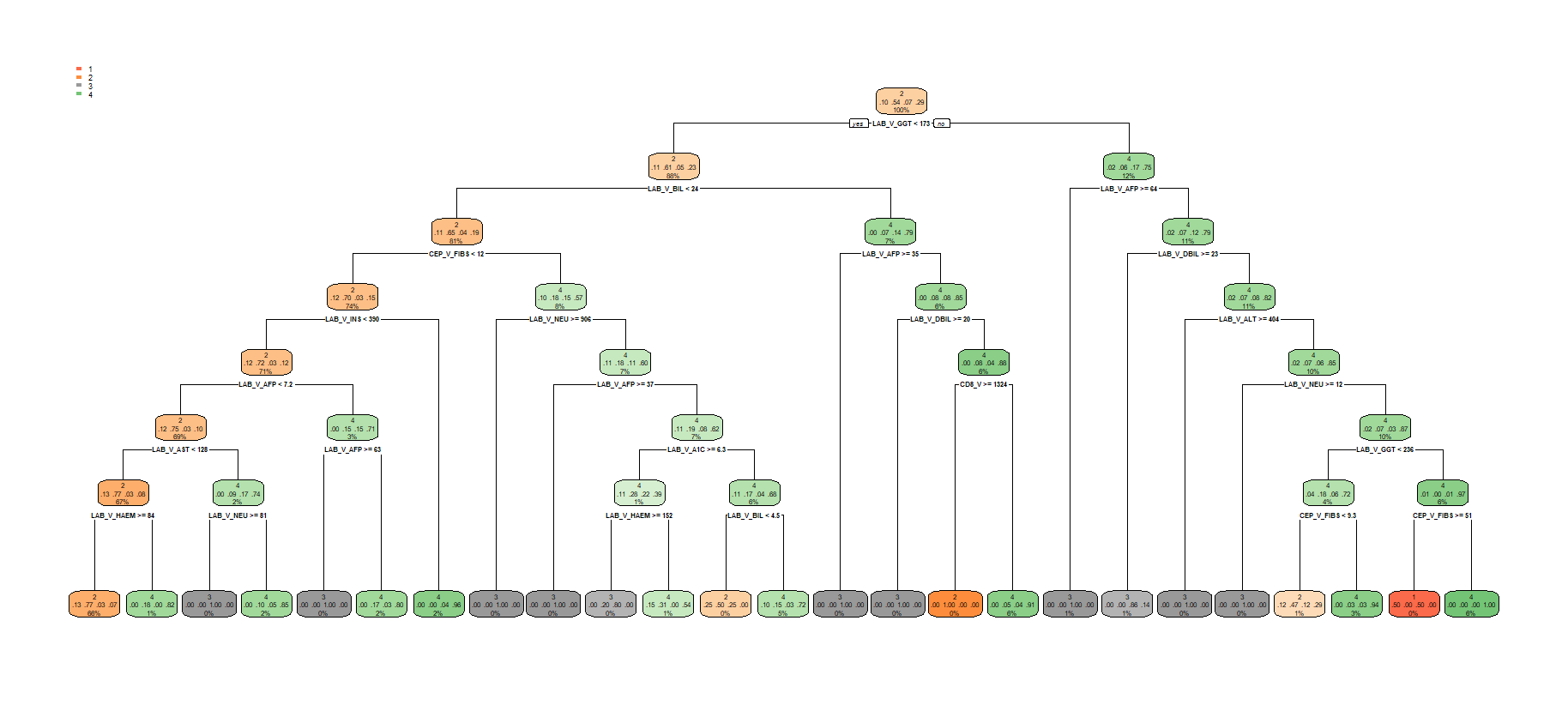}
			\caption{Selected decision tree for the feasible policy rule.}
		\end{figure}
	\end{landscape}
\end{document}